%% file: main_risk_resub.tex
\documentclass[11pt]{article}

 \usepackage{basket_Fourier}

\makeatletter
\def\BState{\State\hskip-\ALG@thistlm}
\makeatother

\title{Optimal Damping with Hierarchical Adaptive Quadrature  for  Efficient Fourier Pricing of  Multi-Asset Options in L\'evy Models}

\author[1]{Christian Bayer}
\author[2]{Chiheb Ben Hammouda}
\author[3]{Antonis Papapantoleon}
\author[4]{Michael Samet\thanks{michael.samet@kaust.edu.sa}}
\author[4,5]{Ra\'ul Tempone}
\affil[1]{Weierstrass Institute for Applied Analysis and Stochastics (WIAS), Berlin, Germany.}
\affil[2]{Mathematical Institute, Utrecht University, Utrecht, The Netherlands}
\affil[3]{Delft Institute of Applied Mathematics, EEMCS, TU Delft, 2628 Delft, The Netherlands, and Department of Mathematics, SAMPS, NTUA, 15780 Athens, Greece and Institute of Applied and Computational Mathematics, FORTH, 70013 Heraklion, Greece.}
\affil[4]{King Abdullah University of Science and Technology (KAUST), Computer, Electrical and Mathematical Sciences \& Engineering Division (CEMSE), Thuwal, Saudi Arabia.}
\affil[5]{Chair of Mathematics for Uncertainty Quantification, RWTH Aachen University, Aachen, Germany.}

\begin{document}
	\date{}
\maketitle


\begin{abstract}

Efficiently pricing multi-asset options is a challenging problem in quantitative finance. When the characteristic function is available, Fourier-based methods are competitive compared to alternative techniques because the integrand in the frequency space often has a higher regularity than that in the physical space. However, when designing a  numerical quadrature method for most Fourier pricing approaches, two key aspects affecting the numerical complexity should be carefully considered: (i) the choice of  damping parameters that ensure integrability and control the regularity class of the integrand and (ii) the effective treatment of high dimensionality. We propose an efficient numerical method for pricing European multi-asset options based on two complementary ideas to address these challenges. First, we smooth the Fourier integrand via an optimized choice of  damping parameters based on a proposed optimization rule. Second, we employ sparsification and dimension-adaptivity techniques to accelerate the convergence of the quadrature in high dimensions. The extensive numerical study on basket and rainbow options under the multivariate geometric Brownian motion and some L\'evy models demonstrates the advantages of adaptivity and the damping rule on the numerical complexity of quadrature methods. Moreover, for the tested two-asset examples, the proposed approach outperforms the COS method in terms of computational time. Finally, we show significant speed-up compared to the Monte Carlo method for up to six dimensions.

\textbf{Keywords} 
	option pricing,  Fourier methods, damping parameters,  adaptive sparse grid  quadrature,  basket and rainbow options,  multivariate L\'evy models.

\textbf{2010 Mathematics Subject Classification} 65D32, 65T50, 65Y20, 	91B25, 91G20, 91G60
\end{abstract}

\section{Introduction}

\input{Introduction.tex}

\section{Problem Setting and Pricing Framework}\label{sec:Problem Setting and Pricing Framework}
\input{Problem_setting.tex}
\section{Methodology}\label{sec:Methodology of our Approach}
\input{Methodology.tex}

\section{Numerical Experiments and Results}\label{sec: num_exp_results}
\input{Num_exp.tex}
\textbf{Acknowledgments} C. Bayer gratefully acknowledges support from the German Research Foundation (DFG) via the Cluster of Excellence MATH+ (Project AA4-2). This publication is based on work supported by the King Abdullah University of Science and Technology (KAUST) Office of Sponsored Research (OSR) under Award No. OSR-2019-CRG8-4033 and the Alexander von Humboldt Foundation. Antonis Papapantoleon gratefully acknowledges the financial support from the Hellenic Foundation for Research and Innovation (Grant No. HFRI-FM17-2152).

 \textbf{Declarations of Interest} The authors report no conflicts of interest. The authors alone are responsible for the content and writing of the paper.


\bibliographystyle{plain}
\bibliography{bibliography} 
\appendix
\input{appendix}
\end{document}

%% file: introduction.tex
Pricing multi-asset options, such as basket and rainbow options, is an interesting and challenging problem in quantitative finance because prices cannot be analytically computed in most cases;   thus, efficient numerical methods are required.  Moreover,	despite the popularity of the Black--Scholes model, where the  stock dynamics  follow the geometric Brownian motion (GBM),   L\'evy models, such as the variance Gamma (VG) \cite{Madan1990TheVG}  and  normal inverse Gaussian (NIG) models \cite{barndorff1997processes},  have shown a better fit to  empirical market behavior \cite{tankov2003financial, schoutens2003levy} by accounting for market jumps in prices, semi-heavy tails, and high leptokurtosis.
	
Under the no-arbitrage assumption, option prices are given as expectations under an (equivalent) martingale measure and approximated using numerical integration methods. In this context, the prevalent numerical method is the Monte Carlo (MC) method \cite{glasserman2004monte}, which has a  convergence rate insensitive to  the input space dimensionality and payoff regularity, except for multilevel MC methods \cite{bayer2020multilevel}, where Lipschitz continuity is necessary to obtain optimal convergence rates. However, the convergence may be very slow, and one may not exploit the available regularity structure to achieve better convergence rates. An alternative stream of research on multi-asset option pricing based on continuous-time Markov chains  approximation \cite{mijatovic2013continuously,xi2019simultaneous,kirkby2020general} has emerged, and was shown to outperform the MC approach for options in two and three dimensions. However, pricing in more than three dimensions still poses a great numerical challenge for these approaches. Another class of methods relies on deterministic quadrature techniques whose performance highly depends on the input space dimension and integrand regularity. Some studies \cite{bayer2020hierarchical,bayer2021numerical} have combined adaptivity,  sparsification techniques and hierarchical representations (Brownian bridge and Richardson extrapolation) with quadrature methods to treat the high dimensionality effectively. Moreover, financial payoffs usually have  low  regularity; therefore, analytic and numerical smoothing techniques were introduced for better convergence  \cite{bayer2018smoothing,bayer2020hierarchical,ben2020hierarchical,bayer2021numerical}.  All aforementioned improvements were performed in the physical space.  

In this work, we propose a novel approach for pricing European multidimensional	basket and rainbow\footnote{Rainbow options \cite{margrabe1978value} are appealing to investors because they allow  the reduction of  risk exposure to the market at a cheap cost by betting more on individual performance among a group of stocks than the overall performance of the portfolio stocks when considering basket options, for instance  see \cite{guillaume2008making}.}  options  under multivariate GBM and L\'evy models.  Compared to the previously mentioned approaches, we recover  the high regularity of the integrand by mapping the problem from the physical space to the frequency space, when the Fourier transforms of  the payoff and  density are well-defined and known explicitly. Moreover, when designing our method, we  effectively treat two key aspects affecting the numerical complexity: (i) the choice of the damping parameters that ensure integrability and  control the regularity class of the integrand  and (ii) the high dimensionality of the integration problem. Based on the extension of the one-dimensional (1D) Fourier  valuation formula  \cite{raible2000levy,lewis2001simple} to the multivariate case,  first, we smooth the Fourier integrand via an optimal choice of the damping parameters based on a proposed  optimization rule. Second, we use adaptive sparse grid  quadrature (ASGQ) based on sparsification and dimension-adaptivity techniques, to accelerate the numerical quadrature convergence in high dimensions.

Fourier-based pricing methods   \cite{carr1999option,raible2000levy,lewis2001simple,duffie2003affine,fang2008novel,lord2008fast,leentvaar2008multi,kwok2012efficient,baschetti2021sinc} map the original  problem to the frequency  space and  obtain  the solution in the physical space using the Fourier inversion theorem. The approximation of the resulting integral is performed numerically using direct integration (DI) methods or the fast Fourier transform (FFT). The common ingredient for these approaches is the explicit knowledge of   the characteristic function (i.e., the Fourier transform of the probability density function) corresponding to the price dynamics.  There are mainly three popular Fourier   valuation approaches. In the first approach, originally proposed by Carr and Madan, see  \cite{carr1999option,Lee04,CarrWu04}, a Fourier transform is applied in the log-strike variable, $k$. Hence, for fixed maturity $T$, the whole curve of option prices, $C(T,\cdot)$, is computed. To ensure the existence of the Fourier transform, one must multiply the pricing function by a damping factor with respect to (w.r.t.)~the strike  parameter. This method is appropriate for 1D problems, however, extending it to the multi-asset option pricing context is difficult. The strike price is not defined for all stocks, whereas the multivariate density depends on all the underlyings. Moreover, the derivations  must be performed  separately for each payoff and stock dynamics.  The second approach \cite{fang2008novel,ruijter2012two,zhang2013efficient}, named as the COS method,  relies on the Fourier cosine series expansion of the density function, and relates the cosine coefficients to the characteristic function. Although the COS method has shown to be  efficient at handling 1D and 2D problems, it is still challenging to generalize this class of methods to the  multidimensional setting for multiple reasons. First, in general, the cosine series coefficients of the payoff function are not known analytically, and hence they need to be recovered numerically by evaluating  high dimensional integrals using    Clenchaw-Curtis quadrature or discrete cosine transform, as suggested in \cite{ruijter2012two}. 
	Second, even though this approach does not introduce damping parameters to ensure integrability,  truncation parameters of the integration domain must be determined.  In  \cite{junike2022precise}, authors showed that there exist cases where the method fails to converge to the correct price if these  parameters are chosen based on the cumulants rule of thumb suggested by the authors in \cite{fang2008novel,ruijter2012two}. To circumvent this issue,  they propose an alternative truncation heuristic for 1D cases but a practical choice in a high-dimensional setting remains a challenging open problem. To avoid determining a-priori truncation range but with a higher cost, the authors of \cite{colldeforns2017two} replaced the Fourier cosine expansion by expressing the density as a finite combination of Shannon wavelet scaling functions, allowing for adaptive estimation of the truncation range. Finally, the number of Fourier cosine series coefficients required for the density expansion grows exponentially with the number of underlyling assets, as pointed out in \cite{chau2019exploration}. Given the characteristic function and Fourier transform of the payoff function,  an alternative third approach \cite{raible2000levy,lewis2001simple,hurd2010fourier,eberlein2010analysis}   uses a highly modular pricing framework. This method separates the underlying stochastic process from the derivative payoff using the Plancherel-Parseval Theorem and uses the generalized inverse Fourier  transform to recover the option price.  In addition, this approach introduces damping parameters w.r.t. the stock prices variables  to ensure integrability, which shifts the integration contour to a parallel line to the real axis in the complex space to avoid singularities. This technique is  easier to extend to the multivariate case compared to the  other  two approaches in \cite{carr1999option,Lee04,CarrWu04} and \cite{fang2008novel,ruijter2012two,zhang2013efficient}.   To the best of our knowledge, when using the  Parseval-based Fourier valuation approach (as  in \cite{raible2000levy,lewis2001simple,hurd2010fourier,eberlein2010analysis}), there is no precise analysis of the effect of the damping parameters on the  convergence speed of quadrature methods or guidance on choosing them to improve the numerical performance, particularly in the multivariate setting. Previous works  have set arbitrary choices for the damping parameter, and only \cite{lord2007optimal,kahl2010fourier} studied the damping parameter selection for the first type Fourier valuation approach (as in  \cite{carr1999option,Lee04,CarrWu04})  in the 1D setting to obtain  robust  integration behavior.  In this work, when pricing basket and rainbow options under the multivariate GBM and L\'evy models, and  based on the extension of the one-dimensional Fourier valuation formula \cite{raible2000levy,lewis2001simple}  to the multivariate case, we demonstrate that the choice of the damping parameters highly affects the speed of convergence  of the numerical quadrature. In addition, motivated by error estimates based on contour integration tools, we propose a general  framework for the  optimal choice of the damping parameters, which can be tailored and extended to various pricing models, resulting in a smoother integrand and improving the efficiency of the numerical quadrature. Based on this proposed  rule, the vector of the optimal damping parameters can be obtained by solving a simple optimization problem. Moreover, we demonstrate the consistent advantage of the optimal damping rule through numerical examples with different dimensions and  parameter constellations.

The numerical evaluation of the resulting inverse Fourier integral can be performed using the FFT algorithm  \cite{carr1999option, crisostomo2018speed}, which    could be faster than DI methods because it exploits periodicities and symmetries. However, it cannot satisfy the requirement for matching the pricing algorithm to the structure of the market data and must be assisted by  interpolation
and extrapolation methods for the smile surface, in contrast to DI methods, which allow for flexible strikes (refer to Chapter 4 of \cite{zhu2009applications} and \cite{lord2007optimal} for further comparisons of FFT and DI). An additional downside of the FFT method is that it has an additional truncation error and requires  the determination of the upper and lower truncation parameters of the integral. This task is nontrivial for multidimensional integrals because the speed of decay to zero of the integrand depends on the damping factors, which are unknown a priori, creating  dependence between the truncation and  damping parameters. In this work, we opt for the DI approach combined with an unbounded quadrature (Gaussian quadrature rule) to evaluate the option price. This approach can be efficiently vectorized,  allowing for a faster calibration procedure. Investigating the optimal choice of the damping and truncation parameters for FFT when pricing multi-asset options remains open for future work.

Through an extensive numerical study on basket and rainbow options under the multivariate GBM,  VG, and NIG models, we demonstrate the advantages of the dimension-adaptive quadrature  and our rule for choosing the damping parameters on the numerical complexity of the  quadrature methods for approximating the Fourier valuation integrals. Moreover, we illustrate cases  where our approach outperforms the COS method. Finally, we show that our approach achieves substantial computational gains over the MC  method for different dimensions and parameter constellations.

Section \ref{sec:Problem Setting and Pricing Framework}  introduces the proposed pricing framework in the Fourier space and  the multivariate valuation formula. In Section \ref{sec:Methodology of our Approach}, we explain our methodology. In Section \ref{damping_section}, we motivate and state our heuristic rule for  choosing of the damping parameters. Moreover, we present   the different hierarchical deterministic quadrature methods used for numerically evaluating  the inverse Fourier integrals of interest  in Section \ref{section_det_quad}. Finally, in Section \ref{sec: num_exp_results}, we report and analyze the obtained results, illustrating the advantages of the proposed approach and highlighting the considerable computational gains achieved over the COS and MC methods.

%% file: Problem_setting.tex
Section \ref{sec:Fourier Pricing Valuation Formula for Multi-Asset Options} introduces the general Fourier valuation framework for multi-asset options that we  consider in this work. Then, in Section \ref{sec:Payoffs and Multivariate Models of the Study}, we present specific details on the type of options  and models investigated in this study.
 
 \subsection{Multivariate Fourier Pricing Valuation Formula}\label{sec:Fourier Pricing Valuation Formula for Multi-Asset Options}
We aim to efficiently price European multi-asset options (e.g., basket/rainbow option, \dots) where the  assets dynamics follow   a certain multivariate stochastic  model (e.g., L\'evy model, \dots).    We extend the 1D representation  \cite{lewis2001simple}  to derive the pricing valuation formula in the Fourier space for the multivariate setting. Before  stating the general valuation formula in  proposition \ref{prop:Multivariate Fourier pricing valuation formula} and its proof, we introduce some needed notations, definitions and assumptions.
 	\begin{notation}[Notations and definitions]\label{notation}\

 		\begin{itemize}
 			\item   $\boldsymbol{X}_t:=\left(X_t^1,\dots, X_t^d\right)$ is a $d$-dimensional   vector  of  log-asset prices\footnote{$X^i_t := \log(S^i_t), i = 1, \ldots, d$,  $\{S^i_t \}_{i=1}^d$ are the  prices of the  underlying assets at time $t$.} at time $t$, $0\leq t \leq T$, where $T$ is the maturity time of the option.  The dynamics of $\boldsymbol{X}_t$ follow a  multivariate stochastic model with   parameters denoted by the vector $\boldsymbol{\Theta}_m$.

 			    $\rho_{\boldsymbol{X}_t}(\cdot)$ is the corresponding    risk-neutral conditional transition probability density function.
 			\item For $\mathbf{z} \in \mathbb{C}^d$, $\Phi_{\boldsymbol{X}_T}(\mathbf{z}):= \mathbb{E}_{ \rho_{\boldsymbol{X}_T}}[e^{\mathrm{i} \langle \mathbf{z}, \boldsymbol{X}_T \rangle}]$ denotes the  joint characteristic function of  $\boldsymbol{X}_T$ extended to the complex plane,  $\langle.,.\rangle$ denotes the inner product in $\mathbb{R}^d$ extended to $\mathbb{C}^d$ i.e., for $\mathbf{y}, \mathbf{z} \in \mathbb{C}^d, \; \langle \mathbf{y}, \mathbf{z} \rangle = \sum_{j = 1}^d y_j z_j$.
 			
 			\item For $\mathbf{x} \in \mathbb{R}^d$, $P(\mathbf{x} )$ denotes  the payoff function.  For $\mathbf{z} \in \mathbb{C}^d$,     $\hat{P}(\mathbf{z} )  := \int_{\mathbb{R}^d} e^{-\mathrm{i} \langle \mathbf{z}, \mathbf{x}\rangle} P(\mathbf{x}) d\mathbf{x} $ is  the Fourier transform of   $P(\cdot)$. $P_{\mathbf{R}}(\mathbf{x}) := e^{\langle \mathbf{R}, \mathbf{x} \rangle}P(\mathbf{x}) $ is the  dampened payoff function.
 			 
 			\item  $\boldsymbol{\Theta}_p = (K ,T,r )$ is the vector of payoff  and market parameters, with $K$ being the strike price, and $r$   the deterministic interest rate.
 			\item We denote  by $\mathrm{i}$  the  unit imaginary number,  by $\Re[\cdot]$ and  $ \Im[\cdot]$  the real  and imaginary part of a complex number, respectively.

 			\item $L_{bc}^1(\mathbb{R}^d)$ is the space of bounded and continuous functions in $L^1(\mathbb{R}^d)$.
 			
 			\item $\boldsymbol{I}_d$ denotes the $d \times d$ identity matrix.
 
 		\end{itemize}
 		\end{notation}

 	\begin{assumption}[Assumptions on the payoff]\label{ass:Assumptions on  the payoff}\
 		 	\begin{enumerate}
 			\item The payoff function, $\mathbf{x} \mapsto P(\mathbf{x})$, is continuous $\forall \; \mathbf{x} \in \mathbb{R}^d$. 
 			\item There exists $\mathbf{R} \in \delta_P :=  \{ \mathbf{R} \in \mathbb{R}^d \; |  \; \mathbf{x} \mapsto P_{\mathbf{R}}(\mathbf{x}) \in L_{bc}^1(\mathbb{R}^d), \; \mathbf{u} \mapsto \hat{P}(\mathbf{u} + \mathrm{i} \mathbf{R}) \in L^1(\mathbb{R}^d) \}$, with $\delta_P $  is the strip of regularity (analyticity) of the payoff's Fourier transform.
 		\end{enumerate}
 		\end{assumption} 
 	
	\begin{assumption}[Assumption on  the model and the corresponding  characteristic function]\label{ass:Assumptions on  the distribution}\
	\begin{enumerate}	
				\item   There exists $\mathbf{R} \in \delta_X : = \{ \mathbf{R} \in \mathbb{R}^d \; |  \; \mathbf{u} \mapsto \Phi_{\boldsymbol{X}_T}(\mathbf{u} + \mathrm{i} \mathbf{R}) \; \text{exists and } | \Phi_{\boldsymbol{X}_T}(\mathbf{u} + \mathrm{i} \mathbf{R})| < \infty, \forall \;  \mathbf{u} \in \mathbb{R}^d\}$, with $\delta_X$ is the strip of regularity (analyticity) of the extended characteristic function.  
	\end{enumerate} 	
\end{assumption}

\begin{proposition}[Multivariate Fourier pricing valuation formula]\label{prop:Multivariate Fourier pricing valuation formula} We use Notation \ref{notation} and 
 		suppose  Assumptions \ref{ass:Assumptions on the payoff} and \ref{ass:Assumptions on  the distribution} hold, and   that $  \delta_V = \delta_X \cap \delta_P \neq \emptyset $, then, for $\mathbf{R} \in \delta_V$,  the option value is given by
 		\begin{equation}
 			V\left(\boldsymbol{\Theta}_m, \boldsymbol{\Theta}_p\right)   = (2 \pi)^{-d} e^{-r T} \Re\left[\int_{\mathbb{R}^d} \Phi_{\boldsymbol{X}_T}(\mathbf{u+\mathrm{i} \mathbf{R}}) \widehat{P}(\mathbf{u+\mathrm{i} \mathbf{R}}) d \mathbf{u}\right].
 			\label{QOI}
 		\end{equation}
 		\end{proposition}
 	\begin{proof}
 Given Assumption  \ref{ass:Assumptions on the payoff} and using the Fourier inversion theorem (see Chapter 7 in \cite{hormander2015analysis})
 		\begin{equation}
 		\label{eq:inverse generalized Fourier transform}
 		P_{\mathbf{R}}(\mathbf{x}) =  (2 \pi)^{-d} \int_{\mathbb{R}^d } e^{  \mathrm{i} \langle \mathbf{u}, \mathbf{x}\rangle} \widehat{P}_{\mathbf{R}}(\mathbf{u}) d\mathbf{u}, \quad \mathbf{R} \in \delta_{P},\: \mathbf{x} \in \mathbb{R}^d.
 		\end{equation}
 		Morerover, we have that
 		\begin{equation}
 			\label{eq:generalized_fourier_transform}
 			\widehat{P}_{\mathbf{R}}(\mathbf{u}) = \int_{\mathbb{R}^d} e^{-\mathrm{i} \langle \mathbf{u}, \mathbf{x}\rangle} e^{\langle \mathbf{R}, \mathbf{x} \rangle} P(\mathbf{x}) d\mathbf{x} = \int_{\mathbb{R}^d} e^{-\mathrm{i} \langle \mathbf{u} + \mathrm{i} \mathbf{R}, \mathbf{x}\rangle} P(\mathbf{x})  d\mathbf{x} = \widehat{P}(\mathbf{u} + \mathrm{i} \mathbf{R}), \quad \mathbf{u} \in \mathbb{R}^d, \mathbf{R} \in \delta_P,
 		\end{equation}
 		where $\widehat{P}(\mathbf{u} + \mathrm{i} \mathbf{R})$ is sometimes called the generalized Fourier transform \cite{titchmarsh1948introduction} or the Fourier-Laplace transform \cite{hormander2015analysis}, a holomorphic extension of the Fourier transform to horizontal strips (tubes) $\mathbf{z} \in \mathbb{R}^d + \mathrm{i} \delta_P \subset \mathbb{C}^d$, in the complex domain \cite{carlsson2017note}. Using  \eqref{eq:generalized_fourier_transform},  Equation \eqref{eq:inverse generalized Fourier transform} can be written as 
 		\begin{equation}
 			\label{eq:payoff_inverse_transform}
 			P(\mathbf{x}) = \Re \left[ (2 \pi)^{-d} e^{-\langle \mathbf{R}, \mathbf{x} \rangle} \int_{\mathbb{R}^d } e^{  \mathrm{i} \langle \mathbf{u}, \mathbf{x}\rangle} \widehat{P}(\mathbf{u} + \mathrm{i} \mathbf{R}) d\mathbf{u} \right], \quad \mathbf{x} \in \mathbb{R}^d, \mathbf{R} \in \delta_P.
 		\end{equation}
 		Then, using  \eqref{eq:payoff_inverse_transform}, Assumption \ref{ass:Assumptions on  the distribution} and Fubini theorem, we obtain that 
 		\begin{equation*}
 			\begin{aligned}
 				V\left(\boldsymbol{\Theta}_m, \boldsymbol{\Theta}_p\right)  &= e^{-r T} \mathbb{E}_{ \rho_{\boldsymbol{X}_T}}\left[P\left(\boldsymbol{X}_T\right)\right] 
 				 \\
 				&=(2 \pi)^{-d} e^{-r T} \mathbb{E}_{ \rho_{\boldsymbol{X}_T}}\left[\Re\left[e^{-\langle \mathbf{R}, \boldsymbol{X}_T \rangle} \int_{\mathbb{R}^d} e^{ \mathrm{i}\left\langle\mathbf{u}, \boldsymbol{X}_T\right\rangle} \widehat{P}(\mathbf{u} + \mathrm{i} \mathbf{R}) d \mathbf{u}\right]\right], \mathbf{R} \in \delta_P \\
 				&=(2 \pi)^{-d} e^{-r T} \Re\left[\int_{\mathbb{R}^d}\mathbb{E}_{ \rho_{\boldsymbol{X}_T}}\left[ e^{\mathrm{i} \langle \mathbf{u} + \mathrm{i} \mathbf{R}, \boldsymbol{X}_T \rangle } \right] \widehat{P}(\mathbf{u} + \mathrm{i} \mathbf{R}) d \mathbf{u}\right], \mathbf{R} \in \delta_V:= \delta_P \cap \delta_X \\
 				&= (2 \pi)^{-d} e^{-r T} \Re\left[\int_{\mathbb{R}^d} \Phi_{\boldsymbol{X}_T}(\mathbf{u+\mathrm{i} \mathbf{R}})  \widehat{P}(\mathbf{u+\mathrm{i} \mathbf{R}}) d \mathbf{u}\right].
 			\end{aligned}
 		\end{equation*}
 			The application of Fubini's Theorem is justified by  imposing $\mathbf{R} \in \delta_P$ to enforce $\widehat{P}(\mathbf{u} + \mathrm{i} \mathbf{R}) \in L^1(\mathbb{R}^d)$ and  imposing $\mathbf{R} \in \delta_X$  to ensure that $\Phi_{\boldsymbol{X}_T}(\mathrm{i}  \mathbf{R})$ exists and is bounded.
 \end{proof}
 In what follows, from \eqref{QOI}, we define the integrand of interest by
 	\begin{equation}
 		g\left(\mathbf{u} ; \mathbf{R}, \boldsymbol{\Theta}_m, \boldsymbol{\Theta}_p\right) := (2 \pi)^{-d} e^{-r T} \Re [  \Phi_{\boldsymbol{X}_T}(\mathbf{u}+\mathrm{i} \mathbf{R}) \widehat{P}(\mathbf{u}+\mathrm{i} \mathbf{R}) ], \mathbf{u} \in \mathbb{R}^d, \mathbf{R} \in \delta_V
 		\label{g_integrand}
 	\end{equation}
 \begin{remark}[Connection to the valuation formula in \cite{eberlein2010analysis}]
The notation we used for the definition of the Fourier transform and  the payoff dampening is the same as in \cite{hurd2010fourier}, with different set of assumptions. The valuation formula in Theorem 3.2 of \cite{eberlein2010analysis}  can be easily recovered from \eqref{QOI} by considering $- \mathbf{R}$ instead of $\mathbf{R}$, $ \widehat{P}(- \mathbf{z})$ instead of $\widehat{P}(\mathbf{z})$, and using the relation $\Phi_{\boldsymbol{X}_T}(\mathbf{u}) = M_{\boldsymbol{X}_T}(\mathrm{i} \mathbf{u}) $, where $ M_{\boldsymbol{X}_T}(\cdot)$ denotes the moment generating function of $\boldsymbol{X}_T$. 
 \end{remark}
 
 \begin{remark}[Case of discontinuous payoffs]
 	In general, the regularity assumptions on the payoff, such as its continuity, can be compensated by more regularity assumptions on the model. However, in the particular case of European options and the considered Lévy models, the continuity condition in Assumption \ref{ass:Assumptions on  the payoff} can be dropped because the considered processes possess a Lebesgue density.  We refer to \cite{eberlein2010analysis} for more details. 
 \end{remark}
\subsection{Payoffs and Multivariate Asset Models}\label{sec:Payoffs and Multivariate Models of the Study}
\subsubsection{Payoffs and their Fourier Transforms}
In this work, we focus on two specific examples of payoffs, namely (i) basket put\footnote{ $w_i > 0, \; \forall \; i = 1 \ldots, d$. }
 and (ii)  and call on min,  which are respectively given, for $K > 0$, by:
\begin{equation}\label{payoff_scaled}
	(i)	\; P(\boldsymbol{X}_T)= \max \left(K -  \sum_{i=1}^d  w_i  e^{X_T^i}, 0\right); \quad  (ii)\;	P(\boldsymbol{X}_T) = \max \left( \min\left(e^{X_T^1},\ldots,e^{X_T^d}\right) - K  , 0\right).
\end{equation} 
The  Fourier transforms for both payoffs in \eqref{payoff_scaled}
	 are respectively given by  \eqref{eq:payoff_transforms_basket_put}\footnote{To simplify the presentation, we consider the unweighted basket put payoff. The  generalization to the weighted case as presented in Section \ref{sec: num_exp_results} can be done straightforwardly by considering $X_t^i = \log(\frac{S^i_t}{w_i}),i=1,\ldots,d$.}  for $ w_i = 1, i = 1\ldots d$,  and \eqref{eq:payoff_transforms_call_min}, with regularity strips,  $\delta_{P}$,  expressed in Table \ref{table:fourier_payoff_strip_table}. The considered payoffs satisfy  Assumption \ref{ass:Assumptions on  the payoff}, we refer to  \cite{eberlein2010analysis,hubalek2005variance,hurd2010fourier}  for further details on the derivation.

	\begin{align}
		& \hat{P}(\mathbf{z}) =   K^{1-\mathrm{i} \sum_{j=1}^{d} z_{j}}  \frac{\prod_{j=1}^{d} \Gamma\left(- \mathrm{i} z_{j}\right)}{\Gamma\left(-\mathrm{i} \sum_{j=1}^{d} z_{j}+2\right)}, \quad \mathbf{z} \in \mathbb{C}^d,\: \Im[z_j] > 0 \:  \forall j \in \{1 ,\ldots, d\};\label{eq:payoff_transforms_basket_put}\\
		& \hat{P} (\mathbf{z}) =  \frac{  K^{1- \mathrm{i} \sum_{j=1}^{d} z_{j}}  }{\left(\mathrm{i}\left(\sum_{j=1}^{d} z_{j}\right)-1\right) \prod_{j=1}^{d}\left(\mathrm{i} z_{j}\right)}, \quad \mathbf{z} \in \mathbb{C}^d, \: \Im[z_j]< 0 \:  \forall j \in \{1 ,\ldots, d\},\: \sum_{j=1}^{d} \Im[z_j] < -1 \label{eq:payoff_transforms_call_min},
	\end{align}
where  $ \Gamma(z) = \int_{0}^{+ \infty} e^{-t} t^{z-1} dt$, is the complex Gamma function  defined for $z \in \mathbb{C}$, with $\Re[z] > 0$.

		\begin{table}[h!]
		\centering
		\begin{tabular}{ | c| c| }
			\hline
			\textbf{Payoff} &  \small $ \delta_{P}$\\
			\hline
			Basket put &  \small $\{\mathbf{R} \in \mathbb{R}^d,\: R_i > 0\: \forall i \in \{1 ,\ldots, d\}\}$ \\
			\hline
			Call on min &  \small $\{\mathbf{R} \in \mathbb{R}^d,\:  R_i< 0 \: \forall i \in \{1 ,\ldots, d\},\sum_{i=1}^{d} R_i < - 1\}$ \\
			\hline
		\end{tabular}
		\captionsetup{justification=centering,margin=0.5cm}
		\caption{Strip of regularity, $\delta_{P}$, of payoff transforms.}
		\label{table:fourier_payoff_strip_table}
	\end{table}
	\subsubsection{Multivariate Models and the Corresponding Charactersitic Functions}
For the asset dynamics,  in this work we study the three models given by Examples \ref{MGBM_sec}, \ref{VG_sec} and \ref{NIG_sec}.

\begin{example}	[Multivariate Geometric Brownian Motion (GBM)]\label{MGBM_sec}
		 The joint risk-neutral dynamics of the stock prices are modeled as follows:
	\begin{equation}
		S_i(t) = S_{i}(0) \exp \left[\left(r-\frac{\sigma_{i}^{2}}{2} \right) t+  \sigma_i W_i(t)\right], \quad i=1, \ldots, d,
	\end{equation}
	where $\sigma_1, \dots, \sigma_d> 0 $ and $\{W_1(t),\dots, W_d(t) , t \ge 0\}$ are correlated standard Brownian motions with correlation matrix $\mathbf{C} \in \mathbb{R}^{d \times d}$ with components   
	  $(\boldsymbol{C})_{i,j} = \rho_{i,j}$, with  $-1\le \rho_{i, j} \le 1$  denoting the correlation between  $W_i$ and $W_j$. 
	Moreover,  $\boldsymbol{\Sigma} \in \mathbb{R}^{d \times d}$ denotes  the covariance matrix of the log returns, $\{ \log(\frac{S_i(t)}{S_i(0)}) \}_{i=1}^d$, with $\boldsymbol{\Sigma}_{ij}=\rho_{i, j} \sigma_{i} \sigma_{j}$.
	\end{example}
\begin{example}[Multivariate Variance Gamma (VG) \cite{luciano2006multivariate}]
\label{VG_sec}
The joint risk-neutral dynamics of the stock prices are modeled as follows:
\begin{equation}
	S_{i}(t) = S_{i}(0) \exp \left[\left(r+\mu_{VG,i}\right) t+\theta_{i} G(t)+\sigma_{i} \sqrt{G(t)} W_i(t)\right], \quad  i=1, \ldots, d,
\end{equation}
where    $\{W_1(t),\dots, W_d(t) \}$  are independent standard Brownian motions, $\{ G(t) | t \geq 0 \}$  is a common Gamma process 
 with parameters  $( \frac{t}{\nu}, \frac{1}{\nu})$, and independent of all the Brownian motions. Additionally, $\theta_i \in \rset$ and $\sigma_i> 0$,  $1 \le  i \le d$.
  The covariance matrix  $\mathbf{\Sigma} \in \mathbb{R}^{d \times d}$ satisfies $\Sigma_{i,j} = {\sigma_i }^2$ for $i=j$, and 0 otherwise.  Finally,  $\boldsymbol{\mu}_{VG}:=(\mu_{VG,1}, \dots,\mu_{VG,d})$ are the  martingale correction terms that  ensure that  $\{e^{-rt}S_i(t) | t \geq 0\}$ is a martingale and are given by 
\begin{equation}
	\mu_{VG,i}=\frac{1}{\nu} \log \left(1-\frac{1}{2} \sigma_{i}^{2} \nu-\theta_{i} \nu\right), \quad  i=1, \ldots, d.
\end{equation}
\end{example}
\begin{example}[Multivariate Normal Inverse Gaussian (NIG) \cite{barndorff1997normal,barndorff1977exponentially}]
	\label{NIG_sec}
	 The joint risk-neutral dynamics of the stock prices are modeled as follows:
	\begin{equation}
		S_{i}(t) = S_{i}(0) \exp \left\{\left(r+\mu_{NIG,i}\right) t+\beta_{i} IG(t)+ \sqrt{IG(t)}  W_{i}(t)\right\}, \quad i=1, \ldots, d,
	\end{equation}
	where $\{W_1(t),\dots, W_d(t)\}$  are independent  standard Brownian motions, $\{ IG(t) | t \geq 0 \}$  is a common inverse Gaussian process 
 with parameters  $(\delta^2 t^2, \alpha^{2}-\boldsymbol{\beta}^{\mathrm{T}} \boldsymbol{\Delta} \boldsymbol{\beta})$, and 	independent of all  Brownian motions.  Additionally, $\alpha \in \mathbb{R}_{+}$,  $\boldsymbol{\beta} \in \mathbb{R}^{d}$,  $\alpha^2 >  \boldsymbol{\beta}^{\mathrm{T}} \boldsymbol{\Delta} \boldsymbol{\beta} $, $\delta > 0$,  and 	$ \boldsymbol{\Delta} \in \mathbb{R}^{d \times d}$ is a symmetric  positive definite matrix with a unit determinant.   $\{\mu_{NIG,i}\}_{i=1}^{d}$ are the  martingale correction terms that  ensure  that  $\{e^{-rt}S_i(t) | t \geq 0\}$ is a martingale, given by 
	\begin{equation}
		\mu_{NIG,i}=- \delta \left( \sqrt{\alpha^2 - \beta_i^2} - \sqrt{\alpha^2 - (\beta_i + 1)^2}\right), \quad  i=1, \ldots, d.
	\end{equation} 
	\end{example}
	 For each model in Examples \ref{MGBM_sec}, \ref{VG_sec}, and \ref{NIG_sec}, we provide   in Table \ref{table:chf_table}  the expression of the  characteristic function, with  regularity strips,  $\delta_X$,  expressed in Table \ref{table:strip_table}. The characteristic functions of the studied models fulfill Assumption \ref{ass:Assumptions on  the distribution}. We refer to \cite{eberlein2010analysis}   for further details. 

		\begin{table}[h]
			\centering
			\begin{tabular}{| p{1.2cm} | p{14.8cm} |}
				\hline   \textbf{Model} &  $\Phi_{\boldsymbol{X}_T}(\mathbf{z}) $ \\
				\hline
				\textbf{GBM} &   \small $ \exp \left(\mathrm{i}\left\langle \mathbf{z} ,\boldsymbol{X}_0\right\rangle\right) \times  \exp \left(\mathrm{i} \langle\mathbf{z}, r \mathbf{1}_{\mathbb{R}^d}-\frac{1}{2} \operatorname{diag}(\boldsymbol{\Sigma}) \rangle T -\frac{T}{2} \langle \mathbf{z},\boldsymbol{\Sigma} \mathbf{z} \rangle\right),\quad \mathbf{z} \in \mathbb{C}^d, \: \Im[\mathbf{z}] \in \delta_X$  \\ 
				\hline
				\textbf{VG} & \small $ \exp \left(\mathrm{i}\left\langle \mathbf{z} ,\boldsymbol{X}_0\right\rangle\right) \times \exp (\mathrm{i} \langle  \mathbf{z}, r \mathbf{1}_{\mathbb{R}^d}+\boldsymbol{\mu}_{VG}\rangle T)\left(1-\mathrm{i} \nu\langle \boldsymbol{\theta}, \mathbf{z}\rangle+\frac{1}{2} \nu\langle\mathbf{z}, \boldsymbol{\Sigma} \mathbf{z}\rangle\right)^{-T / \nu},\quad \mathbf{z} \in \mathbb{C}^d, \: \Im[\mathbf{z}] \in \delta_X$\\
				\hline
				\textbf{NIG}  & \small $ \exp \left(\mathrm{i}\left\langle \mathbf{z} ,\boldsymbol{X}_0\right\rangle\right) \times \exp \left(\mathrm{i} \langle \mathbf{z}, r \mathbf{1}_{\mathbb{R}^d}+\boldsymbol{\mu}_{NIG}\rangle T+\delta T\left(\sqrt{\alpha^{2}-\langle\boldsymbol{\beta}, \boldsymbol{\Delta} \boldsymbol{\beta}\rangle}-\sqrt{\alpha^{2}-\langle\boldsymbol{\beta}+ \mathrm{i} \mathbf{z}, \boldsymbol{\Delta}(\boldsymbol{\beta}+ \mathrm{i} \mathbf{z})\rangle}\right)\right),\newline\quad \mathbf{z} \in \mathbb{C}^d, \: \Im[\mathbf{z}] \in \delta_X$    \\
				\hline
			\end{tabular}
			\caption{The expression of  the characteristic function, $\Phi_{\boldsymbol{X}_T}(\cdot)$,  for the different pricing models. $\mathbf{1}_{\mathbb{R}^d} $ is the $d$-dimensional unit vector.}
			\label{table:chf_table}
		\end{table}
	\begin{table}[h]
		\centering
		\begin{tabular}{| p{1.7cm} | p{8cm}  |}
			\hline   \textbf{Model} &  $\delta_X$  \\
			\hline
			\textbf{GBM} &  \small $\mathbb{R}^d $  \\ 
			\hline 
			\textbf{VG} & \small $\{\mathbf{R} \in \mathbb{R}^d, \left(1+ \nu\langle\boldsymbol{\theta}, \mathbf{R}\rangle -\frac{1}{2} \nu\langle\mathbf{R}, \boldsymbol{\Sigma} \mathbf{R}\rangle\right) > 0\}$  \\
			\hline
			\textbf{NIG}   & \small $\{\mathbf{R} \in \mathbb{R}^d, \left(\alpha^2 - \langle (\boldsymbol{\beta} -\mathbf{R}) , \boldsymbol{\Delta}( \boldsymbol{\beta} - \mathbf{R} )\rangle \right)>0\}$  \\
			\hline
		\end{tabular}
		\caption{Strip of analyticity,  $\delta_X$, of the characteristic functions  for the different pricing models.}
		\label{table:strip_table}
	\end{table}
\begin{remark}[About the strip of regularity] Compared to the 1D case, in the multivariate setting, the choice of the vector of  damping parameters  $\mathbf{R}$, which satisfies the analyticity condition  in Table \ref{table:strip_table}, is nontrivial requiring numerical approximations. Moreover, to obtain more intuition,  for the multivariate NIG model with $\boldsymbol{\Delta} = \boldsymbol{I_d}$, the strip of regularity $\delta_X^{NIG}$ is  an open ball centered at $\boldsymbol{\beta}$ with radius $\alpha$. This fact further complicates the arbitrary choice for damping parameters when the integrand is anisotropic because we  must first  identify   the spherical boundary to determine the admissible combinations of values for the damping parameters enforcing the integrability.   

\end{remark} 
\begin{remark}[Efficient vectorized implementation for model calibration]
It is often more convenient to work with the scaled versions of the payoff (e.g $P(\boldsymbol{X}_T)= \max \left(K -   \sum_{i=1}^d 
 	w_i e^{X_T^i}, 0\right) = K \max \left(1 -   \sum_{i=1}^d e^{{X_T^\prime}^i}, 0\right), {X^{\prime}_T}^i = \log(\frac{S_0^i}{w_i K})$), 
 	 so the strike variable, $K$, is  taken out of the integral in \eqref{QOI}. Moreover, the considered models are stochastic processes with independent increments, which allows us to factorize the  multivariate characteristic function in the following way,  $\Phi_{\boldsymbol{X}_T}(\mathbf{z})=\mathrm{e}^{\mathrm{i}\left\langle \mathbf{z} ,\boldsymbol{X}_0\right\rangle}\phi_{\boldsymbol{X}_T}(\mathbf{z})$ for $\mathbf{z} \in \mathbb{C}^d$, such that $\phi_{\boldsymbol{X}_T}(\mathbf{z})$ is independent of $\boldsymbol{X}_0$.

The advantage in eliminating the dependence of the characteristic function and the payoff from the strike-dependent terms is that we can now evaluate them once in the Fourier domain for multiple values of strike, $K$. This
 allows for an efficient vectorized implementation that is particularly useful for practitioners interested in model calibration. 
\end{remark}

%% file: Methodology.tex
\subsection{Motivation and Characterization of the Optimal Damping Rule}
\label{damping_section}
This section aims to motivate  and  propose a  rule for the optimal choice of the damping parameters $\mathbf{R}$ that can  accelerate the convergence of the numerical quadrature in the Fourier space  when approximating  \eqref{QOI} for pricing    multi-asset options under the considered pricing models for various parameters.	 The main idea is to establish a connection between   the damping parameter values, integrand properties,  and    quadrature error.
\subsubsection{Motivation of the Damping Rule}
Before considering  the integral of interest \eqref{QOI}, we provide the general motivation for  the rule   through a  simple 1D integration example for a real-valued
function $f$ w.r.t.~a  weight function $\lambda(\cdot)$  over the  support interval $[a, b]$ (finite, half-infinite, or doubly infinite interval):  
\begin{equation}\label{eq:simple_integration_problem}
	I[f]:=\int_{a}^{b}  f(x) \lambda(x) dx \approx \sum_{k=1}^{N} w_k f(x_k):=Q_N[f],
\end{equation}
where the quadrature estimator, $Q_N[f]$ is characterized by the  nodes $\{x_k\}_{k=1}^N$ which are  the roots of the appropriate orthogonal polynomial, $\pi_k(x)$, and $\{w_k\}_{k=1}^N$ are the appropriate quadrature weights. Moreover, $	\mathcal{E}_{Q_N}[f]$ denotes the quadrature error (remainder),  defined as $\mathcal{E}_{Q_N}[f]:= I[f]-Q_N[f]$.  

The analysis of the quadrature error can be performed through two  representations: the first  relies on estimates based on high-order derivatives for  a smooth function $f$ \cite{gautschi2004orthogonal,davis2007methods,xiang2012asymptotics}. These error representations are  of limited practical value  because  high-order derivatives are usually challenging to estimate and control. 
For this reason,  to derive our rule, we opt for the second form of quadrature error representation, valid for functions that can be extended holomorphically into 	the complex plane, which corresponds to the case in \eqref{g_integrand}. 

Several approaches exist for estimating the error $\mathcal{E}_{Q_N}[f]$ when $f$ is holomorphic: (i) methods of contour integration \cite{takahasi1971estimation,donaldson1972unified}, (ii) methods based  on Hilbert space norm estimates \cite{davis1954estimation,donaldson1973estimates} which consider $	\mathcal{E}_{Q_N}$ as a linear functional on $f$,  and (iii) methods based on approximation theory \cite{babuvska2007stochastic,trefethen2008gauss}. Independent of  the approach, the results are often  comparable because the error bounds involve the supremum norm of $f$.

We focus on error estimates based on contour integration tools to showcase these error bounds. This approach uses Cauchy’s theorem in the theory of complex variables to  express the value of an analytic function at some point $z$ by means of a contour integral (Cauchy integral) extended over a simple closed curve  (or open arc) in the complex plane encircling the point $z$.  
\begin{theorem}\label{thm:remainder_integration_estimates}
Assuming that the function $f$ can be  extended  analytically into a sizable region of the complex plane, containing the interval $[a, b]$ with no singularities. Then,	the error integral in the approximation \eqref{eq:simple_integration_problem} can be expressed as
	\begin{equation}\label{eq:integ_error_contour}
	\mathcal{E}_{Q_N}[f]=\frac{1}{2 \pi \mathrm{i} } \oint_{\mathcal{C}} K_N(z) f(z) dz,
	\end{equation}
	where 
	\begin{equation}\label{eq:def Psi}
		K_N(z) = \frac{H_N(z)}{\pi_N(z)},  \quad H_N(z)=\int_{a}^{b} \lambda(x) \frac{\pi_N(z)}{z-x} dx,
	\end{equation}
and  $\mathcal{C}$ is a contour\footnote{Two choices of $\mathcal{C}$  are most frequently made: $\mathcal{C}=\mathcal{C}_r$, the circle $|z|= r, r > 1$,  and $C = \mathcal{C}_\rho$, the ellipse with foci at $a$ and $b$, where the sum of its semiaxes is equal to $\rho, \rho > 1$.  Circles  can only be used if the analyticity domain is sufficiently large, and ellipses  have the advantage of shrinking to the interval
	$[ a, b]$ when $\rho \rightarrow 1$, making them suitable for dealing with functions that  are analytic on the segment $[a, b]$.} containing the  interval $[a, b]$  within which $f(z)$ has no singularities.
\end{theorem}
\begin{proof}
We refer to  \cite{donaldson1972unified,gautschi2004orthogonal} 	for a proof of Theorem \ref{thm:remainder_integration_estimates}.
	\end{proof}
 In the finite case, the contour $\mathcal{C}$ is closed and \eqref{eq:def Psi} represents an analytic function in the connected domain $\mathbb{C}\setminus [a, b]$ while  we may take $\mathcal{C}$ to lie on the upper and lower edges of the real axis in the infinite case for large $|x|$. Discussions on choosing adequate contours are found in  \cite{elliott1970uniform,donaldson1972unified,donaldson1973estimates}.  Moreover, precise estimates of $H_N(z)$ were derived in  \cite{donaldson1972unified,elliott1974asymptotic}.
 
As $f(\cdot)$  has no singularities within  $\mathcal{C}$,   using Theorem \ref{thm:remainder_integration_estimates}, we obtain  
\begin{equation}\label{eq:error_bound_estimate}
	|	\mathcal{E}_{Q_N}[f]| \le   \frac{1}{2 \pi }  \;  \underset{z \in \mathcal{C}}{\sup}  | f(z)| \oint_{\mathcal{C}} |K_N(z)| |dz|,
\end{equation}
where  the quantity $\oint_{\mathcal{C}} |K_N(z) | |dz|$ depends only on the quadrature rule.
 We expect that when the size of the contour  increases, $\oint_{\mathcal{C}} |K_N(z)| |dz|$  decreases,  whereas $ \underset{z \in \mathcal{C}}{\sup}  | f(z)| $ increases  by the maximum modulus theorem. The optimal choice of the contour  $\mathcal{C}$  is the one that minimizes the right-handside of \eqref{eq:error_bound_estimate}.

Extending the error bound  \eqref{eq:error_bound_estimate} to the multidimensional setting can be performed straightforwardly  in a recursive way and using   tensorization tools (we refer to Appendix \ref{quad_error_bound} for an illustration).  Moreover,  the term   $ \underset{z \in \mathcal{C}}{\sup}  | f(z)| $ in the upper bound    \eqref{eq:error_bound_estimate}   is independent of the  quadrature method.

\subsection{Characterization of the Optimal Damping Rule}
 Motivated by the error bound \eqref{eq:error_bound_estimate},  we propose a  rule for choosing the damping parameters that improves the numerical  convergence of the designed  numerical quadrature method (see Section \ref{section_det_quad})   when approximating  \eqref{QOI}. Using Notation \ref{notation}, the rule consists in solving the following constrained optimization problem
\begin{equation}\label{or_opt}
	\mathbf{R}^\ast:=	\mathbf{R}^\ast (\boldsymbol{\Theta}_{m},\boldsymbol{\Theta}_{p})= \underset{\mathbf{R}\in \delta_V } {\arg \min} \sup_{\mathbf{u} \in \mathbb{R}^d} \mid g(\mathbf{u};\mathbf{R},\boldsymbol{\Theta}_{m},\boldsymbol{\Theta}_{p}) \mid 
\end{equation} 
where $g(\cdot)$ is defined in \eqref{g_integrand} and  $\mathbf{R}^\ast:=(R_1^\ast,\ldots,R_d^\ast)$ denotes the vector of optimal damping parameters.

Based on Proposition \ref{damping_rule_proposition},  we can  reduce \eqref{or_opt}  to a simpler optimization problem, which consists of finding the vector of damping parameters, $\mathbf{R} \in \delta_V$, that minimize the peak of the integrand in \eqref{g_integrand} at the origin point $\mathbf{u} = \mathbf{0}_{\mathbb{R}^d}$. 

 	\begin{proposition}\label{damping_rule_proposition}
 		For $g$ defined by \eqref{g_integrand} and   $ \mathbf{R} \in  \delta_V$, we have 
\begin{equation}\label{new_opt}
	\mathbf{R}^\ast =\underset{\mathbf{R}\in \delta_V } {\arg \min} \sup_{\mathbf{u} \in \mathbb{R}^d} \mid g(\mathbf{u};\mathbf{R},\boldsymbol{\Theta}_{m},\boldsymbol{\Theta}_{p}) \mid = \underset{\mathbf{R}\in \delta_V } {\arg \min}\;g(\mathbf{0}_{\mathbb{R}^d};\mathbf{R},\boldsymbol{\Theta}_{m},\boldsymbol{\Theta}_{p}). \\
\end{equation} 
\end{proposition}
\begin{proof}
	Let $f : \mathbb{R}^d \mapsto \mathbb{R}_+$ be an arbitrary real-valued non-negative function, and  let $\mathbf{R} \in \mathbb{R}^d$ such that the dampened function $\mathbf{x} \mapsto f_{\mathbf{R}}(\mathbf{x}) \in L^1(\mathbb{R}^d)$, then we have that its Fourier transform satisfies
	
	\begin{equation}\label{eq:ridge property}
	\mid \widehat{f}_{\mathbf{R}}(\mathbf{u}) \mid \leq \int_{\mathbb{R}^d}  \mid e^{\mathrm{i}  \langle\mathbf{u} , \mathbf{x} \rangle} \mid  e^{ \langle\mathbf{R}, \mathbf{x} \rangle} f(\mathbf{x}) d\mathbf{x} = \int_{\mathbb{R}^d}   e^{ \langle\mathbf{R}, \mathbf{x} \rangle}  f(\mathbf{x}) d\mathbf{x}  =  \widehat{f}_{\mathbf{R}}(\mathbf{0}), \mathbf{u} \in \mathbb{R}^d.
\end{equation}
Moreover, we have that by \eqref{eq:generalized_fourier_transform} $\widehat{f}_{\mathbf{R}}(\mathbf{u}) = \widehat{f}( \mathbf{u} + \mathrm{i}\mathbf{R}) $, hence \eqref{eq:ridge property} implies that 
\begin{equation}
	\label{eq:ridge_prop}
	\mid \widehat{f}(\mathbf{u} + \mathrm{i} \mathbf{R}) \mid \leq \widehat{f}(\mathrm{i} \mathbf{R}), \quad \forall \; \mathbf{u} \; \in \mathbb{R}^d.
\end{equation}
Equation \eqref{eq:ridge_prop} is known as the ridge property of Fourier transforms (see \cite{lukacs1970characteristic}). Since both the payoff functions and the probability density functions are real-valued and non-negative then their Fourier transforms satisfy the ridge property i.e. $\mid \widehat{P}(\mathbf{u} + \mathrm{i} \mathbf{R}) \mid \leq \widehat{P}(\mathrm{i} \mathbf{R}), \; \forall \; \mathbf{u} \in \mathbb{R}^d, \mathbf{R} \in \delta_P$, and $\mid \Phi_{\boldsymbol{X}_T}(\mathbf{u} + \mathrm{i} \mathbf{R}) \mid \leq \Phi_{\boldsymbol{X}_T}(\mathrm{i} \mathbf{R}), \; \forall \; \mathbf{u} \in \mathbb{R}^d, \mathbf{R} \in \delta_X$, hence the integrand \eqref{QOI} can be bounded by
	\begin{equation}
		\mid g(\mathbf{u}; \mathbf{R}, \boldsymbol{\Theta}_m, \boldsymbol{\Theta}_p) \mid  \leq  (2 \pi)^{-d} e^{-r T}  \mid \Phi_{\boldsymbol{X}_T}(\mathrm{i} \mathbf{R}) \mid \mid \widehat{P}(\mathrm{i} \mathbf{R}) \mid = \mid g( \mathbf{0}_{\mathbb{R}^d}; \mathbf{R}) \mid, \forall \; \mathbf{u} \in \mathbb{R}^d, \mathbf{R} \in \delta_V.
	\end{equation}
\end{proof}

Equation \eqref{new_opt} cannot be solved analytically,  especially in high dimensions; therefore, we solve it numerically,  approximating  $\mathbf{R}^\ast$ by  $\bar{\mathbf{R}}=(\bar{R}_1,\ldots,\bar{R}_d)$. In this context, we used the interior point method \cite{byrd2000trust,byrd1999interior} with an accuracy of order $ 10^{-6}$; other algorithms such as as L-BFGS-B were tested and work effectively.

The numerical investigation through  different models and parameters  (for illustration, we refer to  Figure  \ref{fig: (Left) Shape of the integrand w.r.t  R, 
	(Right)  convergence w.r.t , using Gauss-Laguerre Quadrature for European put option under (a) GBM  (b) VG (c) NIG pricing models.}  for  the   single put option under the different models, and Figure \ref{damping_smoothing_effect_fig}  for the 2D-Basket put under VG) confirms  that the damping parameters have a considerable effect on the properties of the integrand, particularly its peak,  tail-heaviness, and  oscillatory behavior.  We observed that the damping parameters that produce the lowest peak of the integrand around the origin are associated with  a  faster convergence of the relative quadrature error than  other damping parameters.  Moreover,  we observed that highly peaked integrands are  more likely to oscillate, implying a deteriorated  convergence of the numerical quadrature. 	 Independent of the   quadrature methods   explained in  Section \ref{section_det_quad}, this observation was consistent for several parameter constellations under the three tested pricing dynamics, GBM, VG, and NIG, and for different dimensions of the basket put and rainbow options.  Section \ref{num_high_dim_damping_sec}  illustrates the computational advantage of the optimal  damping rule on the error convergence for the  multi-asset basket put and call on min options under  different  models.
\begin{figure}[h!]
	\centering 
	\begin{subfigure}{0.7\textwidth}
		\includegraphics[width=\linewidth]{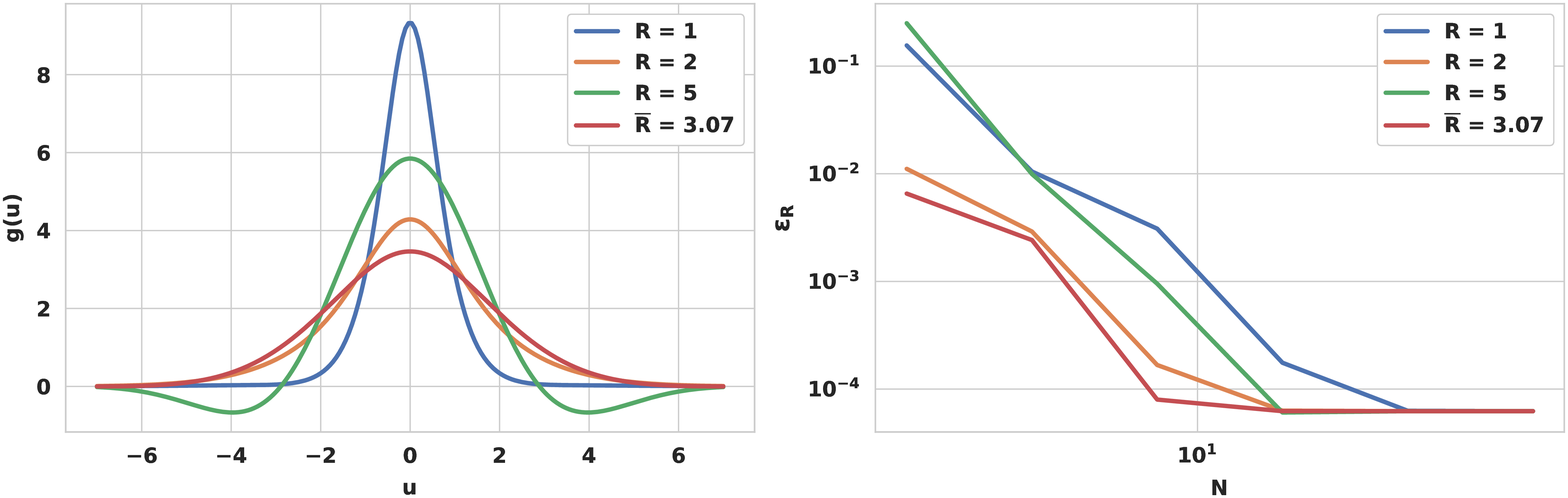}
		\caption{$S_0=100, K=100, r=0 \%, T=1,\sigma = 0.4$}
		\label{1d_put_gbm_damping}
	\end{subfigure}
	\begin{subfigure}{0.7\textwidth}
		\includegraphics[width=\linewidth]{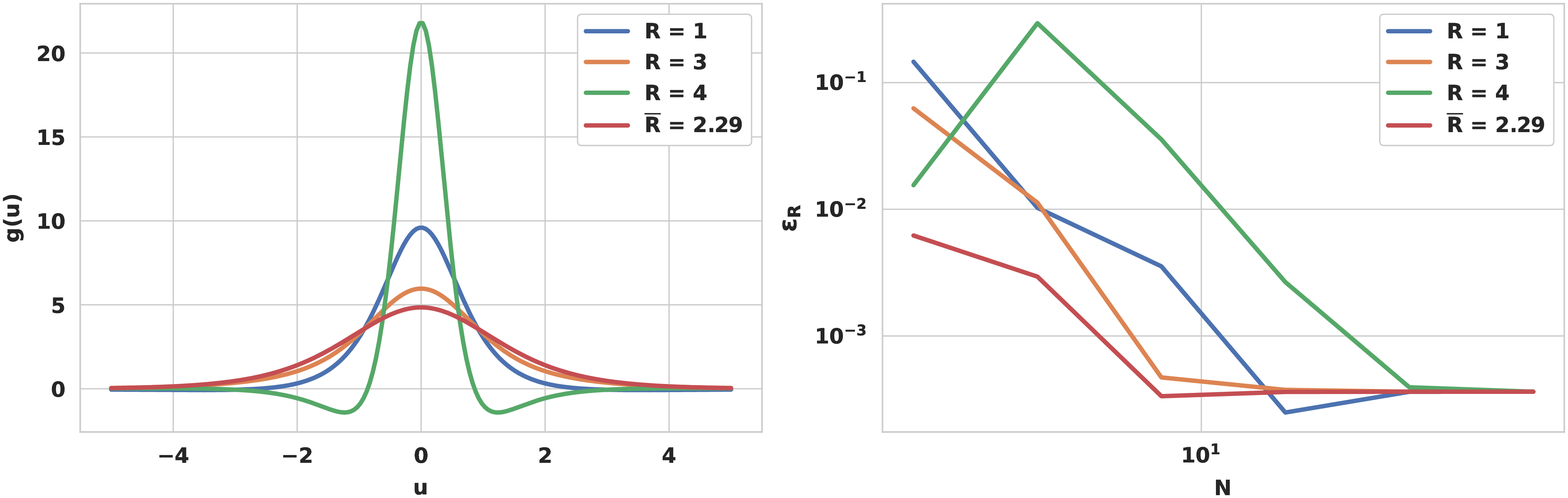}
		\caption{$S_0=100, K=100, r=0 \%, T=1,\sigma = 0.4, \theta = -0.3, \nu= 0.257$}
		\label{1d_put_vg_damping}
	\end{subfigure}
	\begin{subfigure}{0.7\textwidth}
		\includegraphics[width=\linewidth]{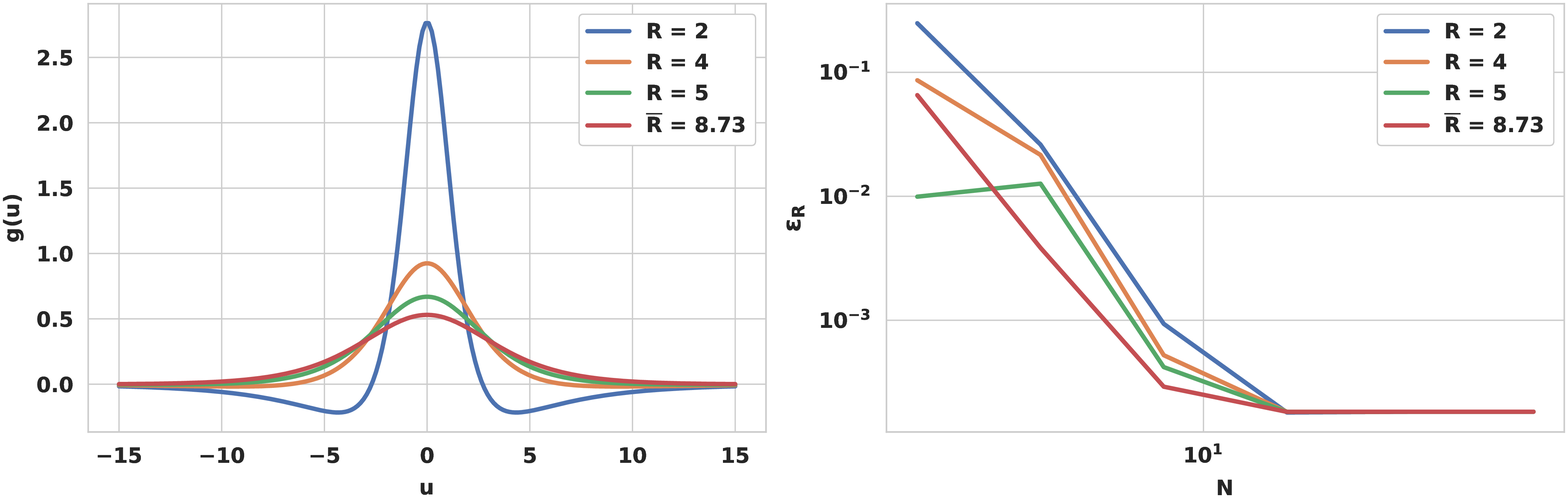}
		\caption{$S_0=100, K=100, r=0 \%, T=1,\alpha = 15, \beta= -3, \delta= 0.2$}
		\label{1d_put_nig_damping}
	\end{subfigure}
	\caption{1D illustration: (Left) Shape of the integrand w.r.t~the damping parameter, R.
		(Right) 	$\mathcal{E}_{R}$ convergence w.r.t.~$N$, using Gauss--Laguerre quadrature for the European put option under (a) GBM,  (b) VG, and (c) NIG pricing models.  The relative quadrature error $\mathcal{E}_{R} $ is defined as $	\mathcal{E}_{R} = \frac{ \mid Q_{N}[g] - \text{Reference Value} \mid }{ \text{Reference  Value}}$, where $Q_N$ is the quadrature estimator of \eqref{QOI} based on the Gauss--Laguerre rule.}
	\label{fig: (Left) Shape of the integrand w.r.t  R, 
		(Right)  convergence w.r.t , using Gauss-Laguerre Quadrature for European put option under (a) GBM  (b) VG (c) NIG pricing models.}
\end{figure}
\begin{figure}[h!]
\centering
	\begin{subfigure}{0.45\textwidth}
		\includegraphics[width=\linewidth]{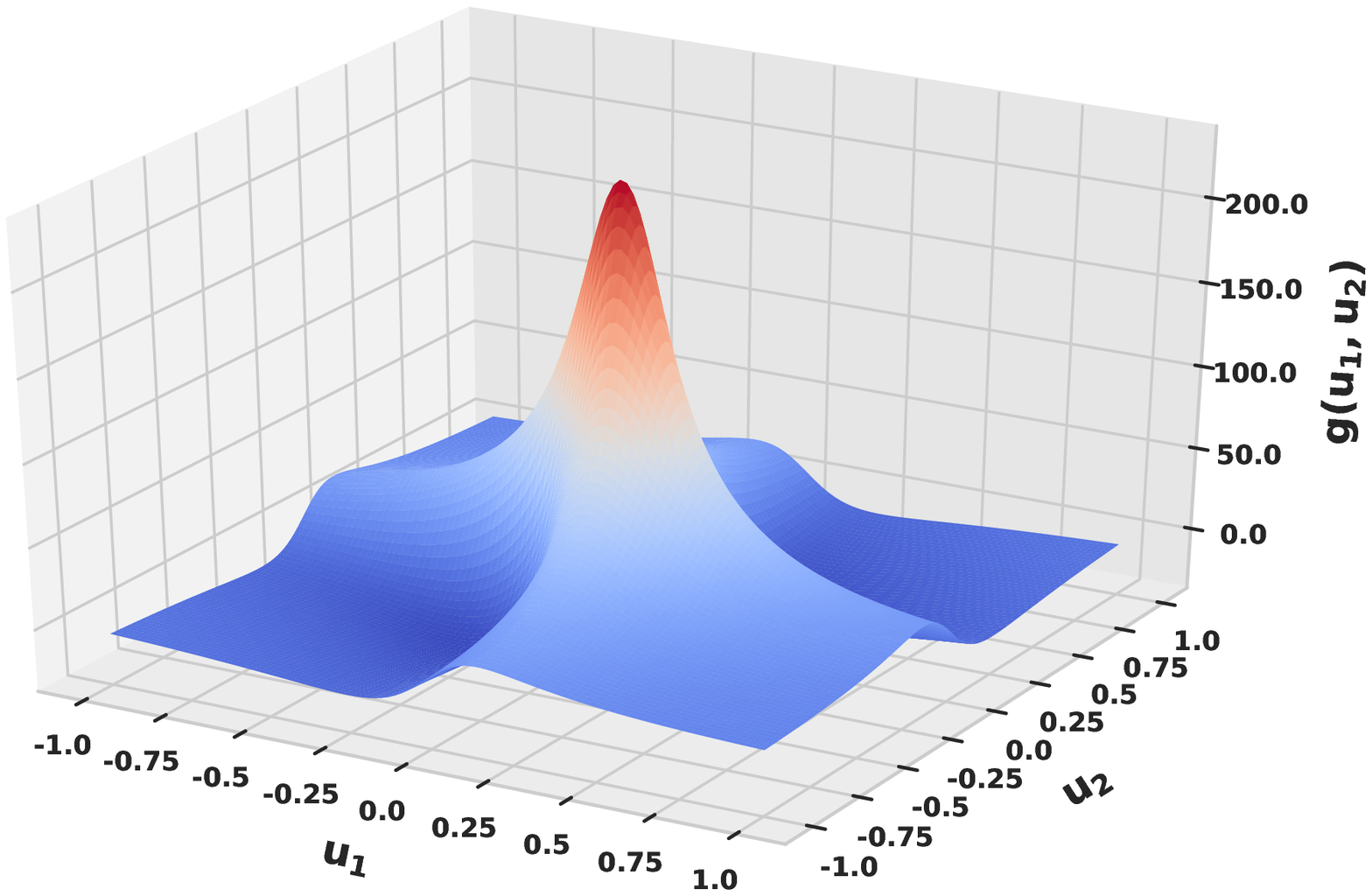}
		\label{damping_smoothing_effect_fig_a}
		\caption{}
	\end{subfigure}
	\begin{subfigure}{0.45\textwidth}
		\includegraphics[width=\linewidth]{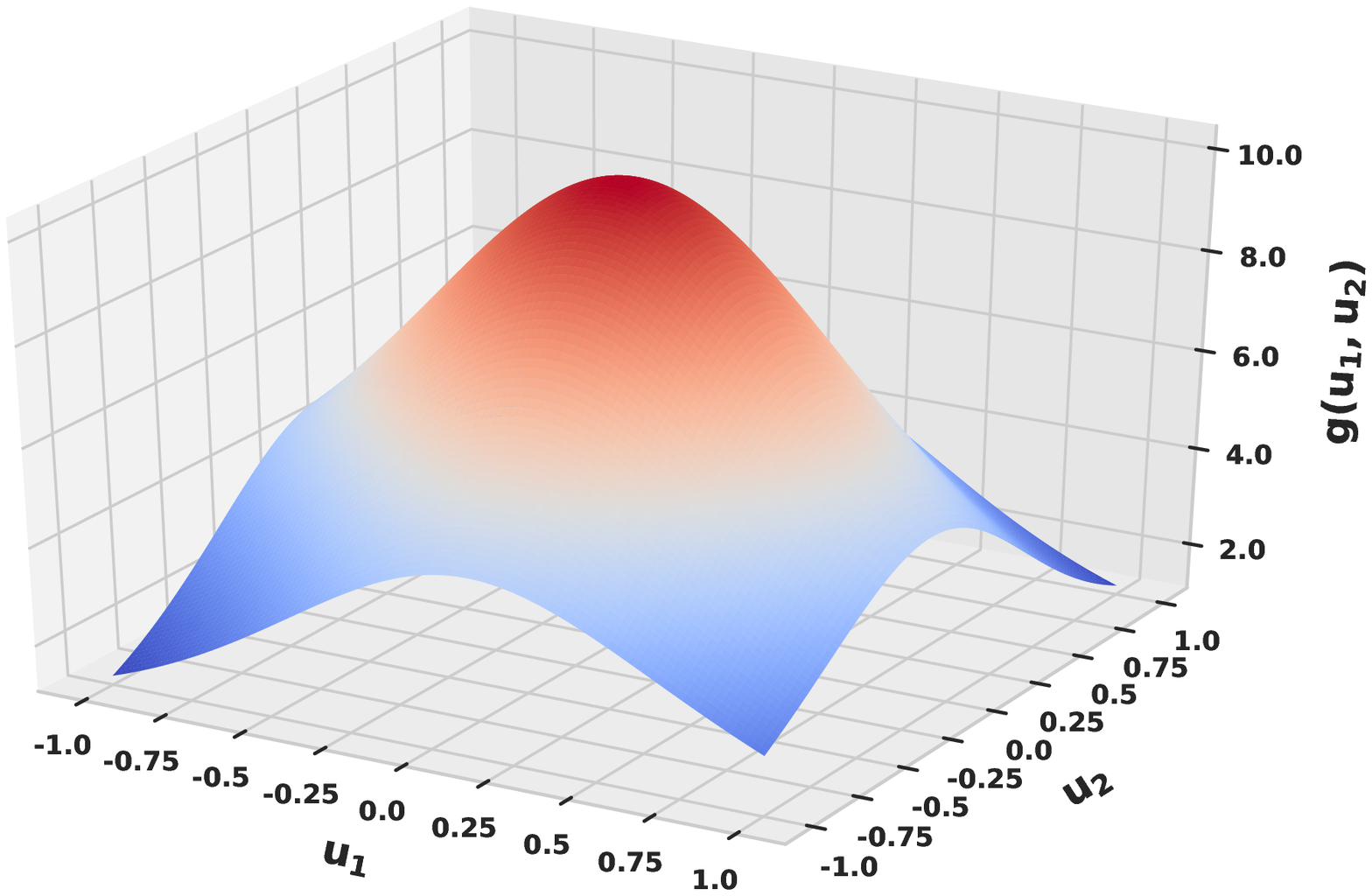}
		\label{damping_smoothing_effect_fig_b}
		\caption{}
	\end{subfigure} 
	\newline
	\begin{subfigure}{0.45\textwidth}
		\includegraphics[width=\linewidth]{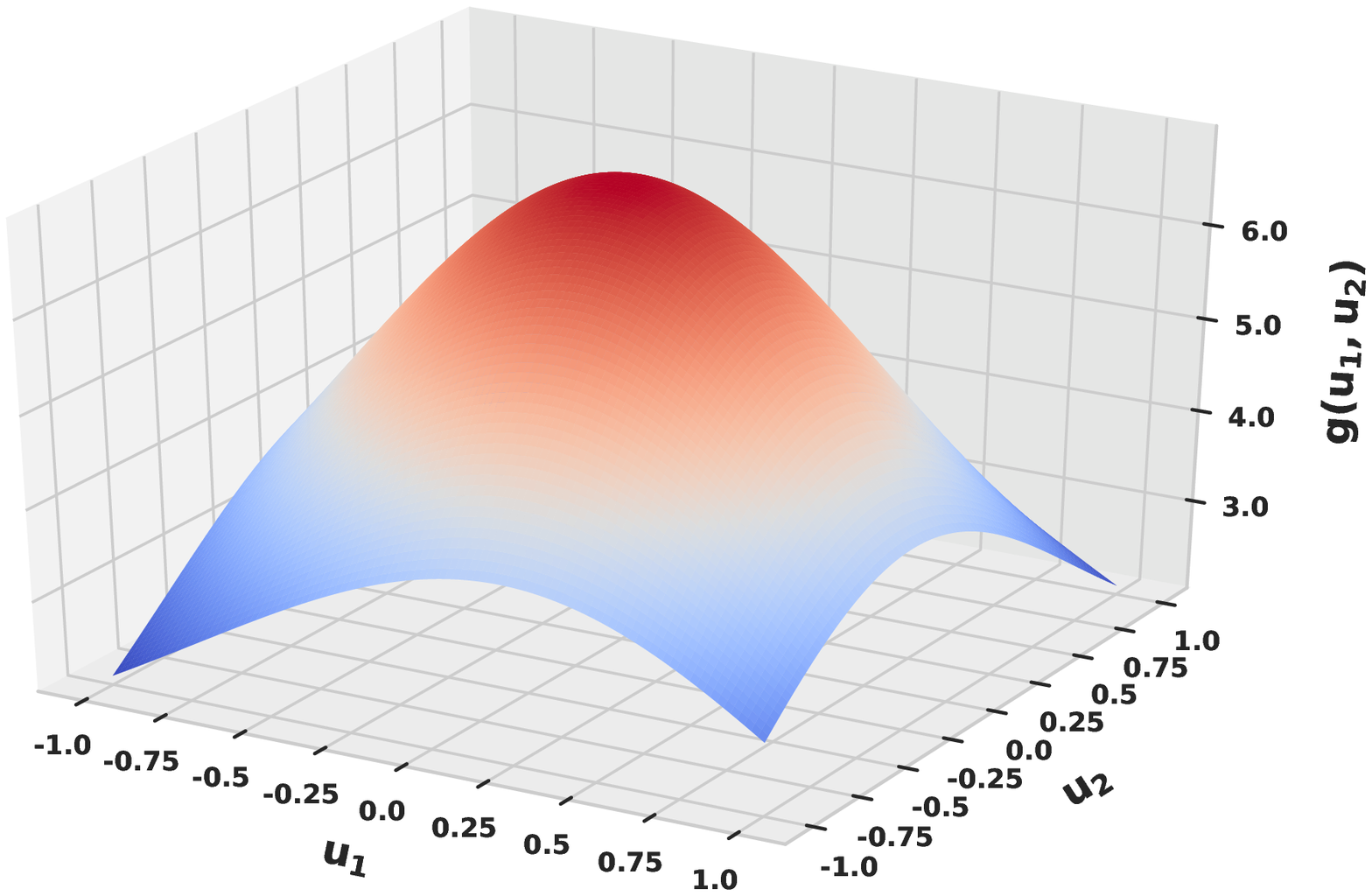}
		\label{damping_smoothing_effect_fig_c}
		\caption{}
	\end{subfigure}
	\begin{subfigure}{0.45\textwidth}
		\includegraphics[width=\linewidth]{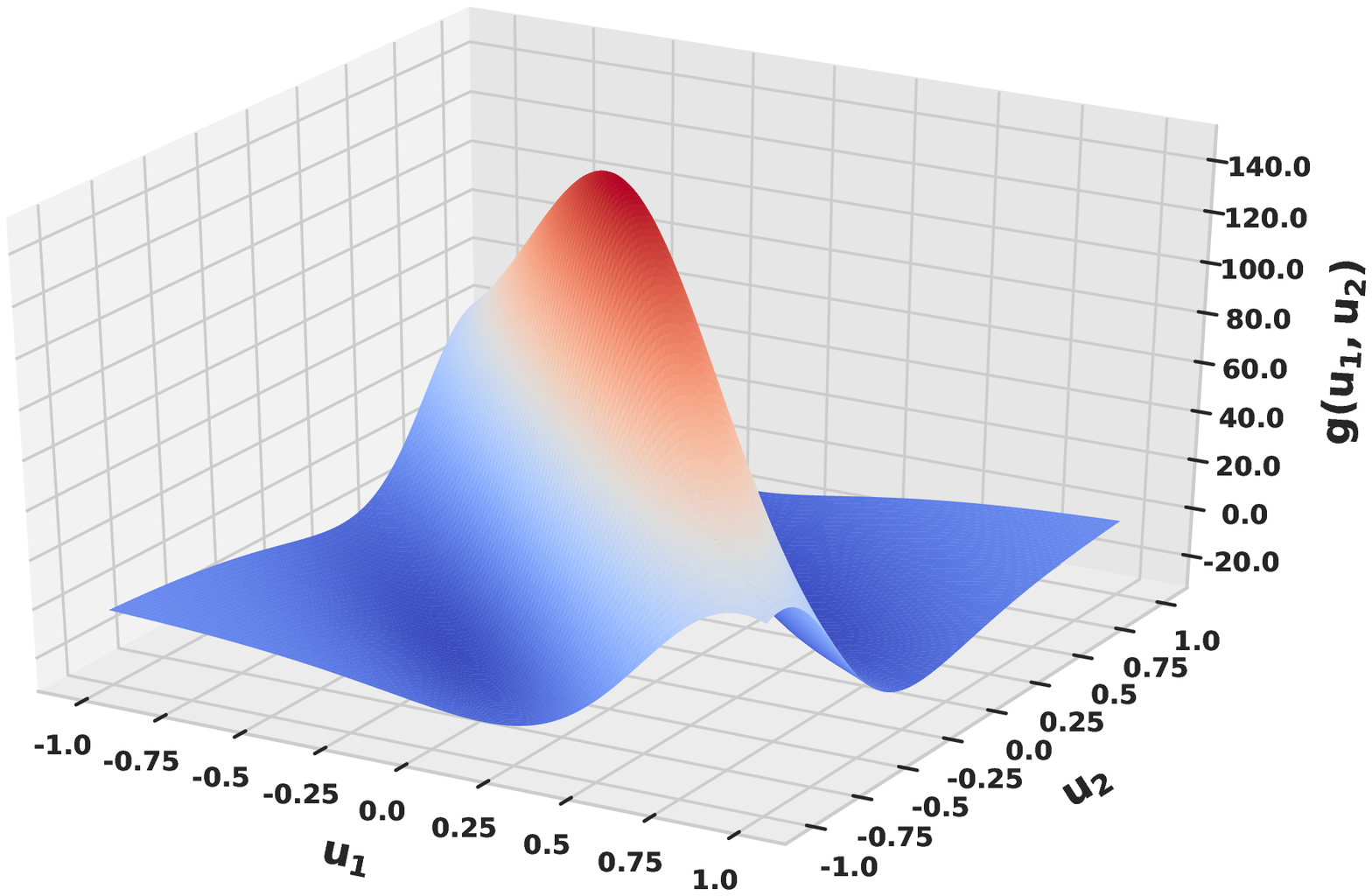}
		\label{damping_smoothing_effect_fig_d}
		\caption{}
	\end{subfigure}
	\vspace{0.2cm}
	\caption{Effect of the damping parameters on the shape of the integrand in the case of 2D-basket put option under the VG model with parameters $\boldsymbol{\sigma} = (0.4,0.4)$, $\boldsymbol{\theta} = (-0.3,-0.3), \nu = 0.257$. (a) $\mathbf{R}=(0.2,0.2)$  
	(b) $\mathbf{R}=(1,1)$, (c) $\mathbf{R}=(2,2)$, (d) $\mathbf{R}=(3,3)$.  }
	
	\label{damping_smoothing_effect_fig}
\end{figure}

\begin{remark}[Case of isotropic integrand]
	The $d$-dimensional optimization problem  \eqref{new_opt} is  simplified further to a 1D problem 	when the integrand is isotropic.
\end{remark}
\begin{remark}[Improving the damping parameters rule]
	Other rules for choosing the damping parameters can be investigated to improve the numerical convergence of quadrature methods. One can   account for additional features,   such as (i) the distance of the damping parameters to the poles, which affects the choice of the integration contour in \eqref{eq:integ_error_contour}, or (ii)  controlling the regularity of the integrand via high-order derivative estimates. 
	 However, we expect   such rules  to be more  complicated and computationally expensive (e.g., the evaluation of the gradient of the integrand). Investigating other rules remains open for  future work.
\end{remark}	

\subsection{Numerical evaluation of the inverse Fourier integrals using hierarchical deterministic quadrature methods}
\label{section_det_quad}
We aim to approximate   \eqref{QOI} efficiently using a tensorization of quadrature formulas over $\mathbb{R}^d$. When using Fourier transforms for option pricing, the standard numerical approach  truncates and discretizes the integration domain and  uses FFT based on bounded equispaced quadrature formulas, such as the trapezoidal and Simpson's rule. The  FFT is restricted to  the use of uniform quadrature mesh, in contrary to Gaussian quadrature rules which have higher polynomial exactness ($N$-point Gaussian quadrature rule is exact up to polynomials of degree $2N-1$).  Moreover, using FFT requires pre-specifying the truncation range. This option is efficient in the 1D setting, as the estimation of the truncation intervals based, for instance, on the cumulants, was widely covered in the literature. It remains affordable even though the additional cost might be high due to the inappropriate choice of truncation parameters. However, this is not the case in the  multidimensional  setting because determining the truncation parameters becomes more challenging. Moreover, the truncation errors nontrivially depend on  the damping parameter values. Choosing larger than necessary truncation domains  leads to a  more significant increase in the computational effort for higher  dimensions. Finally, FFT-based approaches  need to be followed by interpolation techniques to obtain the option values at the desired strikes grid, which may lead to loss in accuracy, in contrary to DI methods, which can be efficiently vectorized w.r.t~the desired strikes grid.

For these reasons, we choose the DI approach with  Gaussian quadrature rules.  Moreover,  our numerical investigation  (see Appendix \ref{num_lag_herm_sec}) suggests that  Gauss--Laguerre quadrature  exhibits faster convergence  than the Gauss--Hermite rule. Therefore, we  used  Laguerre quadrature on semi-infinite domains after applying the necessary transformations.

Before defining the  multivariate quadrature estimators, we first introduce the notation in the univariate setting (For more details see \cite{chen2018sparse}). Let $\beta$ denotes a non-negative integer, referred to as the ``discretization level," and  $m: \mathbb{N} \rightarrow \mathbb{N}$ represents a strictly increasing function with $m(0)=0$ and $m(1)=1$, called a   ``level-to-nodes function." At each level $\beta$, we consider a set of $m(\beta)$  distinct quadrature points $ \mathcal{H}^{m(\beta)}=\left\{x_{\beta}^{1}, x_{\beta}^{2}, \ldots, x_{\beta}^{m(\beta)}\right\} \subset \mathbb{R}$, and a set of  quadrature weights, $\boldsymbol{\omega}^{m(\beta)}=\left\{\omega_{\beta}^{1}, \omega_{\beta}^{2}, \ldots, \omega_{\beta}^{m(\beta)}\right\} .$ We also let $C^{0}(\mathbb{R})$ be the space of real-valued
continuous functions over $\mathbb{R}$. We  define the univariate quadrature operator applied to a function $f \in C^{0}(\mathbb{R}) $ as follows:
\begin{equation*}
	Q^{m(\beta)}: C^{0}(\mathbb{R}) \rightarrow \mathbb{R}, \quad Q^{m(\beta)}[f]:=\sum_{j=1}^{m(\beta)} f\left(x_{\beta}^{j}\right) \omega_{\beta}^{j} \text { . }
\end{equation*}
In our case,  in \eqref{QOI}, we have a multivariate integration problem of $g$ (see \eqref{g_integrand})
 over $\mathbb{R}^{d}$. Accordingly,  for a multi-index $\boldsymbol{\beta}=\left(\beta_{i}\right)_{i=1}^{d} \in \mathbb{N}^{d}$, the $d$-dimensional quadrature operator applied to $g$ is defined as\footnote{The $n$-th quadrature operator  acts only on the $n$-th variable of $g$.}
\begin{align*}
	&Q_d^{m(\boldsymbol{\beta})}: C^{0}\left(\mathbb{R}^{d}\right) \rightarrow \mathbb{R}, \quad Q_d^{m(\boldsymbol{\beta})}=\bigotimes_{i=1}^{d} Q^{m\left(\beta_{i}\right)},\nonumber\\
		&Q_d^{m(\boldsymbol{\beta})}[g]:= \sum_{j_1=1}^{m(\beta_1)} \ldots \sum_{j_d=1}^{m(\beta_d)} \omega_{\beta_{1}}^{j_1}  \ldots \omega_{\beta_{1}}^{j_d}    g(x_{\beta_1}^{j_1}, \ldots, x_{\beta_d}^{j_d}) :=  \sum_{j=1}^{\# \mathcal{T}^{m(\boldsymbol{\beta})}} g\left(\widehat{x}_{j}\right) \bar{\omega}_{j},
\end{align*}
where $\widehat{x}_{j} \in \mathcal{T}^{m(\boldsymbol{\beta})}:=\prod_{i=1}^{d} \mathcal{H}^{m\left(\beta_{i}\right)}$ (with cardinality $
\# \mathcal{T}^{m(\boldsymbol{\beta})}=
\prod_{i=1}^{d} m\left(\beta_{i}\right)$ and 
$m( \beta_i) =N_i$ quadrature points in the dimension of $x_i$), and $\bar{\omega}_{j}$ is a product of the weights of the univariate quadrature rule.  To simplify the notation, we replace $Q_d^{m(\boldsymbol{\beta})}$ with $Q_d^{\boldsymbol{\beta}}$. 

We  define the set of differences $\Delta Q_{d}^{\boldsymbol{\beta}}$ for indices $i \in \{1,\ldots,d\}$ as follows:
\begin{equation}
	\Delta_{i} Q_{d}^{\boldsymbol{\beta}}:=\left\{\begin{array}{l}
		Q_{d}^{\boldsymbol{\beta}}-Q_{d}^{\boldsymbol{\beta}^{\prime}}, \text { with } \; \boldsymbol{\beta}^{\prime}=\boldsymbol{\beta}-\mathbf{e}_{i}, \text { when } \beta_{i}>0, \\
		Q_{d}^{\boldsymbol{\beta}}, \quad \text { otherwise, }
	\end{array}\right. 
\end{equation}
where $\mathbf{e}_{i}$ denotes the $i$th $d$-dimensional unit vector.  Then,  using the telescopic property, the  quadrature estimator, defined w.r.t.~a choice of the set of multi-indices $\mathcal{I} \subset \mathbb{N}^{d}$, is expressed by\footnote{For instance, when $d=2$, then $
		\small	{	\Delta Q_{2}^{\boldsymbol{\beta}}=\Delta_{2} \Delta_{1} Q_{2}^{\left(\beta_{1}, \beta_{2}\right)}
			=Q_{2}^{\left(\beta_{1}, \beta_{2}\right)}-Q_{2}^{\left(\beta_{1}, \beta_{2}-1\right)}-Q_{2}^{\left(\beta_{1}-1, \beta_{2}\right)}+Q_{2}^{\left(\beta_{1}-1, \beta_{2}-1\right)}}$.}$^{,\thinspace}$\footnote{
		To ensure the validity of the telescoping sum expansion, the index set $\mathcal{I}$ must 
		satisfy the admissibility condition (\ie, $
		\boldsymbol{\beta}\in \mathcal{I}, \boldsymbol{\alpha} \leq \boldsymbol{\beta} \Rightarrow \boldsymbol{\alpha} \in \mathcal{I},\text{ where} \: \boldsymbol{\alpha} \leq \boldsymbol{\beta} \: \text{is defined as} \:  \alpha_i \leq \beta_i, i = 1,\ldots,d$).}
	\begin{equation}
		\label{quad_estimate}
		Q_{d}^{\mathcal{I}}=\sum_{\boldsymbol{\beta} \in \mathcal{I}} \Delta Q_{d}^{\boldsymbol{\beta}}, \quad \text{with} 	\: \Delta Q_{d}^{\boldsymbol{\beta}}=\left(\bigotimes_{i=1}^{d} \Delta_{i}\right) Q_{d}^{\boldsymbol{\beta}},
	\end{equation}
 and the quadrature error can  be written as
\begin{equation}\label{eq: quad_error}
	\mathcal{E}_{Q}=\left|Q^{\infty}_d[g]-Q_d^{\mathcal{I}}[g]\right| \leq \sum_{\boldsymbol{\beta} \in \mathbb{N}^{d} \backslash \{ \mathcal{I} \}}\left|\Delta Q_d^{\boldsymbol{\beta}}[g]\right|,
\end{equation}
where
\begin{equation*}
	\label{quad_op}
	Q_{d}^{\infty}:=\sum_{\beta_{1}=0}^{\infty} \cdots \sum_{\beta_{d}=0}^{\infty} \Delta Q_{d}^{\left(\beta_{1}, \ldots, \beta_{d}\right)}=\sum_{\boldsymbol{\beta} \in \mathbb{N}^d} \Delta Q_{d}^{\boldsymbol{\beta}}.
\end{equation*}
In Equation (\ref{quad_estimate}), the choice of  (i) the strategy for  the construction of the index set $\mathcal{I}$ and (ii) the hierarchy of quadrature points determined by $m(\cdot)$ defines different hierarchical quadrature methods.  Table \ref{index_set_table} presents  the details of the   methods considered in this work. 
\begin{table}[h!]
	\centering
	\begin{tabular}{| p{3.87cm} | p{5.55cm} | p{6.6cm} |}
		\hline   \textbf{Quadrature Method} & $m(\cdot)$ &  $\mathcal{I}$  \\
		\hline
		Tensor Product (TP)&  $m(\beta)=\beta$
		& \small $ \mathcal{I^{\text{TP}}}(l) = \{\boldsymbol{\beta} \in\mathbb{N}^d: \;\; \max_{1 \le i \le d}(\beta_i-1) \leq l\}$   \\ 
		\hline
		Smolyak (SM) Sparse Grids & $ m(\beta)=  2^{\beta-1}+1,\, \beta>1,m(1) =1  $
		&  \small $\mathcal{I^{\text{SM}}}(l)=\{\boldsymbol{\beta} \in\mathbb{N}^d: \;\; \sum_{1 \le i \le d}(\beta_i-1) \leq l\}$ \\
		\hline
		Adaptive Sparse Grid  Quadrature (ASGQ) & $ m(\beta )=  2^{\beta-1}+1,\, \beta>1,m(1) =1  $
		& \small 
		$\mathcal{I^{\text{ASGQ}}}=\left\{\boldsymbol{\beta} \in \mathbb{N}_{+}^{d}: P_{\boldsymbol{\beta}} \geq \bar{T}\right\}$  \newline(see \eqref{profit_rule} and  \eqref{error_contr})\\
		\hline
	\end{tabular}
	\caption{Construction details for the quadrature methods. $l \in \mathbb{N}$ represents a given level.  $\bar{T} \in \rset$ is a threshold value.}
	\label{index_set_table}
\end{table} 

In many situations, the tensor product (TP)  estimator can become rapidly unaffordable because the number of function evaluations increases exponentially with the problem dimensionality, known as the \emph{curse of dimensionality}. We use Smolyak (SM) and ASGQ  methods based on sparsification and dimension-adaptivity techniques to overcome this issue. For both TP and SM methods, the construction of the index set is performed  a priori. However,  ASGQ allows for  the a posteriori and adaptive construction of the index set $\mathcal{I}$ by greedily exploiting the mixed regularity of the integrand  during the actual computation of the quantity of interest. The construction of  $\mathcal{I^{\text{ASGQ}}}$  is performed through   profit thresholding,  where new indices are selected  iteratively based on  the error versus cost-profit rule,   with a hierarchical surplus defined by
\begin{equation}
	\label{profit_rule}
	P_{\boldsymbol{\beta}}=\frac{\left|\Delta E_{\boldsymbol{\beta}}\right|}{\Delta \mathcal{W}_{\boldsymbol{\beta}}},
\end{equation}
where $\Delta \mathcal{W}_{\boldsymbol{\beta}}$ is   the work contribution (\ie,  the computational cost required to add $\Delta Q_{d}^{\boldsymbol{\beta}}$ to $Q_{d}^{\mathcal{I^{\text{ASGQ}}}}$)  and  $\Delta E_{\boldsymbol{\beta}}$ is  the error contribution (\ie, a measure of how much the quadrature error would decrease once $\Delta Q_{d}^{\boldsymbol{\beta}}$ has been added to $Q_{d}^{\mathcal{I^{\text{ASGQ}}}}$):
\begin{align}	\label{error_contr}
	\Delta E_{\boldsymbol{\beta}} &=\left|Q_{d}^{\mathcal{I^{\text{ASGQ}}} \cup\{\boldsymbol{\beta}\}}[g]-Q_{d}^{\mathcal{I^{\text{ASGQ}}}}[g]\right|\\
	\Delta \mathcal{W}_{\boldsymbol{\beta}} &=\operatorname{Work}\left[Q_{d}^{\mathcal{I^{\text{ASGQ}}}\cup\{\boldsymbol{\beta}\}}[g]\right]- \operatorname{Work}\left[Q_{d}^{\mathcal{I^{\text{ASGQ}}}}[g]\right]. \nonumber
\end{align}
The convergence speed for  all quadrature methods  in this work is determined by the behavior of the quadrature error defined in \eqref{eq: quad_error}. In this context, given the model and option parameters,  the  convergence rate depends on  the damping parameter values, which control the regularity of the integrand $g$  in the Fourier space.
 
We let $N:= \prod_{i=1}^{d} m\left(\beta_{i}\right)$ denote the total number of quadrature points used by each method. For the TP method, we have the following  \cite{davis2007methods}:\begin{equation}
 	\label{eq:error_bound_TP}
 	\mathcal{E}^{\text{TP}}_{Q}\left(N;\mathbf{R}\right)=\mathcal{O}\left(N^{- \frac{r_t}{d}}\right)
 \end{equation}
 for functions with bounded total derivatives up to order $r_t := r_t(\mathbf{R})$.  When using SM sparse grids (not adaptive), we obtain the following  \cite{smolyak1963quadrature,WASILKOWSKI19951,gerstner1998numerical,barthelmann2000high}:  \begin{equation}
	\label{eq:error_bound_SM}
	\mathcal{E}^{\text{SM}}_{Q}\left(N ;\mathbf{R}\right)=\mathcal{O}\left(N^{-r_m}\left(\log N\right)^{(d-1)(r_m)+1)}\right)
\end{equation}
for functions with bounded mixed partial derivatives up to order $r_m:=r_m(\mathbf{R})$. Moreover, it was observed in \cite{Gerstner2003DimensionAdaptiveTQ} that  the convergence is even spectral for  analytic functions ($r_m \rightarrow +\infty$). For the ASGQ method, we achieve  \cite{chen2018sparse}
 \begin{equation}
 	\label{eq:error_bound_ASGQ}
 	\mathcal{E}^{\text{ASGQ}}_{Q}\left(N;\mathbf{R}\right)=\mathcal{O}\left(N^{-r_w/2}\right)
 \end{equation}
where  $r_w$ is  related to the degree of   weighted mixed regularity of the integrand. 
 
  In \eqref{eq:error_bound_TP}, \eqref{eq:error_bound_SM}, and \eqref{eq:error_bound_ASGQ}, we emphasize the dependence of the convergence rates on the damping parameters $\mathbf{R}$, which is only valid in this context because these parameters control the regularity of the integrand in the Fourier space.    Moreover,  our optimized choice of  $\mathbf{R}$ using \eqref{or_opt} is not only used to increase the order  of  boundedness of the  derivatives of $g$  but also to reduce the bounds on these derivatives.

%% file: Num_exp.tex
In this section, we present the results of  different numerical experiments conducted for pricing	 multi-asset European equally weighted basket put ($w_i = \frac{1}{d}, i =1,\ldots,d$) and call on min options.
These examples are tested under  the multivariate  (i) GBM, (ii)  VG and (iii)  NIG models with various  parameter constellations for different dimensions $d \in \{2, 4, 6\}$. The tested model parameters are justified from  the literature  on model  calibration \cite{kirkby2015efficient,Wiel2015ValuationOI,choi2018sum,bayer2018smoothing,aguilar2020some,healy2021pricing}.   The detailed  illustrated examples are	presented in Tables \ref{tab_mgbm}, \ref{tab_mvg}, and \ref{tab_mnig}. To compare  the methods in this work, we consider  relative errors normalized by the reference prices. The error is the relative quadrature error defined as  $\mathcal{E}_{R} = \frac{ \mid Q_{d}^{\mathcal{I}}[g] - \text{Reference  Value} \mid }{\text{Reference  Value}}$ when using quadrature methods, and the $95\%$ relative statistical error of the MC method is estimated by the virtue of the central limit theorem (CLT) as
	\begin{equation}
		\mathcal{E}_{R} \approx  \frac{ C_{\alpha} \times \sigma_{M}}{ \text{Reference  Value} \times \sqrt{M}}
		\label{CLT_formula}
	\end{equation}
	where $C_{\alpha} = 1.96$ for $95 \%$ confidence level, $M$ is the number of MC samples, and $\sigma_{M}$ is the standard deviation of the quantity of interest.

The numerical results were obtained using a cluster machine with the following characteristics: clock speed 2.1 GHz,  \#CPU cores: 72, and memory per node 256 GB. Furthermore, the computer code is written in the MATLAB (version R2021b). The  ASGQ implementation was based on \url{https://sites.google.com/view/sparse-grids-kit} (For more details on the implementation we refer to \cite{piazzola.tamellini:SGK}).  

Through various tested examples, in section \ref{num_adap_sec}, we demonstrate the importance of sparsification and adaptivity in the Fourier space for accelerating quadrature convergence. Moreover, in section \ref{num_high_dim_damping_sec} , we reveal  the importance of the choice of the damping parameters  on the numerical complexity of the used quadrature methods. In Section \ref{ODHAQ_vs_COS}, we compare our approach against one of the state of the art Fourier-pricing methods, namely the COS method  \cite{fang2008novel,von2015benchop}, for 1D  and 2D cases, and we  show the advantage of our approach when the damping parameters are tuned appropriately. Finally,  in Section \ref{num_quad_vs_mc_sec}, we illustrate  that our approach  achieves substantial computational gains over the MC  method for different dimensions and parameter constellations to meet a certain relative error tolerance (of practical interest) that we set to be sufficiently small.
	\begin{table}[h!]
		\hspace{-2cm}
		\centering
		\small
		\begin{tabular}{|p{1.32cm} | p{1.75cm} | p{4.7cm} | p{2.7cm} |p{5.8cm} |}
		\hline
		 \small\textbf{Example }&\small	\textbf{Option} &\small\textbf{Parameters} & \small{$\underset{(95\%  \text{ Statistical Error)}}{\text{\textbf{Reference Value}}}$} &\small\textbf{ Optimal damping  parameters $\mathbf{\overline{R}}$   } 
			\\
			\hline 
		Example 1&	2D-Basket put & \small{$\boldsymbol{\sigma} = (0.4,0.4), \mathbf{C}= I_2, K = 100$}
			& $\underset{(8e^{-04})}{11.4474}$ & \small{$(2.5,2.5)$} \\
			\hline
	   Example 2&	2D-Basket put & \small{$\boldsymbol{\sigma} = (0.4,0.8), \mathbf{C}= I_2, K = 100$}
			& $\underset{(1.2e^{-03})}{17.831}$ & \small{$(2.1,1.2)$} \\
			\hline
		 Example 3&	2D-Call on min & \small{$\boldsymbol{\sigma} = (0.4,0.4),\mathbf{C}= I_2, K = 100$}
			& $\underset{(6e^{-04})}{3.4603}$ & \small{$(-3.4,-3.4)$} \\
			\hline
	Example	 4&	2D-Call on min & \small{$\boldsymbol{\sigma} = (0.4,0.8), \mathbf{C} = I_2, K = 100$}
			& $\underset{(8.2e^{-04})}{3.7411}$ & \small{$(-3.6,-1.8)$} \\
		\Xhline{7\arrayrulewidth}
		
	Example	5&	4D-Basket put & \small{$\boldsymbol{\sigma} = (0.4,0.4,0.4,0.4), \newline \mathbf{C} = I_4, K = 100$}
			& $\underset{(6e^{-04})}{8.193}$ & \small{$(2.1,2.1, 2.1,2.1)$} \\
			\hline
		Example	6&4D-Basket put  &  \small{$ \boldsymbol{\sigma}=(0.2,0.4,0.6,0.8), \newline \mathbf{C} = I_4, K = 100$}  &  $\underset{(8 e^{-04})}{11.3014}$ & \small{$ (2.4,1.9 ,1.5,1.2) $} 
			 \\
			\hline
		Example	7& 4D-Call on min   &  \small{$\boldsymbol{\sigma} = (0.4,0.4,0.4,0.4),\newline \mathbf{C} = I_4, K = 100$ }& $\underset{( 2  e^{-04})}{0.317}$& \small{$(-3.1,-3.1,-3.1,-3.1)$}  \\
			\hline
		Example	8&4D-Call on min  &   \small{$\boldsymbol{\sigma} = (0.2,0.4,0.6,0.8),\newline \mathbf{C}= I_4,K = 100$}  & $\underset{(1e^{-04})}{0.2382} $ &  \small{$ (-6.4,-3.1,-2.1,-1.6) $}
			 \\
		\Xhline{7\arrayrulewidth}
		Example	9&6D-Basket put  & \small{$\boldsymbol{\sigma}=(0.4,0.4,0.4,0.4,0.4,0.4),$} $ \mathbf{C}= I_6, K= 60$ & $\underset{(8.8  e^{-06})}{0.0041}    $ & \small{$(2.0, 2.0, 2.0, 2.0, 2.0, 2.0)$} 
			\\
			\hline
		Example	10&6D-Basket put   &  \small{$ \boldsymbol{\sigma} = (0.2,0.3,0.4,0.5,0.6,0.7),$ $ \mathbf{C}= I_6,K = 60$}  &  $ \underset{(1.8  e^{-05})}{0.012702 }  $ &  \small{$(2.3, 2.1,  1.9, 1.7, 1.5, 1.3) $}
			 \\
			\hline
		Example	11&6D-Call on min  &  \small{$\boldsymbol{\sigma} = (0.4,0.4,0.4,0.4,0.4,0.4),$ $ \mathbf{C}= I_6,K=100$}
			& $ \underset{(4.4  e^{-05})}{0.038} $ &  \small{$(-3.0, -3.0, -3.0, -3.0, -3.0, -3.0)$}\\
			\hline
		Example	12&6D-Call on min  & \small{$\boldsymbol{\sigma} = (0.2,0.3,0.4,0.5,0,6,0.7),$ $\mathbf{C}= I_6, K=100 $} 
			& $\underset{(3.7  e^{-05})}{0.0301}$    & \small{$(-6.0, -3.9, -3.0, -2.4, -2.0, -1.8)$}\\
			\hline
		\end{tabular}
		\caption{Examples  of multi-asset options under the multivariate GBM model. In all examples, $S_0^i=100, i=1,\ldots,d, T=1, r = 0.$ Reference values are computed with MC using $10^9$ samples, with  $95\%$ statistical error estimates reported between parentheses. $\bar{R}$ is rounded  to one decimal place.}
		\label{tab_mgbm}
	\end{table}
	\begin{table}[h!]
		\centering
		\small
		\begin{tabular}{|p{1.32cm} | p{1.75cm} | p{5.08cm} | p{2.7cm} |p{5.4cm} |}
			\hline
			\small\textbf{Example}&\small	\textbf{Option} &\small\textbf{Parameters} & \small{$\underset{(95\%  \text{ Statistical Error)}}{\text{\textbf{Reference Value}}}$} &\small\textbf{Optimal damping } 
			\\
			&&& &\small\textbf{parameters $\mathbf{\overline{R}}$  } 
			\\
			\hline 
			Example 13&	2D-Basket put & \small{$ \boldsymbol{\sigma}= (0.4,0.4)$,$\boldsymbol{\theta} =( -0.3,-0.3),$ \newline $\nu = 0.257, K = 100$}
			& $\underset{(1e^{-03})}{11.7589}$ & \small{$(1.7,1.7)$} \\
			\hline
		Example	14&	2D-Basket put & \small{$ \boldsymbol{\sigma}= (0.4,0.8)$,$\boldsymbol{\theta} =( -0.3,0),$ \newline $\nu = 0.257,  K = 100$}
			& $\underset{(1.2e^{-03})}{17.6688}$ & \small{$(1.7,1.0)$} \\
			\hline
		Example	15&	2D-Call on min & \small{$ \boldsymbol{\sigma}= (0.4,0.4)$,$\boldsymbol{\theta} =( -0.3,-0.3),$ \newline $\nu = 0.257,  K = 100$}
			& $\underset{(7e^{-04})}{3.9601}$ & \small{$(-3.5,-3.5)$} \\
			\hline
		Example	16&	2D-Call on min & \small{$ \boldsymbol{\sigma}= (0.4,0.8)$,$\boldsymbol{\theta} =( -0.3,0),$ \newline $\nu = 0.257,  K = 100$}
			& $\underset{(8e^{-04})}{3.3422}$ & \small{$(-4.0,-3.5)$} \\
		\Xhline{7\arrayrulewidth}
		Example	17&	4D-Basket put & \small{$ \boldsymbol{\sigma}= (0.4,0.4,0.4,0.4)$,\newline $\boldsymbol{\theta} =( -0.3,-0.3,-0.3,-0.3),$ \newline $\nu = 0.257,  K = 100$}
			& $\underset{(8e^{-04})}{8.9441}$ & \small{$(1.2,1.2, 1.2,1.2)$} \\
			\hline
		Example	18&4D-Basket put  &  \small{$ \boldsymbol{\sigma}= (0.2,0.4,0.6,0.8)$,\newline$\boldsymbol{\theta} =( -0.3,-0.2,-0.1,0)$,\newline $\nu = 0.257,  K = 100$}  &  $\underset{(8 e^{-04})}{11.2277}$ & \small{$ (1.6, 1.4, 1.1, 0.9) $} 
			\\
			\hline
	Example		19& 4D-Call on min   &  \small{$\boldsymbol{\sigma}= (0.4,0.4,0.4,0.4) $,\newline $\boldsymbol{\theta} =( -0.3,-0.3,-0.3,-0.3)$,\newline $ \nu = 0.257,  K = 100$ }& $\underset{( 2  e^{-04})}{0.6137}$& \small{$(-3.2,-3.2,-3.2,-3.2)$}  \\
			\hline
		Example	20&4D-Call on min  &   \small{$\boldsymbol{\sigma}= (0.2,0.4,0.6,0.8) $,\newline $\boldsymbol{\theta} =( -0.3,-0.2,-0.1,0),$ \newline $ \nu = 0.257, K = 100$}  & $\underset{(1e^{-04})}{0.2384} $ &  \small{$ (-6.6,-3.0,-2.0,-1.5) $}
			\\
		\Xhline{7\arrayrulewidth}
		Example	21&6D-Basket put  & \small{$\boldsymbol{\sigma}= (0.4,0.4,0.4,0.4,0.4,0.4)$, \newline $\boldsymbol{\theta}=-(0.3,0.3,0.3,0.3,0.3,0.3)$, $\nu = 0.257,K=60$} & $\underset{(1 e^{-06})}{0.1691}    $ & \small{$(1.1, 1.1, 1.1, 1.1, 1.1, 1.1)$} 
			\\
			\hline
		Example	22&6D-Basket put   &  \small{$ \boldsymbol{\sigma}= (0.2,0.3,0.4,0.5,0.6,0.7)$, \newline $\boldsymbol{\theta}=(-0.3,-0.2,-0.1,0,0.1,0.2)$,\newline $ \nu = 0.257,K=60$}  &  $ \underset{(5 e^{-05})}{0.04634 }  $ &  \small{$(2.1, 1.9, 1.7, 1.6, 1.4, 1.2) $}
			\\
			\hline
		Example	23&6D-Call on min  &  \small{$\boldsymbol{\sigma}= (0.4,0.4,0.4,0.4,0.4,0.4)$, \newline $\boldsymbol{\theta}=-(0.3,0.3,0.3,0.3,0.3,0.3)$,\newline $\nu = 0.257,K=100$}
			& $ \underset{(1 e^{-04})}{0.16248} $ &  \small{$ (-3.1, -3.1, -3.1, -3.1, -3.1, -3.1) $}\\
			\hline
		Example	24&6D-Call on min  & \small{$\boldsymbol{\sigma}= (0.2,0.3,0.4,0.5,0.6,0.7)$, \newline $\boldsymbol{\theta}=(-0.3,-0.2,-0.1,0,0.1,0.2)$,\newline $\nu = 0.257,K=100$}
			& $\underset{(4  e^{-05})}{0.02269 }$    & \small{$(-6.5, -3.7, -2.6, -2.0, -1.7, -1.4)$}\\
			\hline
		\end{tabular}
		\caption{Examples  of multi-asset options under the multivariate VG model.  In all examples, $S_0^i=100, i=1,\ldots,d, T=1, r = 0.$ Reference values are computed with MC using $10^9$ samples, with  $95\%$ statistical error estimates reported between parentheses. $\bar{R}$ is rounded  to one decimal place.}
		\label{tab_mvg}
	\end{table}
	\begin{table}[h!]
		\centering
		\small
		\begin{tabular}{|p{1.32cm} | p{1.75cm} | p{4.7cm} | p{2.7cm} |p{5.8cm} |}
			\hline
			\small\textbf{Example}&\small	\textbf{Option} &\small\textbf{Parameters} & \small{$\underset{(95\%  \text{ Statistical Error)}}{\text{\textbf{Reference Value}}}$} &\small\textbf{ Optimal damping  parameters $\mathbf{\overline{R}}$   } 
			\\
			\hline 
		Example	25&	2D-Basket put & \small{$\boldsymbol{\beta}=(-3,-3), \alpha=15,$\newline $ \delta = 0.2, \boldsymbol{ \Delta= I_2}, K = 100$}
			& $\underset{(3e^{-04})}{3.3199}$ & \small{$(6.1,6.1)$} \\
			\hline
	Example	26&	2D-Basket put & \small{$\boldsymbol{\beta}=(-3,0), \alpha=10,$\newline $ \delta = 0.2, \boldsymbol{ \Delta= I_2}, K = 100$}
			& $\underset{(4e^{-04})}{3.8978}$ & \small{$(4.6,4.8)$} \\
			\hline
		 Example 27&	2D-Call on min & \small{$\boldsymbol{\beta}=(-3,-3), \alpha=15,$\newline $ \delta = 0.2, \boldsymbol{ \Delta= I_2}, K = 100$}
			& $\underset{(2e^{-04})}{1.2635}$ & \small{$(-9.9,-9.9)$} \\
			\hline
	Example	28&	2D-Call on min & \small{$\boldsymbol{\beta}=(-3,0), \alpha=10,$\newline $ \delta = 0.2, \boldsymbol{ \Delta= I_2}, K = 100$}
			& $\underset{(2e^{-04})}{1.4476}$ & \small{$(-7.5,-6.8)$} \\
		\Xhline{7\arrayrulewidth}
	Example	29&	4D-Basket put & \small{$\boldsymbol{\beta}=(-3,-3,-3,-3), \alpha=15,$\newline $ \delta = 0.4, \boldsymbol{ \Delta= I_4}, K = 100$}
			& $\underset{(3e^{-04})}{2.554}$ & \small{$(4.0,4.0, 4.0,4.0)$} \\
			\hline
	Example	30&4D-Basket put  &  \small{$\boldsymbol{\beta}=(-3,-2,-1,0),\alpha=15$,\newline$ \delta = 0.4, \boldsymbol{ \Delta= I_4}, K = 100$}  &  $\underset{(3 e^{-04})}{3.307}$ & \small{$ (4.0,4.2,4.2,4.2) $} 
			 \\
			\hline
	Example	31& 4D-Call on min  &  \small{$\boldsymbol{\beta}=(-3,-3,-3,-3),\alpha=15 $,\newline$ \delta = 0.4, \boldsymbol{ \Delta= I_4}, K = 100$ }& $\underset{( 5 e^{-05})}{0.17374}$& \small{$(-8.8,-8.8,-8.8,-8.8)$}  \\
			\hline
	Example	32&4D-Call on min  &   \small{$\boldsymbol{\beta}=(-3,-2,-1,0),\alpha=15$,\newline$ \delta = 0.4, \boldsymbol{ \Delta= I_4}, K = 100$}  & $\underset{(7e^{-05})}{0.20327} $ &  \small{$(-6.5,-6.4,-6.3,-6.2)$}
			 \\
		\Xhline{7\arrayrulewidth}
	Example	33&6D-Basket put  & \small{$\boldsymbol{\beta}=(-3,-3,-3,-3,-3,-3)$,\newline$ \alpha=15, \delta = 0.2, \boldsymbol{ \Delta= I_6},K=80$} & $\underset{(2  e^{-05})}{0.01039}    $ & \small{$(3.1, 3.1, 3.1, 3.1, 3.1, 3.1) $} 
			\\
			\hline
		Example	34&6D-Basket put   &  \small{$ \boldsymbol{\beta}=(-3,-2,-1,0,1,2)$,\newline$ \alpha=15, \delta = 0.2, \boldsymbol{ \Delta= I_6},K=80$}  &  $ \underset{(3 e^{-06})}{4.39e^{-04} }  $ &  \small{$  (4.5, 4.6, 4.7, 4.8, 4.8, 4.9)  $}
			 \\
			\hline
		Example	35&6D-Call on min  &  \small{$\boldsymbol{\beta}=(-3,-3,-3,-3,-3,-3)$,\newline $ \alpha=15, \delta = 0.2, \boldsymbol{ \Delta= I_6}, K = 110$}
			& $ \underset{(4 e^{-06})}{6.034e^{-05}} $ &  \small{$(-4.0, -4.0, -4.0, -4.0, -4.0, -4.0)$}\\
			\hline
		Example	36&6D-Call on min  & \small{$\boldsymbol{\beta}=(-3,-2,-1,0,1,2),\alpha=15$,  $ \delta = 0.2, \boldsymbol{ \Delta= I_6}, K = 110 $} 
			& $\underset{(2  e^{-06})}{1.572e^{-04}}$    & \small{$(-3.2, -3.2, -3.1, -3.2, -3.2, -3.2) $}\\			\hline
		\end{tabular}
		\caption{Examples  of multi-asset options under the multivariate NIG model. In all examples, $S_0^i=100, i=1,\ldots,d, T=1, r = 0.$ Reference values are computed with MC using $10^9$ samples, with  $95\%$ statistical error estimates reported between parentheses. $\bar{R}$ is rounded  to one decimal place.}
		\label{tab_mnig}
	\end{table}
\FloatBarrier
\subsection{Combining the optimal damping heuristic rule with hierarchical deterministic quadrature methods}
\label{num_quad_comparison_sec}

\FloatBarrier
\subsubsection{Effect of sparsification and dimension-adaptivity}
\label{num_adap_sec}
In this section, we analyze the effect of dimension adaptivity and sparsification on the acceleration of the convergence of the relative quadrature error, $\mathcal{E}_{R}$.  We elaborate on the comparison between the TP, SM, and ASGQ methods when optimal damping parameters are used. Table \ref{cpu_tab}   summarizes  these findings.   Through the numerical experiments,   ASGQ consistently outperformed SM. Moreover, for the  $2$D options, the performance of the ASGQ and TP methods is model-dependent, with ASGQ being the best method for options under the GBM model. For $d=4$, for options under the  GBM and VG models,  ASGQ performs better than TP, which is not the case for options under the NIG model. As for $6$D options,  ASGQ performs better than TP in most cases. These observations confirm that the effect of adaptivity and sparsification becomes more important as the dimension of the option increases.    For the sake of illustration,  Figures \ref{gbm_quad}, \ref{vg_quad}, \ref{nig_quad} compare ASGQ and TP for  $4$D options with anisotropic parameter sets under different pricing models when optimal damping parameters are used. Figure \ref{basket_gbm_4D_aniso_EC} reveals that, for the 4D-basket put option under the GBM model,  the ASGQ method achieves $\mathcal{E_R}$ below $1\%$ using $13.3 \%$ of the work of the TP quadrature. Moreover, 
Figure 	\ref{basket_vg_4D_aniso_EC} indicates that, for the 4D-basket put option under the VG model,  the ASGQ method  achieves $\mathcal{E_R}$ below $0.1\%$ using $25 \%$ of the work of the TP quadrature. In contrast,  for  the 4D-basket put option under the NIG model, Figure \ref{basket_nig_4D_aniso_EC}  reveals that the TP quadrature attains $\mathcal{E_R}$ below $0.1\%$ using $10 \%$ of the work of the ASGQ.
\begin{figure}[h!]
 \centering	
 	\begin{subfigure}{0.4\textwidth}
\includegraphics[width=\linewidth]{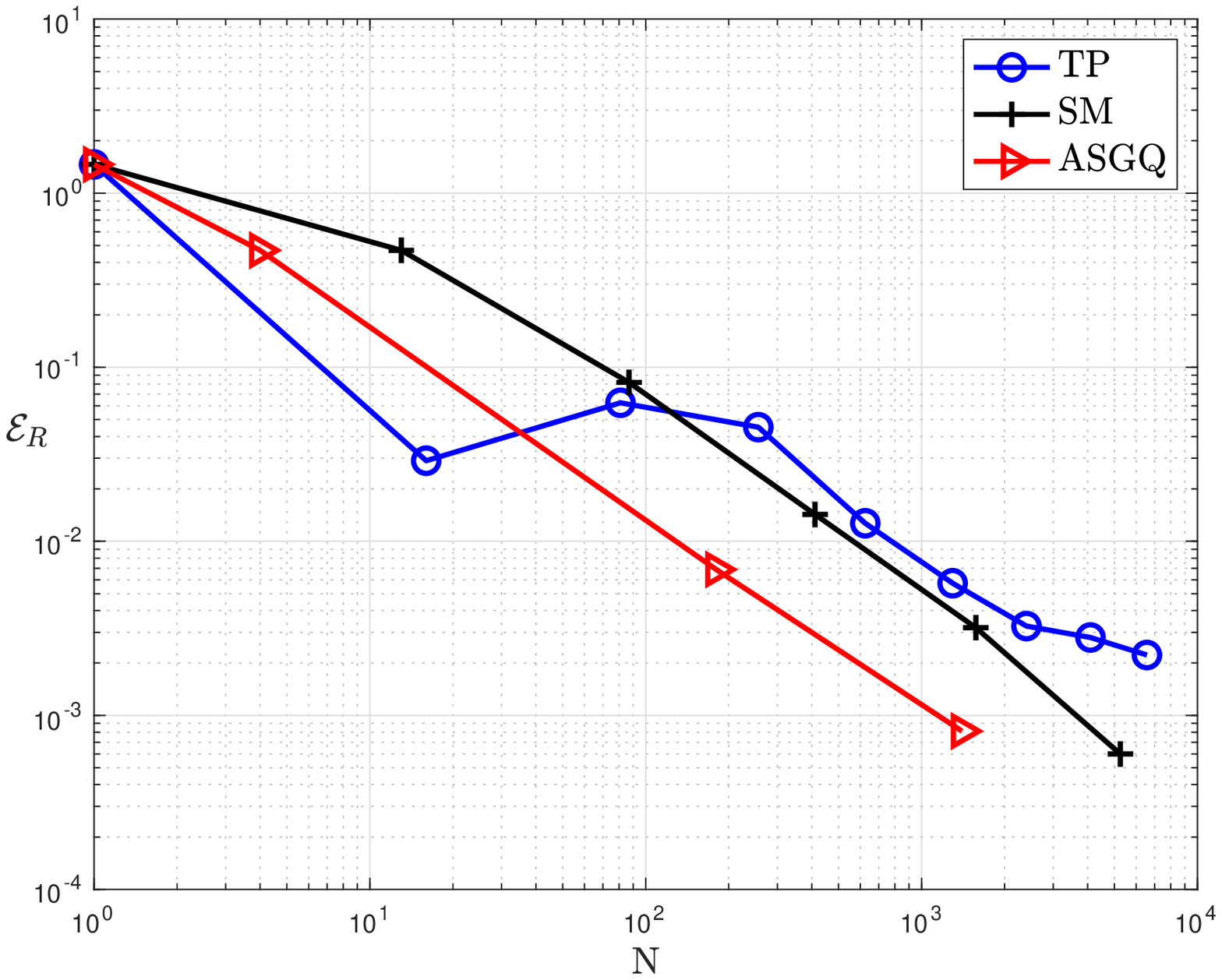}
\caption{Example 6 in Table \ref{tab_mgbm}}
\label{basket_gbm_4D_aniso_EC}
 	\end{subfigure}
  	\begin{subfigure}{0.4\textwidth}
\includegraphics[width=\linewidth]{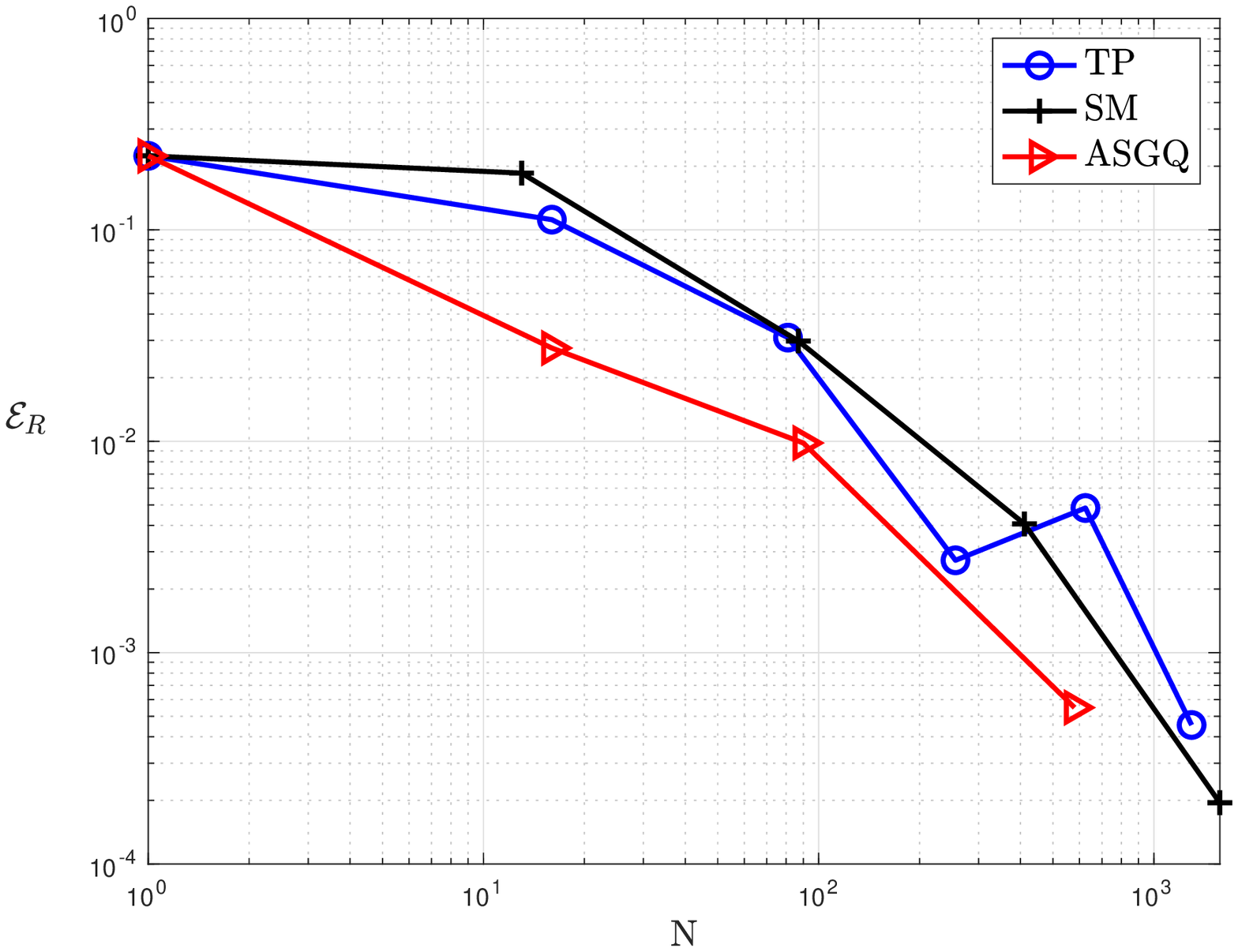}
\caption{Example 8 in Table \ref{tab_mgbm}}
\label{rainbow_gbm_4D_aniso_EC}
 	\end{subfigure}
 \caption{GBM: Convergence of the relative quadrature error, $\mathcal{E}_{R}$, w.r.t.~$N$ for TP, SM and  ASGQ  methods for European $4$-asset options, when  optimal damping parameters, $\mathbf{\overline{R}}$, are used.}
 \label{gbm_quad}
\end{figure}

\begin{figure}[h!]
	\centering	
  	\begin{subfigure}{0.4\textwidth}
 	\includegraphics[width=\linewidth]{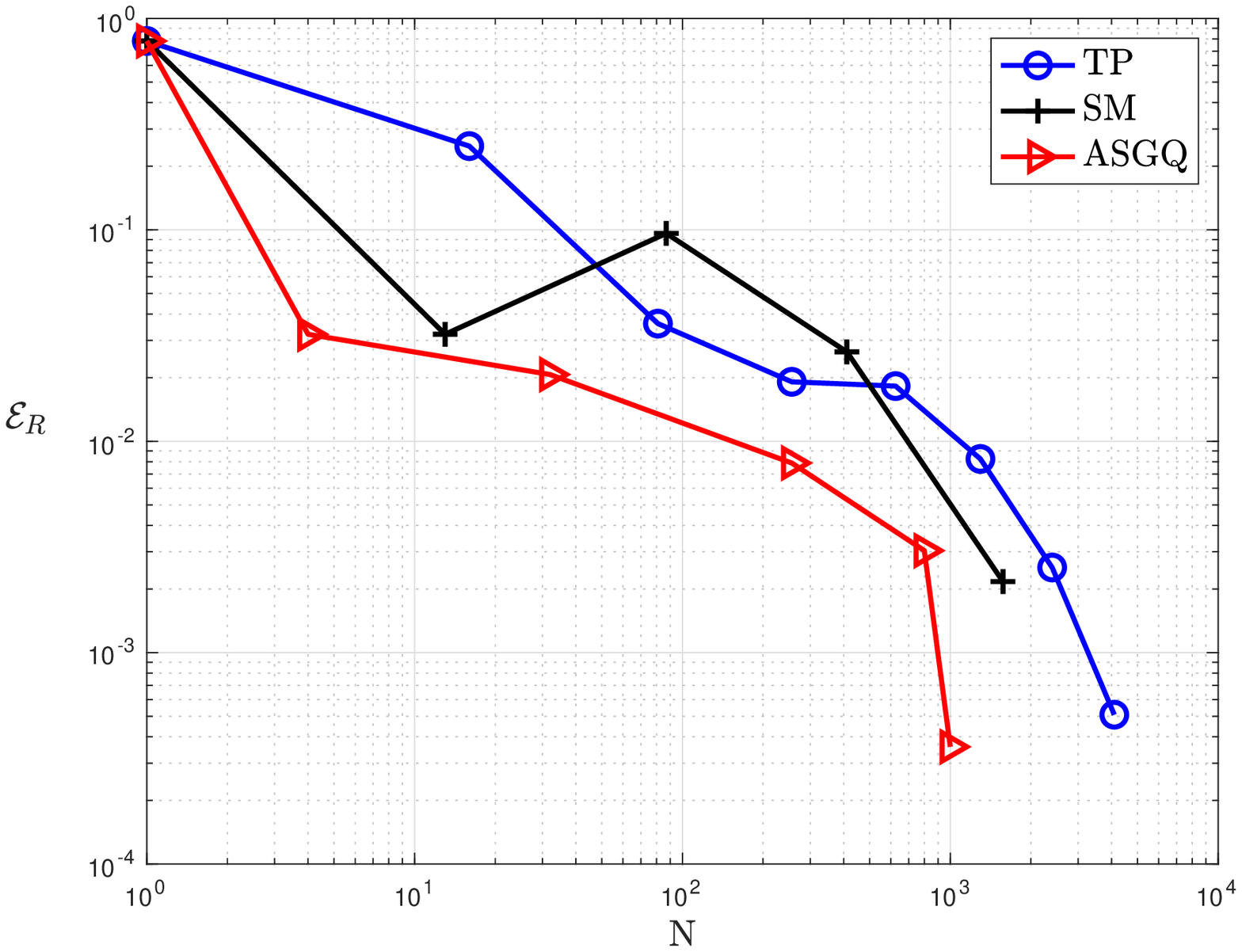}
 	\caption{Example 18 in Table \ref{tab_mvg}}
 	\label{basket_vg_4D_aniso_EC}
 \end{subfigure}
 \begin{subfigure}{0.4\textwidth}
 	\includegraphics[width=\linewidth]{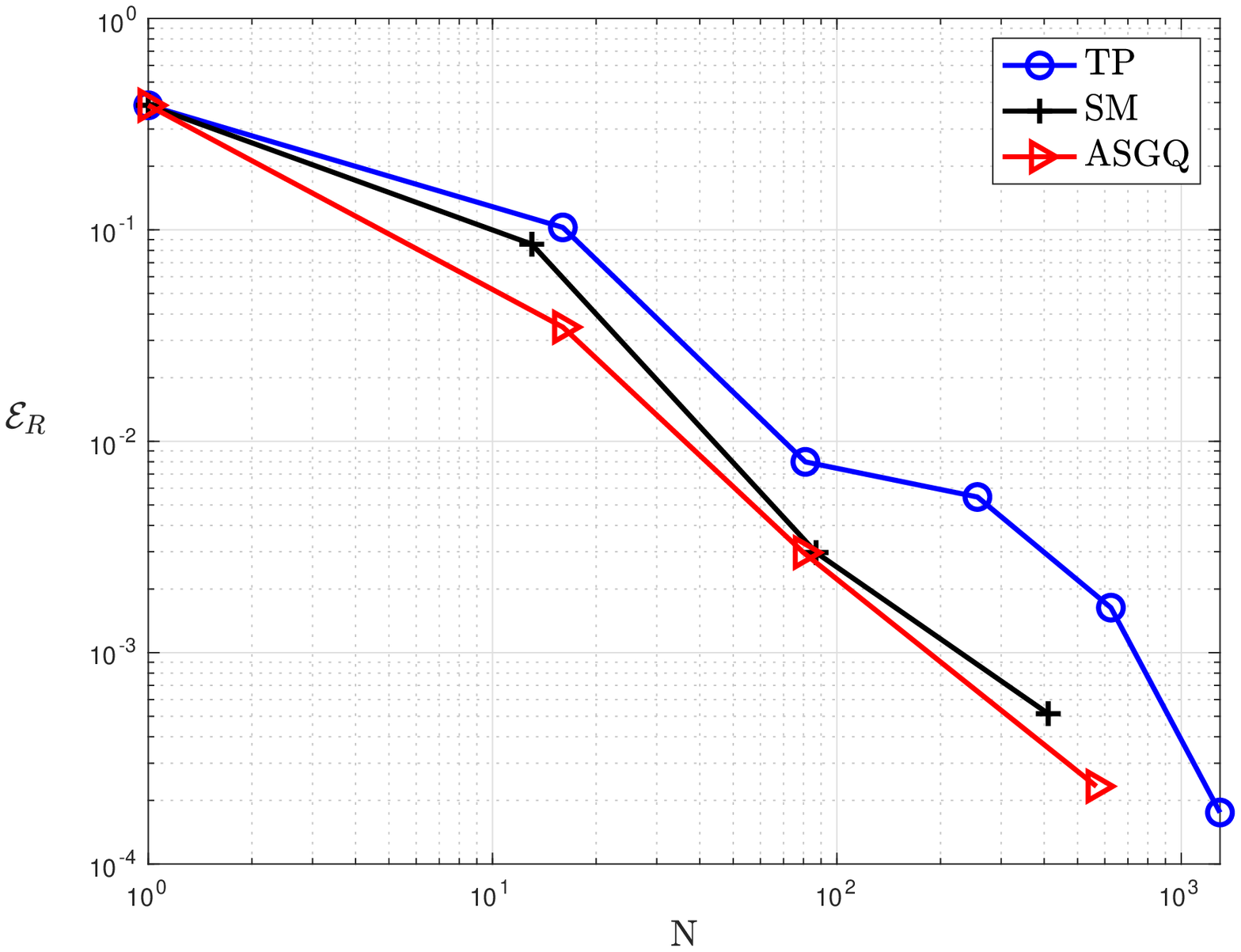}
 	\caption{Example 20 in Table \ref{tab_mvg}}
 	\label{rainbow_vg_4D_aniso_EC}
 \end{subfigure}
\caption{VG: convergence of the relative quadrature error, $\mathcal{E}_{R}$, w.r.t.~$N$ for TP, SM and  ASGQ  methods for European $4$-asset options, when  optimal damping parameters, $\mathbf{\overline{R}}$, are used.}
\label{vg_quad}
\end{figure}

\begin{figure}[h!]
	\centering	
	\begin{subfigure}{0.4\textwidth}
	\includegraphics[width=\linewidth]{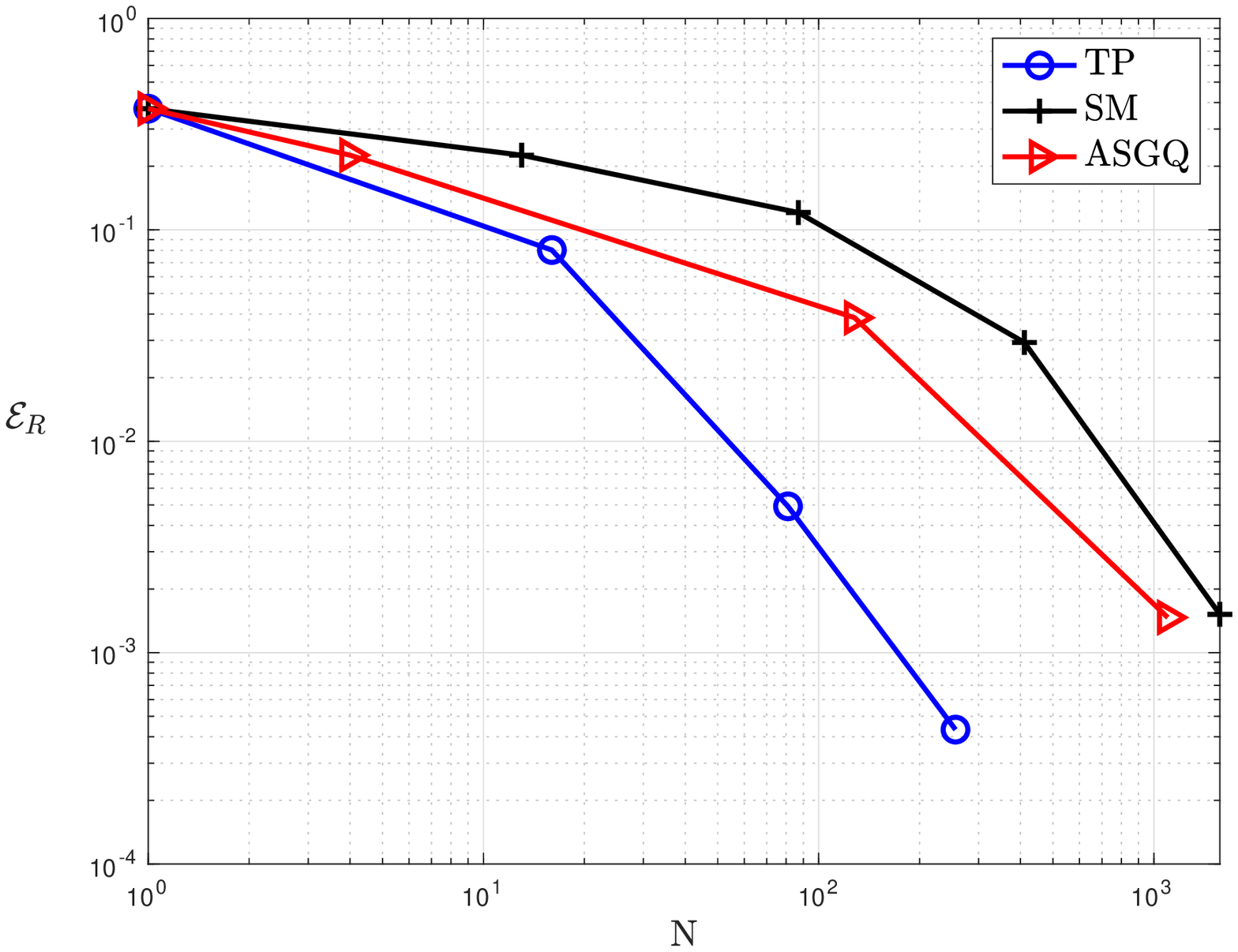}
	\caption{Example 30 in Table \ref{tab_mnig}}
	\label{basket_nig_4D_aniso_EC}
\end{subfigure}
\begin{subfigure}{0.4\textwidth}
	\includegraphics[width=\linewidth]{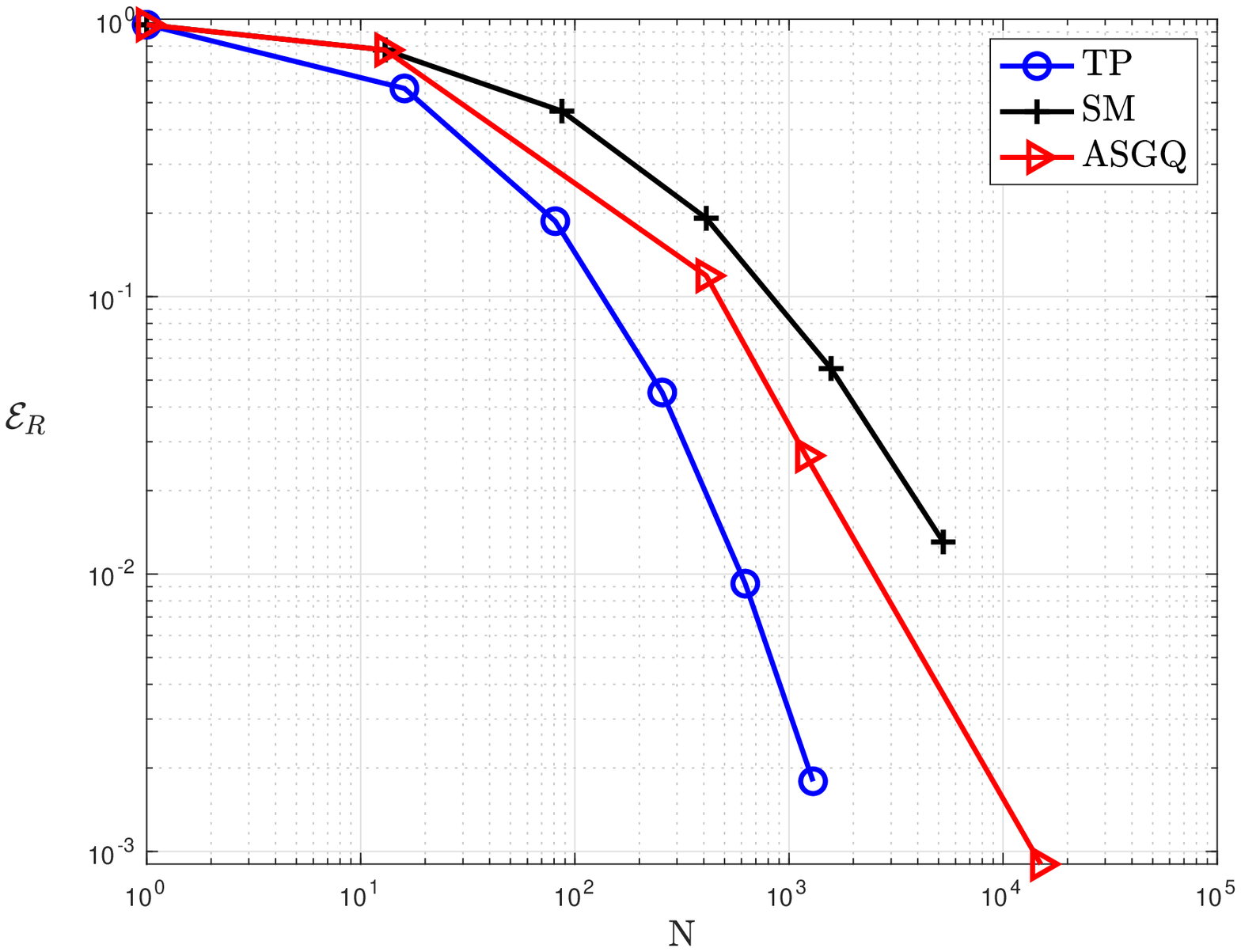}
	\caption{Example 32 in Table \ref{tab_mnig} }
	\label{rainbow_nig_4D_aniso_EC}
\end{subfigure}
\caption{NIG: convergence of the relative quadrature error, $\mathcal{E}_{R}$, w.r.t.~$N$ for TP, SM and  ASGQ  methods for European $4$-asset options, when  optimal damping parameters, $\mathbf{\overline{R}}$, are used.}
\label{nig_quad}
\end{figure}
\FloatBarrier
\subsubsection{Effect of the optimal damping rule}
\label{num_high_dim_damping_sec}
 In this section, we present the computational benefit of using the optimal damping rule  proposed in Section \ref{damping_section}   on the convergence speed of the relative quadrature error of various   methods when pricing the multi-asset European basket and rainbow options. Figures \ref{mgbm_damping}, \ref{mvg_damping}, and \ref{mnig_damping} illustrate that the optimal damping parameters lead to substantially better error convergence behavior.  For instance, Figure \ref{basket_gbm_4D_aniso_damping} reveals that, for  the $4$D-basket put option under the GBM model, ASGQ achieves $\mathcal{E_R}$ below $0.1 \%$ using  around $N=1500$ quadrature points when using optimal damping parameters, compared to around  $N=5000$ points to achieve a similar accuracy for  damping parameters shifted by $+1$ in each direction w.r.t.~the optimal values. When using  damping parameters shifted by $+2$ in each direction w.r.t.~the optimal values, we do not reach $\mathcal{E_R}= 10 \%$, even using $N=5000$ quadrature points.  Similarly, for the  $4$D-call on min  option under the VG model,  Figure  \ref{rainbow_vg_4D_aniso_damping} illustrates that ASGQ achieves $\mathcal{E_R}$ below $0.1 \%$ using  around $N=500$ quadrature points when using the optimal damping parameters.  In contrast, ASGQ cannot achieve $\mathcal{E_R}$ below $1\%$ when using damping parameters shifted by $-1$ in each direction w.r.t.~the optimal values with the same number of quadrature points. Finally,  for the $4$D-basket put option under the NIG model, Figure  \ref{basket_nig_4D_aniso_damping}  illustrates that, when using the optimal damping parameters, the TP quadrature crosses  $\mathcal{E_R}= 0.1 \%$ using $22 \%$ of the work it would have used with damping parameters shifted by $-2$  in each direction w.r.t.~the optimal values.   
 
 	In summary,  in all experiments,  small shifts in both directions  w.r.t.~the optimal damping parameters  lead to worse  error convergence behavior, suggesting that the region of optimality of the damping parameters is tight and  that our rule is  sufficient to obtain optimal quadrature convergence behavior, independently of the  method. Moreover,  arbitrary choices of damping parameters may lead to extremely poor convergence of the quadrature,  as illustrated by the   purple  curves in  Figures \ref{basket_gbm_4D_aniso_damping},\ref{rainbow_gbm_4D_aniso_damping},  \ref{basket_vg_4D_aniso_damping} and \ref{rainbow_nig_4D_aniso_damping}.  All compared damping parameters belong to the strip of regularity of the integrand $\delta_V$ defined in Section \ref{sec:Problem Setting and Pricing Framework}. Finally,  although we only provide  some plots   to illustrate these findings,  the same conclusions were  consistently observed for 
 	different models  and damping parameters.
 	\vspace{-0.2cm}
\begin{figure}[h!]
 \centering	
 	\begin{subfigure}{0.4\textwidth}
\includegraphics[width=\linewidth]{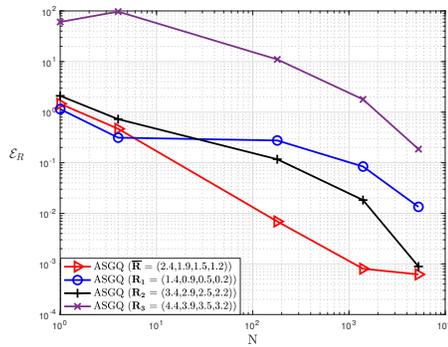}
 	\caption{Example 6 in Table \ref{tab_mgbm}}
\label{basket_gbm_4D_aniso_damping}
\end{subfigure}
 	\begin{subfigure}{0.4\textwidth}
\includegraphics[width=\linewidth]{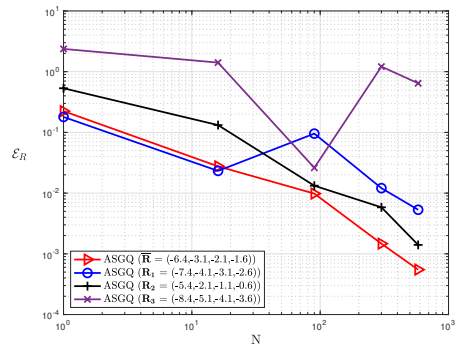}
\caption{Example 8 in Table \ref{tab_mgbm}}
\label{rainbow_gbm_4D_aniso_damping}
\end{subfigure}
	\caption{GBM: convergence of the relative quadrature error, $\mathcal{E}_{R}$, w.r.t.~$N$ for   the ASGQ  method  for different   damping parameter values. } 
	\label{mgbm_damping}
\end{figure}
\begin{figure}[h!]
 \centering	
 	\begin{subfigure}{0.4\textwidth}
\includegraphics[width=\linewidth]{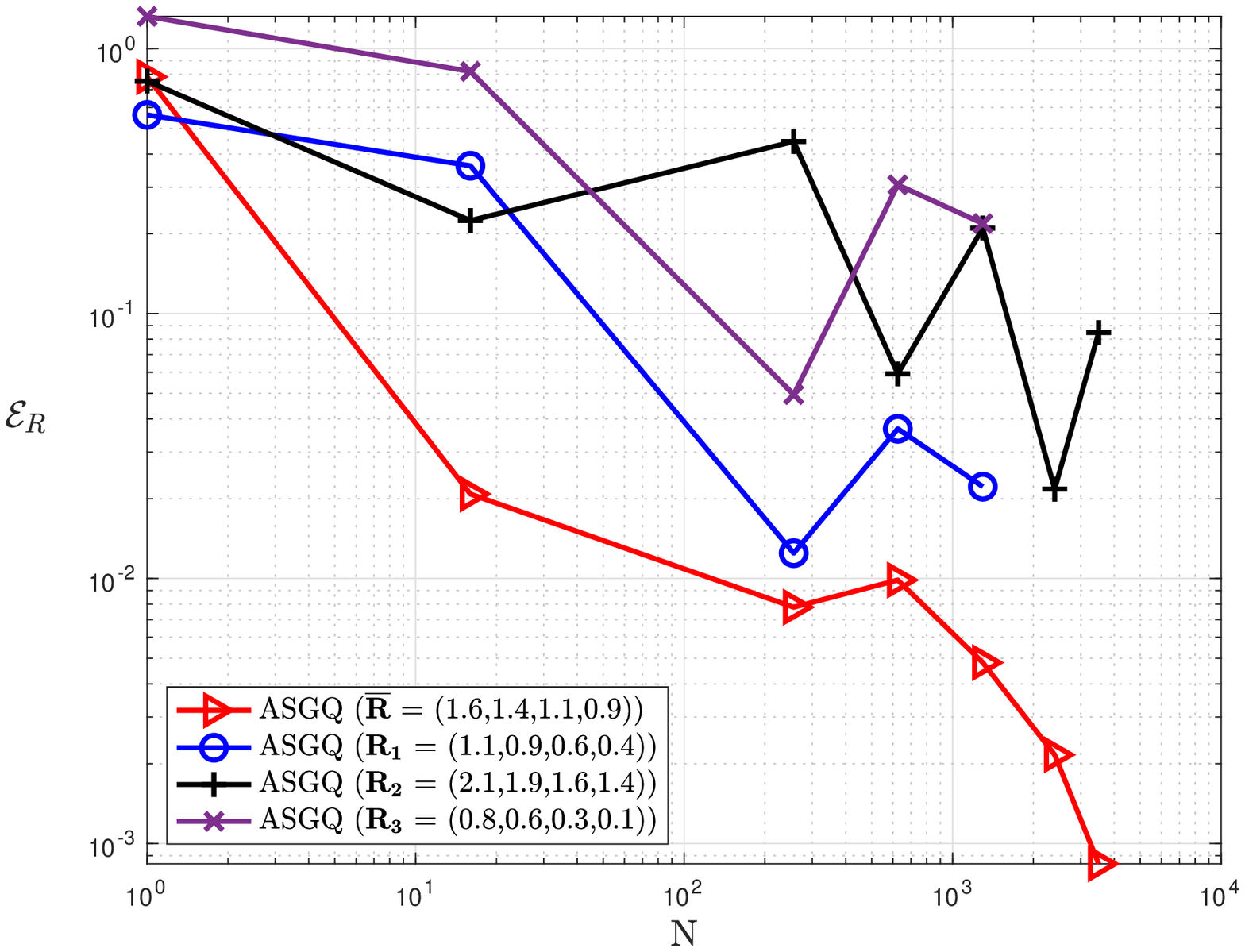}
 	\caption{Example 18 in Table \ref{tab_mvg}}
\label{basket_vg_4D_aniso_damping}
\end{subfigure}
 	\begin{subfigure}{0.4\textwidth}
\includegraphics[width=\linewidth]{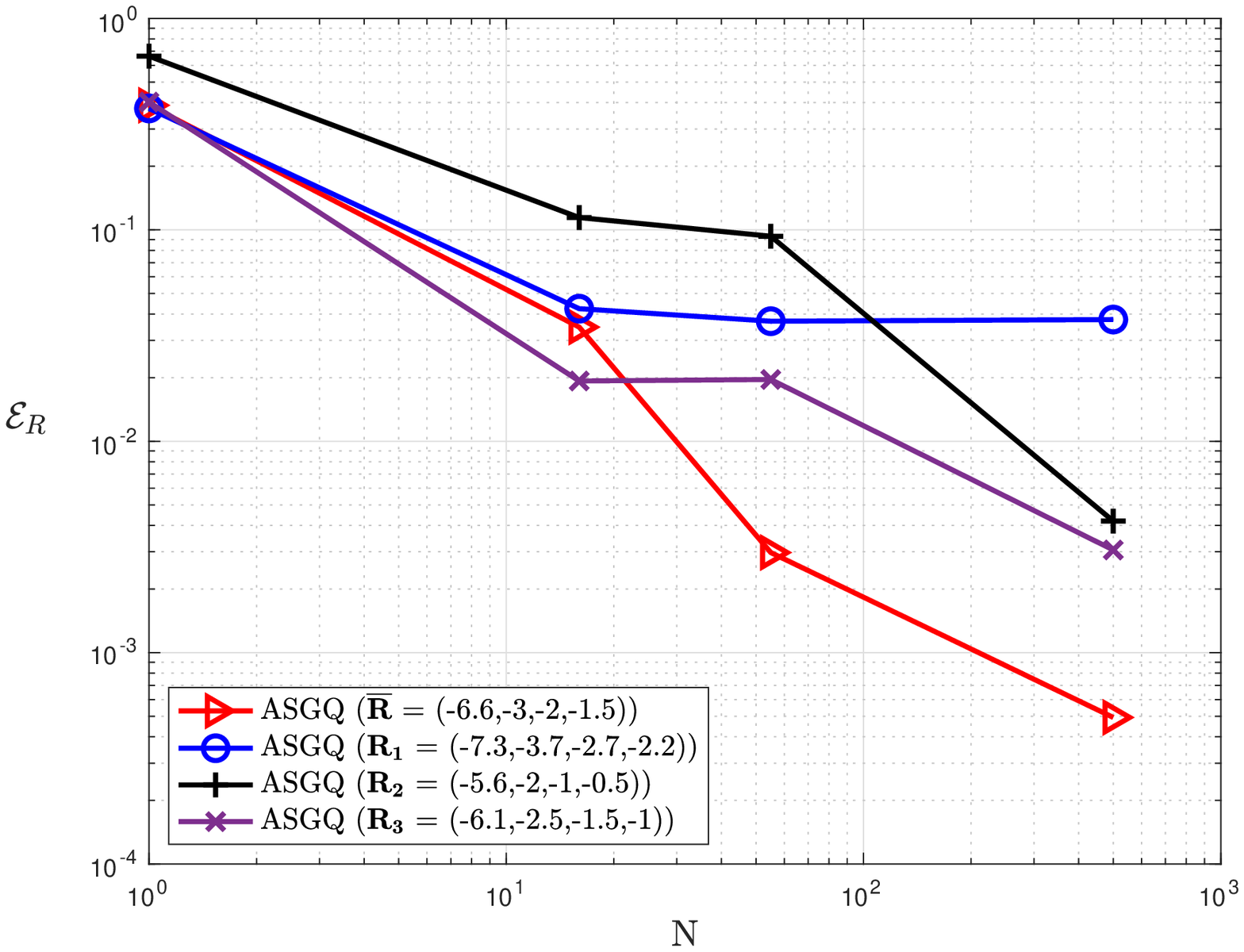}
 	\caption{Example 20 in Table \ref{tab_mvg}}
 	\label{rainbow_vg_4D_aniso_damping}
\end{subfigure}
	\caption{VG: convergence of the relative quadrature error, $\mathcal{E}_{R}$, w.r.t.~$N$ for   the ASGQ  method  for different damping parameter values.} 
		\label{mvg_damping}
\end{figure}
\begin{figure}[h!]
 \centering	
  	\begin{subfigure}{0.4\textwidth}
\includegraphics[width=\linewidth]{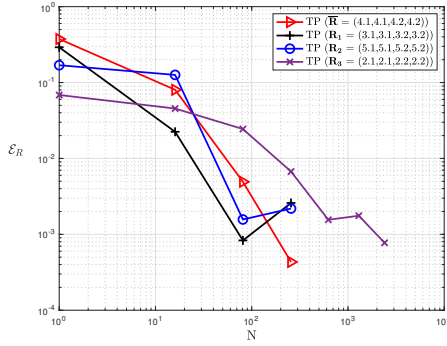}
 	\caption{Example 30 in Table \ref{tab_mnig}}
 	\label{basket_nig_4D_aniso_damping}
\end{subfigure}
  	\begin{subfigure}{0.4\textwidth}
\includegraphics[width=\linewidth]{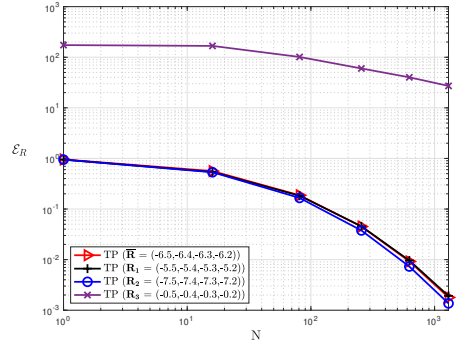}
 	\caption{Example 32 in Table \ref{tab_mnig} }
 	\label{rainbow_nig_4D_aniso_damping}
\end{subfigure}
	\caption{NIG:  convergence of the relative quadrature error, $\mathcal{E}_{R}$, w.r.t.~$N$ for   the TP  method  for different   damping parameter values.} 
		\label{mnig_damping}
\end{figure}
\FloatBarrier
\subsection{Comparison of our approach  to the COS Method}
\label{ODHAQ_vs_COS}
This section presents an empirical comparison of our proposed approach  with the COS method \cite{fang2008novel}. Our aim is to compare the performance of the two methods when both of them are appropriately tuned, which seems to be lacking in benchmarking works in the literature \cite{von2015benchop,crisostomo2018speed} due to the absence of guidance on the suitable choice of the damping parameters. We use two metrics for the comparison of the approaches (i) the CPU time required to achieve a pre-defined relative error, $\mathcal{E}_R = 1e^{-03}$, (see Tables \ref{tab:cpu_cos_vs_odhaq_2DCOM}, \ref{tab:cpu_cos_vs_odhaq_2DPUT}) and (ii) the number of times the characteristic function is evaluated, $\text{N}_{CF}$, to reach this accuracy (see Figures \ref{fig:ODTPQ_vs_COS_2D_VG}, \ref{fig:ODTPQ_vs_COS_2D_NIG}). The second metric is particularly interesting when the costly part of the approximation formula is the evaluation of the characteristic function, and has the merit of being independent of the implementation and the used computer characteristics. 
In this work, the numerical comparison of our approach with the COS method is restricted for options with up to two underlyings because the implementation of the COS method in higher dimensions is not available to us. Since the CPU time needed to achieve a certain accuracy highly depends on the way the methods are implemented, we did not use the sparse grids kit \cite{piazzola.tamellini:SGK} for the comparison. To the extent possible, both of the methods were implemented in similar style to have reproducible results. For the sake of fair comparison, we compare the optimal damping rule with isotropic TP quadrature and Gauss-Laguerre rule, which we denote ODTPQ for short, with the isotropic version of the COS method. The reported CPU times in Tables \ref{tab:cpu_cos_vs_odhaq_2DCOM},  \ref{tab:cpu_cos_vs_odhaq_2DPUT} are given in seconds, and are computed from the average over $10^3$ replications of each experiment.  Moreover, the ODTPQ CPU time in Tables \ref{tab:cpu_cos_vs_odhaq_2DCOM},  \ref{tab:cpu_cos_vs_odhaq_2DPUT} includes both the cost of the quadrature and the cost of the optimization to obtain the damping parameters, $\overline{\mathbf{R}}$. Reference values are computed using MC with $M = 10^9$ samples.
\subsubsection{Implementation details of the COS method}
The 2D-COS formula to approximate the option value is given by \cite{ruijter2012two}
\begin{equation}
	\begin{aligned}
	V(\boldsymbol{\Theta}_m, \boldsymbol{\Theta}_p) &\approx e^{-rT} \frac{(b-a)^2}{4} \sum_{k_1 = 0}^{N_{COS}}  \sum_{k_2 = 0}^{N_{COS}} \frac{1}{2} \biggl( \Re \biggl\{   \phi_{\boldsymbol{X}_T}\left( \frac{k_1 \pi}{ b-a},  \frac{k_2 \pi}{b-a} \right) e^{\mathrm{i} k_1 \pi \frac{X_0^1 - a}{b - a} + \mathrm{i} k_2 \pi \frac{X_0^2 - a}{b - a}  }  \biggl\}  \\
	&+   \Re \biggl\{ \phi_{\boldsymbol{X}_T}\left( \frac{k_1 \pi}{ b-a},  -\frac{k_2 \pi}{b-a}  \right)    e^{\mathrm{i} k_1 \pi \frac{X_0^1 - a}{b - a} - \mathrm{i} k_2 \pi \frac{X_0^2 - a}{b - a}  }  \biggl\} \biggl) P_{k_1, k_2}(T),
	\end{aligned}
\end{equation}
where $\Phi_{\boldsymbol{X}_T}(\mathbf{u}) =e^{\mathrm{i} \langle \mathbf{u}, \boldsymbol{X}_0 \rangle} \phi_ {\boldsymbol{X}_T}(\mathbf{u})$ is the characteristic function (see Table \ref{table:chf_table}), and $P_{k_1,k_2}$ are the Fourier cosine coefficients of the payoff function $P(\cdot)$. We use the isotropic version  where the number of Fourier modes is the same in each dimension, $N_1 = N_2 = N_{COS}$.  For the truncation range, we use the domain $[a,b]^d$ as suggested in Section 5.2 in \cite{ruijter2012two}, which is given by 
\begin{equation}
	\begin{aligned}
		a &= \min_{1 \leq i \leq 2} \biggl\{ X_0^i +  c^i_1-L \sqrt{c^i_2+\sqrt{c^i_4}} \biggl\}\\
	 b &=  \max_{1 \leq i \leq 2}  \biggl\{  X_0^i +c^i_1+L \sqrt{c^i_2+\sqrt{c^i_4}}  \biggl\}
\end{aligned}
\end{equation}
where $c_n^i$ is the $n$th cumuant of the random variable $X_T^i$ and $L=10$. For the table of cumulants used in this work, we refer to \cite{oosterlee2019mathematical}. 
In the case of 2D-basket put and 2D-call on min options, we approximate $\{P_{k_1, k_2}\}_{k_1, k_2 = 0}^{N_{COS}-1}$ numerically using discrete cosine transform (dct2 function in MATLAB). The number of terms used in each spatial dimension of the DCT approximation is denoted by $Q$ (see Section 3.2.1 in \cite{ruijter2012two} for more details). 	To the best of our knowledge, there is no rule for the choice of $Q$. In \cite{ruijter2012two}, authors solely state that the number of terms used in the DCT to approximate the payoff cosine coefficients, $Q$, must satisfy $ Q \geq \max(N_1,N_2) = N_{COS}$. We observed that the choice $Q = N_{COS}$,  may result in oscillatory behavior of the error convergence. For this reason, we used $Q = 1000$ to ensure that the payoff cosine coefficients are approximated accurately, as shown in \cite{ruijter2012two}.  We note that $\text{N}_{CF} = 2^{d-1} N_{COS}^d$ in $d$ dimensions.
\FloatBarrier
\subsubsection{Numerical experiments}
Figures \ref{fig:ODTPQ_vs_COS_2D_VG} and \ref{fig:ODTPQ_vs_COS_2D_VG} show that if the damping parameters are chosen appropriately based on the proposed rule in \eqref{or_opt}, our approach achieves  a desired relative error with significantly less characteristic functions evaluations, $\text{N}_{CF}$ , for both tested 2D-call on min and 2D-basket put options under VG and NIG. For instance, Figure \ref{fig:ODTPQ_vs_COS_2D_VG_COM} shows that the COS method achieves $\mathcal{E}_R = 1e^{-03}$ using $\text{N}_{CF} = 25538$, whereas  ODTPQ reaches $\mathcal{E}_R = 6 e^{-04}$ using only $\text{N}_{CF} = 108$.
\FloatBarrier
\begin{figure}[h!]	\centering		
		\begin{subfigure}{0.4\textwidth}		
		\includegraphics[width=\linewidth]{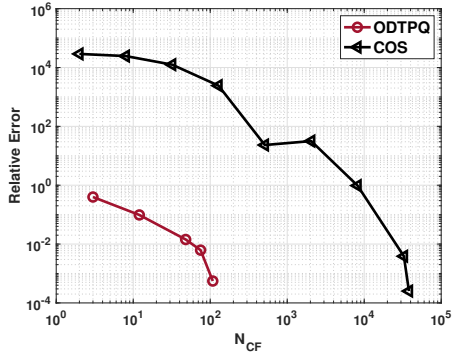}	
		\caption{2D-call on min}
\label{fig:ODTPQ_vs_COS_2D_VG_COM}
	\end{subfigure}
	\begin{subfigure}{0.4\textwidth}		\includegraphics[width=\linewidth]{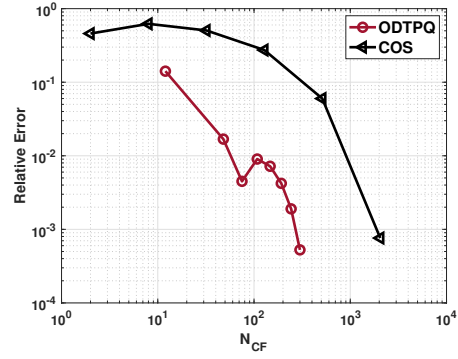}		\caption{2D-basket put}	
	\end{subfigure}	
\caption{Convergence of the  relative error w.r.t. $\text{N}_{CF}$, the number of characteristic function evaluations of both COS and ODTPQ under the VG model. The used model, payoff and damping parameters are given in Tables \ref{tab:cpu_cos_vs_odhaq_2DCOM}, \ref{tab:cpu_cos_vs_odhaq_2DPUT}.	} 	\label{fig:ODTPQ_vs_COS_2D_VG}
\end{figure}

\FloatBarrier
\begin{figure}[h!]	\centering		
	\begin{subfigure}{0.4\textwidth}		
		\includegraphics[width=\linewidth]{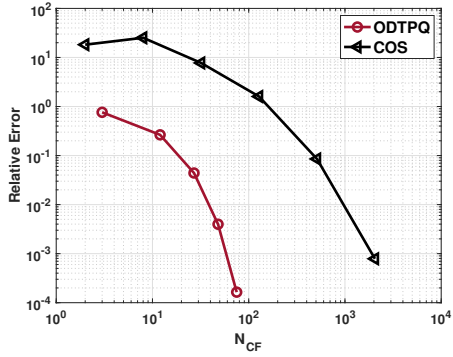}	
		\caption{2D-call on min}
	
		\end{subfigure}
		\begin{subfigure}{0.4\textwidth}		\includegraphics[width=\linewidth]{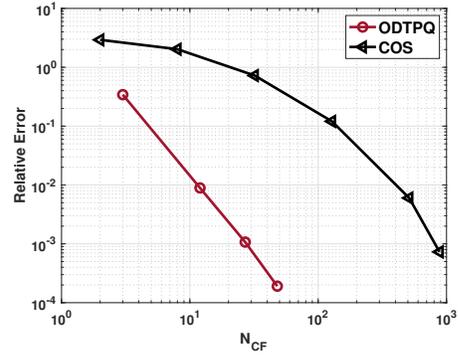}		\caption{2D-basket put}	
		\end{subfigure}	
	\caption{Convergence of the  relative error w.r.t. $\text{N}_{CF}$, the number of characteristic function evaluations of both COS and ODTPQ under the NIG model. The used model, payoff and damping parameters are given in Tables \ref{tab:cpu_cos_vs_odhaq_2DCOM}, \ref{tab:cpu_cos_vs_odhaq_2DPUT}.	} 	\label{fig:ODTPQ_vs_COS_2D_NIG}
	\end{figure}

	
	Table \ref{tab:cpu_cos_vs_odhaq_2DCOM} demonstrates that ODTPQ approach achieves the relative error, $\mathcal{E}_R = 1e^{-03}$, approximately from 13-29.5 times faster than the COS method for the tested 2D-call on min option under GBM, VG and NIG. Moreoever, Table \ref{tab:cpu_cos_vs_odhaq_2DPUT}  shows that  ODTPQ reaches the relative tolerance,  $\mathcal{E}_R = 1e^{-03}$, approximately from 2.5-6.5 times faster than the COS method for 2D-basket put options under GBM, VG and NIG. 
		\begin{table}[h]
		\vspace{-0.4cm}
		\centering
		\begin{tabular}{| p{1.7cm} |  p{4.2cm}  | p{2cm}  | p{2.6cm}  | p{2.6cm} |}
			\hline   \textbf{Model} &  \textbf{Parameters} & $\mathbf{\overline{R}}$ &  \textbf{ODTPQ CPU time} &  \textbf{COS CPU time} \\
			\hline
			\textbf{GBM} &  $\boldsymbol{\sigma} = (0.2,0.8), \boldsymbol{C} = \boldsymbol{I}_d$ &   -(7.18, 1.65)& $ 2.4 e^{-03}$ & $4.2 e^{-02} $ \\
			\hline 
			\textbf{VG} & $\boldsymbol{\sigma} = (0.2,0.8), \newline \boldsymbol{\theta} = (-0.3,-0.1), \nu = 0.5$  & -(7.38, 1.79) & $2.2 e^{-03}  $ & $6.5 e^{-02}  $ \\
			\hline
			\textbf{NIG} & $\boldsymbol{\beta} = (-3,-3), \boldsymbol{\Delta} = \boldsymbol{I}_d, \newline \alpha = 15, \delta = 0.5$ & -(6.88, 6.88) & $ 1.8 e^{-03}  $ & $2.36 e^{-02}   $ \\
			\hline
		\end{tabular}
		\caption{2D-call on min: CPU time in seconds of ODTPQ and COS to achieve relative error $\mathcal{E}_R = 1e^{-03}$ with $\boldsymbol{S}_0 = (100,100), K = 100, r = 0, T = 1$. }
		\label{tab:cpu_cos_vs_odhaq_2DCOM}
	\end{table}
	\FloatBarrier
	
	\begin{table}[h]
	\centering
	\begin{tabular}{| p{1.7cm} |  p{4.2cm}  | p{1.9cm}  |  p{2.6cm} | p{2.6cm}  |}
		\hline   \textbf{Model} &  \textbf{Parameters} & $\mathbf{\overline{R}}$&  \textbf{ODTPQ CPU time} &  \textbf{COS CPU time} \\
		\hline
		\textbf{GBM} &  $\boldsymbol{\sigma} = (0.2,0.8), \boldsymbol{C} = \boldsymbol{I}_d$ & $(3.05, 1.36)$ &$5.4 e^{-03} $ & $2.7 e^{-02}$ \\
		\hline 
		\textbf{VG} & $\boldsymbol{\sigma} = (0.2,0.8), \newline \boldsymbol{\theta} = (-0.3,-0.1), \nu = 0.5$  & (1.81, 0.89) & $9.1 e^{-03} $ & $2.5 e^{-02} $ \\
		\hline
		\textbf{NIG} & $\boldsymbol{\beta} = (-3,-3),  \boldsymbol{\Delta} = \boldsymbol{I}_d, \newline \alpha = 15, \delta = 0.5$ & (4.5, 4.5) &  $3.6 e^{-03}  $  & $2.4 e^{-02} $ \\
		\hline
	\end{tabular}
	\caption{2D-basket put:  CPU time in seconds of ODTPQ and COS to achieve relative error $\mathcal{E}_R = 1e^{-03}$ with $\boldsymbol{S}_0 = (100,100), K = 100, r = 0, T = 1$.}
\label{tab:cpu_cos_vs_odhaq_2DPUT}
\end{table}
\FloatBarrier

\begin{remark}[About the COS method in multiple dimensions]
	For insights on the performance of the COS method in more than two dimensions, we refer to \cite{junike2023multidimensional}, where the numerical experiments indicate that the MC method outperforms the COS method for cash-or-nothing put option with more than 3 underlyings if the target error tolerance is of order $1e^{-03}$. 
\end{remark}

\subsection{Computational comparison of  quadrature  methods with  optimal damping and MC}
\label{num_quad_vs_mc_sec}
 This section compares the MC method and our proposed approach based on  on the best quadrature method in the Fourier space combined with the optimal damping parameters in terms of errors and computational time. The comparison is performed for all option examples  in Tables \ref{tab_mgbm}, \ref{tab_mvg}, and \ref{tab_mnig}. While fixing a sufficiently small relative error tolerance in the price estimates,  we compare the necessary computational time for different methods to meet it  in the following way:
 \begin{enumerate}
 \item Find the least number of quadrature points to reach a pre-defined relative quadrature error.
 \item Estimate, using the CLT formula given in Equation (\ref{CLT_formula}), the required number of MC samples to achieve the same relative error achieved by the quadrature method.
 \item Compare the CPU times of the both methods, including the cost of numerical optimization of \eqref{new_opt} preceding the numerical quadrature for the Fourier approach.  The  MC  CPU time is obtained through an average of $10$ runs.
  \end{enumerate}

 The results presented in Table \ref{cpu_tab} highlight that our approach significantly outperforms the MC method for all the tested options with various models, parameter sets, and dimensions. In particular, for all tested $2$D and $4$D  options,  the proposed approach requires less than $20\%$ (even less than $1\%$ for most cases)  of the MC work to achieve a total relative error below $0.1\%$.  In general, these gains  degrade for the tested $6$D options.  For Example $21$ in Table \ref{tab_mvg}, this approach requires around  $43\%$  of the work of MC,  to achieve a total relative error below $1\%$.  The magnitude of the CPU gain varies depending on different factors, such as the model and payoff parameters affecting the integrand differently in physical space (related to the MC estimator variance), and the integrand regularity in Fourier space (related to the quadrature error for quadrature methods). 
Finally, numerical experiments suggest that the  advantage of employing ASGQ over TP is more pronounced when pricing of options with dimension higher than two, except for the NIG model, where TP quadrature performs exceptionally well even in the 4D case. Nevertheless, empirical results demonstrate that in the 6D case or higher, it is  recommended to use the ASGQ over TP  for all pricing models. 
\begin{table}[h]
	\centering
		\hspace{-0.4cm}
		\small
	\begin{tabular}{|p{3.6cm}| p{1.2cm}|  p{1.4cm} | p{1.2cm} | p{1.6cm} | 
	p{1.cm} |p{1.2cm} |
	p{3cm} |}
	
	\hline  \textbf{Example} & \textbf{Best} \textbf{Quad}& $\mathcal{E}_R$  & \textbf{MC CPU Time} & \textbf{$M$ (MC samples)}
	& \textbf{Quad CPU Time}  & \textbf{$N$ (Quad. Points) }
	& \textbf{CPU Time Ratio (Quad/MC) in $\%$}  \\
		\hline
Example 1 in Table \ref{tab_mgbm} & ASGQ & $ 7 e^{-04}$  &  $7.36$ & $ 1.2 \times 10^{7}$   
	& $0.63$ & $33$
	& $8.5 \%$\\
	 	\hline	
Example 2 in Table \ref{tab_mgbm} & ASGQ & $ 3.7 e^{-04}$  &  $20.7$ & $3.3 \times 10^{7}$   
	& 0.65 & 67
	& $3.14 \%$\\
	 	\hline
Example 13 in Table \ref{tab_mvg}& TP & $ 2.9 e^{-04}$  &  $44$ & $ 8.8 \times 10^{7}$   
	& 0.25 & 64
	& $0.57 \%$\\
	 	\hline	
Example 14 in Table \ref{tab_mvg}& TP & $ 1.8 e^{-04}$  &  $70.9$ & $1.4\times 10^{8}$   
	& 0.23 & 64
	& $0.32 \%$\\
	 	\hline
Example 25 in Table \ref{tab_mnig} & TP & $ 2.9 e^{-04}$  &  $75.3$ & $1.1 \times 10^{8}$   
	& 0.2 & 36
	& $0.26 \%$\\
	 	\hline
Example 26 in Table \ref{tab_mnig} & TP & $ 5.86 e^{-04}$  &  $17.2$ & $2.6\times 10^{7}$   
	&$ 0.2$ & 25
	& $1.16 \%$\\
	 	\hline
Example 3 in Table \ref{tab_mgbm} & ASGQ & $ 7 e^{-04}$  &  $47.3$ & $ 7.6\times 10^{7}$   
	& 0.6 & 37
	& $1.26 \%$\\
	 	\hline	
Example 4 in Table \ref{tab_mgbm}  & ASGQ & $ 5.8 e^{-04}$  &  $102$ & $ 1.4\times 10^{8}$   
	& 0.63 & 37
	& $0.62 \%$\\
	 	\hline
Example 15 in Table \ref{tab_mvg} & ASGQ & $ 8.26 e^{-04}$  &  $19.5$ & $ 4.1\times10^{7}$   
	& 0.54 & 25
	& $2.77 \%$\\
	 	\hline	
Example 16 in Table \ref{tab_mvg} & TP & $ 5.37 e^{-04}$  &  $87.1$ & $ 1.4\times 10^{8}$   
	& 0.16 & 49
	& $0.18 \%$\\
		\hline
Example 26 in Table \ref{tab_mnig} & TP & $ 6.7 e^{-04}$  &  $35.8$ & $5.3\times 10^{7}$   
	& 0.22 & 100
	& $0.61 \%$\\
	 	\hline
Example 27 in Table \ref{tab_mnig}  & TP & $ 6.46 e^{-04}$  &  $42.2$ & $6.5\times 10^{7}$   
	& 0.22 & 64
	& $0.52 \%$\\
\Xhline{7\arrayrulewidth}
Example 5 in Table \ref{tab_mgbm} & ASGQ & $ 2.46e^{-04}$  &  $207$ & $ 10^8$   
	& 7.8 & 5257
	& $3.77\%$\\
	 	\hline
Example 6 in Table \ref{tab_mgbm}&ASGQ & $8.12 e^{-04}$ &  $14.5$ & $ 7.9 \times 10^6$ 
	& 2.73 & 1433
	& $18.83\%$ \\
	 	\hline
Example 17 in Table \ref{tab_mvg}& ASGQ & ${2.58 e^{-04}}$ &   $106.3$ & $1.23 \times 10^8$ 
	& 5  & 3013
	& $4.7\%$\\
	 	\hline
Example 18 in Table \ref{tab_mvg}& ASGQ & $3.58 e^{-04}$  &  $38.7$  & $4.5 \times10^7$  
	& 2 & 1109 
	& $5.17 \%$ \\
	 	\hline
Example 27 in Table \ref{tab_mnig}& TP &   $4.57 e^{-04}$& $50.2$ & $4.7\times10^7$ 
& 0.5 & 256 
	 & $1 \%$\\
	 	\hline
Example 28 in Table \ref{tab_mnig}&	 TP & ${4.1 e^{-04}}$&   $49.4$ & $4.8 \times 10^7$ 
	& 0.52 & 256
	& $1\%$ \\
	 	\hline
Example 7 in Table \ref{tab_mgbm}& ASGQ & ${ 5.7e^{-04}}$  &  $1147$ & $7 \times10^8$  
	& 1 & 435
	& $0.09 \%$\\
	 	\hline
Example 8 in Table \ref{tab_mgbm}  &ASGQ & ${5.5 e^{-04}}$ &  ${1580}$ & $ 9.6 \times 10^8$  	& 0.95 &654 & $0.06\%$ \\
	 	\hline
Example 19 in Table \ref{tab_mvg} & ASGQ & ${5.9 e^{-04}}$ &   $220$ & $3 \times10^{8}$ 
	& 1.25  & 567
	& $0.57\%$\\
	 	\hline
Example 20 in Table \ref{tab_mvg}& ASGQ & ${8.9 e^{-04}}$  &  249  & $3.3 \times 10^8$  
	& 1.4 & 862
	& $0.56 \%$ \\
	 	\hline
Example 29 in Table \ref{tab_mnig} & TP &   ${7.2 e^{-04}}$& 193.5 & $2\times 10^8$
	 & 8.7 & 20736
	 & $4.5 \%$\\
	 	\hline
Example 30 in Table \ref{tab_mnig} &   TP & ${4.2e^{-04}}$&   716 & $7.8 \times 10^8$ 
	& 0.8 & 2401
	& $0.11\%$ \\
	 	\hline
\Xhline{7\arrayrulewidth}
Example 9 in Table \ref{tab_mgbm} & ASGQ & $2.9 e^{-02}$ &  18.53 & $5.5 \times 10^6$ 
	& 2  & 318
	& $11\%$\\
	 	\hline
Example 10 in Table \ref{tab_mgbm}  & ASGQ & $3.3e^{-03}$ &  548 & $1.5 \times 10^8 $ 
	& 2.1  & 340
	& $0.38\%$ \\
	 	\hline
Example 21 in Table \ref{tab_mvg} &	 ASGQ & ${7.8e^{-03}}$ &  5.4 & $ 4.7 \times 10^6 $  
	& 2.3 & 453
 & $42.6\% $\\
	 	\hline
Example 22 in Table \ref{tab_mvg}  & ASGQ & $5.4e^{-03}$ &  31.5 & $2.5 \times 10^7$ 
	& 3.5 & 566
	& $11\% $ \\
	 	\hline
Example 31 in Table \ref{tab_mnig} & ASGQ & ${1.47e^{-02}}$ &  $14.2$ & $  10^7$  
	& 3.4 & 616
 & $24\% $\\
	 	\hline
Example 32 in Table \ref{tab_mnig} & TP & $3.75e^{-02}$ &  33.5 & $2.5 \times 10^7$ 
	& 11.7 & 4096
	& $35\% $  \\
	 	\hline
Example 11 in Table \ref{tab_mgbm}  & ASGQ & ${ 1.4e^{-03}}$  &  $2635$ & $6.9 \times10^8$  
	& 6 & 3070
	& $0.23\%$\\
	 	\hline
Example 12 in Table \ref{tab_mgbm} & ASGQ & $1.7e^{-03}$ &  2110 & $ 5.3 \times 10^8 $
	& 4.5 & 1642
 & $0.21\% $ \\
	 	\hline
Example 23 in Table \ref{tab_mvg} & ASGQ & $2 e^{-03}$ &  85 & $6.8 \times 10^7$ 
	& 19.5 & 7401
	& $23\%$\\
	 	\hline
Example 24 in Table \ref{tab_mvg} &	 ASGQ & ${2.6 e^{-03}}$  &  360  & $ 2.8 \times 10^8 $ 
	& 4.6 & 1671
	& $1.28 \%$ \\
	 	\hline
Example 33 in Table \ref{tab_mnig} &  ASGQ & ${5.7e^{-02}}$ &  $85.5$ & $6.3 \times 10^7 $  
	& 1 & 105
 & $1.17\% $\\
	 	\hline
Example 34 in Table \ref{tab_mnig} & ASGQ & $3.79e^{-02}$ &  108 & $7.5 \times 10^7 $ 
	& 1.4 & 340
	& $1.3\% $  \\
	\hline
	\end{tabular}
	\caption{Errors, CPU times in seconds, and number of quadrature points comparing the Fourier approach combined with the optimal damping rule and the best quadrature (Quad) method with the  Gauss--Laguerre rule against the MC method for the European basket and rainbow options under the multivariate GBM, VG, and NIG pricing dynamics for various  dimensions.   Tables \ref{tab_mgbm}, \ref{tab_mvg}, \ref{tab_mnig}  present the selected parameter sets for each pricing model,  the reference values with their corresponding statistical errors, and the optimal damping parameters.} 
	\label{cpu_tab}
\end{table}

%% file: appendix.tex
	\section{Extension of the  Error Bound \eqref{eq:error_bound_estimate} to the Multivariate Case}
	\label{quad_error_bound}
	The extension of the error bound to the multivariate case can be done by recursively applying the following reasoning. For illustration and notation convenience, we consider the 2D case. Let $[a,b] \subset \mathbb{R}$, and $\mathcal{C}_1, \mathcal{C}_2$ be closed contours of integration as defined in Theorem  \eqref{thm:remainder_integration_estimates}. We define the quantity of interest by
	\begin{small}
	\begin{equation}
		\label{integrand_value}
		I[f] : = \int_{a}^b \int_{a}^b \lambda(x_1) \lambda(x_2) f(x_1,x_2) dx_1 dx_2,
	\end{equation}
	\end{small}
	where $\lambda(\cdot)$ is the weight function corresponding to the Gaussian quadrature rule. The TP quadrature  of \eqref{integrand_value} with $N$ points in each dimension  is defined as
		\begin{small}
	\begin{equation}
		Q_{N}[f] : = \sum_{k_1 = 1 }^N \sum_{k_2 = 1}^N w_{k_1} w_{k_2} f(x_{k_1}, x_{k_2}),
	\end{equation}
		\end{small}
	and  the quadrature remainder is thus given by 
			\begin{small}
	\begin{equation}
		\mathcal{E}_{Q_{N}}[f] : = \left|  I[f] - Q_N[f]) \right|.
	\end{equation}
		\end{small}
	In the first step, for fixed $x_2 \in [a,b]$, applying  Theorem \ref{thm:remainder_integration_estimates} on $f(x_1,x_2)$ implies  
			\begin{small}
	\begin{equation}
		\int_{a}^b \lambda(x_1) f(x_1, x_2) dx_1 = \sum_{k_1 = 1}^N w_{k_1} f(x_{k_1}, x_2) + \frac{1}{2\pi i} \oint_{\mathcal{C}_1} K_N (z_1) f(z_1, x_2) dz_1
	\end{equation}
			\end{small}
	Plugging the above expression in \eqref{integrand_value}, we  obtain
		\begin{small}
	\begin{equation}
		\begin{aligned}
			I[f] &= \int_{a}^b  \lambda(x_2) \left[  \sum_{k_1 = 1}^N w_{k_1} f(x_{k_1}, x_2) + \frac{1}{2\pi i} \oint_{\mathcal{C}_1} K_N (z_1) f(z_1, x_2) dz_1 \right] dx_2 \\
			&= \sum_{k_1 = 1}^N w_{k_1} \int_{a}^b  \lambda(x_2) f(x_{k_1}, x_2) dx_2 + \frac{1}{2\pi i} \oint_{\mathcal{C}_1} K_N (z_1) \left[ \int_{a}^b  \lambda(x_2) f(z_1, x_2) dx_2 \right] dz_1.
		\end{aligned}
	\end{equation}
			\end{small}
	In a second stage, applying  Theorem \ref{thm:remainder_integration_estimates} for fixed $x_{k_1} \in [a,b]$ on  $f(x_{k_1},x_2)$, and for fixed $z_1 \in \mathcal{C}_1$ on $f(z_1,x_2)$, implies
			\begin{small}
	\begin{equation}
		\begin{aligned}
			I[f] &= \sum_{k_1 = 1 }^N \sum_{k_2 = 1}^N w_{k_1} w_{k_2} f(x_{k_1}, x_{k_2})	 +  \sum_{k_1 = 1}^N w_{k_1} \frac{1}{2 \pi i} \oint_{\mathcal{C}_2} K_N(z_2) f(x_{k_1},z_2) dz_2   \\ 
			&+ \sum_{k_2 = 1}^N w_{k_2} \frac{1}{2\pi i} \oint_{\mathcal{C}_1} K_N (z_1) f(z_1, x_{k_2}) dz_1 + \left(\frac{1}{2 \pi i}\right)^2	\oint_{\mathcal{C}_1} \oint_{\mathcal{C}_2} K_N (z_1) K_N(z_2) f(z_1, z_2) dz_1 dz_2
		\end{aligned}
	\end{equation}
				\end{small}
	Consequently, the quadrature error bound is given as
			\begin{small}
	\begin{equation}
		\begin{aligned}
			|	\mathcal{E}_{Q_N}[f] | &= | \sum_{k_1 = 1}^N w_{k_1} \frac{1}{2 \pi i} \oint_{\mathcal{C}_2} K_N(z_2) f(x_{k_1},z_2) dz_2  +  \sum_{k_2 = 1}^N w_{k_2} \frac{1}{2\pi i} \oint_{\mathcal{C}_1} K_N (z_1) f(z_1, x_{k_2}) dz_1  \\ &+ \left(\frac{1}{2 \pi i}\right)^2	\oint_{\mathcal{C}_1} \oint_{\mathcal{C}_2} K_N (z_1) K_N(z_2) f(z_1, z_2) dz_1 dz_2  | \\
			&\leq \sup_{ x_1 \in {\mathcal{C}_1}, x_2 \in {\mathcal{C}_2}} |f(x_1, x_2) | \left[ \sum_{k_1 = 1}^N w_{k_1} \frac{1}{2 \pi i} \oint_{\mathcal{C}_2} K_N(z_2) dz_2 + 
			\sum_{k_2 = 1}^N w_{k_2} \frac{1}{2 \pi i} \oint_{\mathcal{C}_1} K_N(z_1) dz_1 \right] \\ 
			&+  \sup_{ x_1 \in {\mathcal{C}_1}, x_2 \in {\mathcal{C}_2}} |f(x_1, x_2) |  \left[  \oint_{\mathcal{C}_1} \oint_{\mathcal{C}_2} K_N (z_1) K_N(z_2)  dz_1 dz_2 \right]
		\end{aligned}
	\end{equation}
				\end{small}
	Which  shows that the error bound depends on  $\sup_{ x_1 \in {\mathcal{C}_1}, x_2 \in {\mathcal{C}_2}} | f(x_1, x_2)  | $, similarly to \eqref{eq:error_bound_estimate}.

\section{On the Choice of the Quadrature Rule}\label{num_lag_herm_sec}
 In this section,  through numerical examples on vanilla put options, we show that the Gauss--Laguerre quadrature rule significantly outperforms the Gauss--Hermite quadrature rule  for the numerical evaluation of the inverse Fourier integrals; hence, we adopt the Gauss--Laguerre measure for the rest of the work.  Figures  \ref{1d_put_gbm_laguerre_vs_hermite}, \ref{1d_put_vg_laguerre_vs_hermite}, and \ref{1d_put_nig_laguerre_vs_hermite} reveal that  the Gauss--Laguerre quadrature rule significantly outcompetes the  Gauss--Hermite quadrature  independently of the values of the damping parameters in the strip of regularity for the tested models: GBM, VG, and NIG.  For instance,  Figure  \ref{1d_put_gbm_laguerre_vs_hermite} illustrates that, when $\text{R}=4$ is used, the Gauss--Laguerre quadrature rule reaches approximately the relative quadrature $ \mathcal{E}_{R} =0.01 \%$ using $12 \%$ of the work required by the Gauss--Hermite quadrature to attain the same accuracy.  These observations were consistent for all tested parameter constellations and dimensions,  and independent of the choice of the quadrature  methods (TP, ASGQ, or SM).  

\begin{figure}[h!]
	\centering 
	\begin{subfigure}{0.4\textwidth}
		\includegraphics[width=\linewidth]{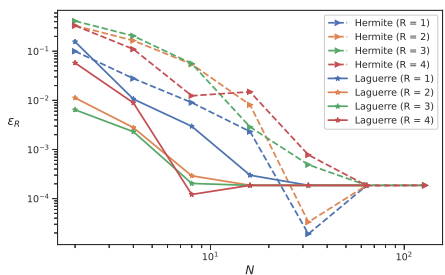}
		\caption{$\sigma = 0.4$}
		\label{1d_put_gbm_laguerre_vs_hermite}
	\end{subfigure} 
	\begin{subfigure}{0.4\textwidth}
		\includegraphics[width=\linewidth]{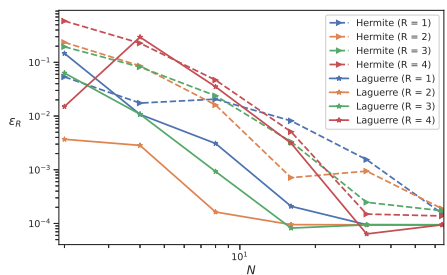}
		\caption{$\sigma = 0.4, \theta = -0.3,  \nu= 0.257$}
		\label{1d_put_vg_laguerre_vs_hermite}
	\end{subfigure}
	\begin{subfigure}{0.4\textwidth}
		\includegraphics[width=\linewidth]{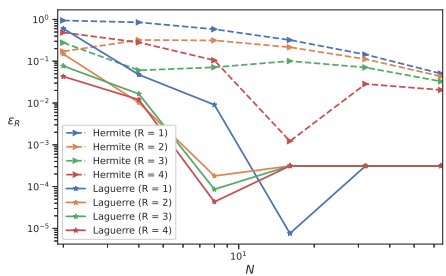}
		\caption{$\alpha = 10, \beta= -3, \delta= 0.2$}
		\label{1d_put_nig_laguerre_vs_hermite}
	\end{subfigure}
	\caption{Relative  quadrature error, $\mathcal{E}_{R}$, convergence w.r.t.~$N$ of Gauss--Laguerre and Gauss--Hermite quadrature rules for a European put option with $S_0=100$, $K=100$, $r=0$ ,  and $T=1$ under (a) GBM,  (b) VG, and (c) NIG.}
	\label{fig: convergence w.r.t  N of Gauss-Laguerre and Gauss-Hermite quadrature rules for European put option under (a) GBM  (b) VG (c) NIG.}
\end{figure}

%% file: main_risk_resub.bbl
\begin{thebibliography}{10}

\bibitem{aguilar2020some}
Jean-Philippe Aguilar.
\newblock Some pricing tools for the variance {G}amma model.
\newblock {\em International Journal of Theoretical and Applied Finance},
  23(04):2050025, 2020.

\bibitem{babuvska2007stochastic}
Ivo Babu{\v{s}}ka, Fabio Nobile, and Ra{\'u}l Tempone.
\newblock A stochastic collocation method for elliptic partial differential
  equations with random input data.
\newblock {\em SIAM Journal on Numerical Analysis}, 45(3):1005--1034, 2007.

\bibitem{barndorff1977exponentially}
Ole Barndorff-Nielsen.
\newblock Exponentially decreasing distributions for the logarithm of particle
  size.
\newblock {\em Proceedings of the Royal Society of London. A. Mathematical and
  Physical Sciences}, 353(1674):401--419, 1977.

\bibitem{barndorff1997normal}
Ole~E Barndorff-Nielsen.
\newblock {N}ormal inverse {G}aussian distributions and stochastic volatility
  modelling.
\newblock {\em Scandinavian Journal of statistics}, 24(1):1--13, 1997.

\bibitem{barndorff1997processes}
Ole~E Barndorff-Nielsen.
\newblock Processes of normal inverse {G}aussian type.
\newblock {\em Finance and {S}tochastics}, 2(1):41--68, 1997.

\bibitem{barthelmann2000high}
Volker Barthelmann, Erich Novak, and Klaus Ritter.
\newblock High dimensional polynomial interpolation on sparse grids.
\newblock {\em Advances in Computational Mathematics}, 12(4):273--288, 2000.

\bibitem{baschetti2021sinc}
Fabio Baschetti, Giacomo Bormetti, Silvia Romagnoli, and Pietro Rossi.
\newblock The sinc way: a fast and accurate approach to {F}ourier pricing.
\newblock {\em Quantitative Finance}, pages 1--20, 2021.

\bibitem{bayer2020hierarchical}
Christian Bayer, Chiheb Ben~Hammouda, and Ra{\'u}l Tempone.
\newblock Hierarchical adaptive sparse grids and {Q}uasi-{M}onte {C}arlo for
  option pricing under the rough {B}ergomi model.
\newblock {\em Quantitative Finance}, 20(9):1457--1473, 2020.

\bibitem{bayer2020multilevel}
Christian Bayer, Chiheb Ben~Hammouda, and Ra{\'u}l Tempone.
\newblock Multilevel {M}onte {C}arlo combined with numerical smoothing for
  robust and efficient option pricing and density estimation.
\newblock {\em arXiv preprint arXiv:2003.05708}, 2020.

\bibitem{bayer2021numerical}
Christian Bayer, Chiheb Ben~Hammouda, and Ra{\'u}l Tempone.
\newblock Numerical smoothing with hierarchical adaptive sparse grids and
  quasi-{M}onte {C}arlo methods for efficient option pricing.
\newblock {\em Quantitative Finance}, pages 1--19, 2023.

\bibitem{bayer2018smoothing}
Christian Bayer, Markus Siebenmorgen, and Ra{\'u}l Tempone.
\newblock Smoothing the payoff for efficient computation of basket option
  prices.
\newblock {\em Quantitative Finance}, 18(3):491--505, 2018.

\bibitem{ben2020hierarchical}
Chiheb Ben~Hammouda.
\newblock {\em Hierarchical Approximation Methods for Option Pricing and
  Stochastic Reaction Networks}.
\newblock PhD thesis, 2020.

\bibitem{byrd2000trust}
Richard~H Byrd, Jean~Charles Gilbert, and Jorge Nocedal.
\newblock A trust region method based on interior point techniques for
  nonlinear programming.
\newblock {\em Mathematical programming}, 89(1):149--185, 2000.

\bibitem{byrd1999interior}
Richard~H Byrd, Mary~E Hribar, and Jorge Nocedal.
\newblock An interior point algorithm for large-scale nonlinear programming.
\newblock {\em SIAM Journal on Optimization}, 9(4):877--900, 1999.

\bibitem{carlsson2017note}
Marcus Carlsson and Jens Wittsten.
\newblock A note on holomorphic functions and the fourier-laplace transform.
\newblock {\em Mathematica Scandinavica}, pages 225--248, 2017.

\bibitem{carr1999option}
Peter Carr and Dilip Madan.
\newblock Option valuation using the fast {F}ourier transform.
\newblock {\em Journal of Computational Finance}, 2(4):61--73, 1999.

\bibitem{CarrWu04}
Peter Carr and Liuren Wu.
\newblock {Time-changed L{\'e}vy processes and option pricing}.
\newblock {\em Journal of Financial Economics}, 71:113--141, 2004.

\bibitem{chau2019exploration}
Ki~Wai Chau and Cornelis~W Oosterlee.
\newblock Exploration of a cosine expansion lattice scheme.
\newblock {\em arXiv preprint arXiv:1907.02758}, 2019.

\bibitem{chen2018sparse}
Peng Chen.
\newblock Sparse quadrature for high-dimensional integration with {G}aussian
  measure.
\newblock {\em ESAIM: Mathematical Modelling and Numerical Analysis},
  52(2):631--657, 2018.

\bibitem{choi2018sum}
Jaehyuk Choi.
\newblock Sum of all {B}lack--{S}choles--{M}erton models: An efficient pricing
  method for spread, basket, and {A}sian options.
\newblock {\em Journal of Futures Markets}, 38(6):627--644, 2018.

\bibitem{colldeforns2017two}
Gemma Colldeforns-Papiol, Luis Ortiz-Gracia, and Cornelis~W Oosterlee.
\newblock Two-dimensional {S}hannon wavelet inverse {F}ourier technique for
  pricing {E}uropean options.
\newblock {\em Applied Numerical Mathematics}, 117:115--138, 2017.

\bibitem{tankov2003financial}
Rama Cont and Peter Tankov.
\newblock {\em Financial {M}odelling with {J}ump {P}rocesses}.
\newblock Chapman and Hall/CRC, 2003.

\bibitem{crisostomo2018speed}
Ricardo Cris{\'o}stomo.
\newblock Speed and biases of fourier-based pricing choices: a numerical
  analysis.
\newblock {\em International Journal of Computer Mathematics},
  95(8):1565--1582, 2018.

\bibitem{davis1954estimation}
Philip Davis and Philip Rabinowitz.
\newblock On the estimation of quadrature errors for analytic functions.
\newblock {\em Mathematical Tables and Other Aids to Computation},
  8(48):193--203, 1954.

\bibitem{davis2007methods}
Philip~J. Davis and Philip Rabinowitz.
\newblock {\em Methods of Numerical Integration}.
\newblock Courier Corporation, 2007.

\bibitem{donaldson1973estimates}
J.D. Donaldson.
\newblock Estimates of upper bounds for quadrature errors.
\newblock {\em SIAM Journal on Numerical Analysis}, 10(1):13--22, 1973.

\bibitem{donaldson1972unified}
J.D. Donaldson and David Elliott.
\newblock A unified approach to quadrature rules with asymptotic estimates of
  their remainders.
\newblock {\em SIAM Journal on Numerical Analysis}, 9(4):573--602, 1972.

\bibitem{duffie2003affine}
Darrell Duffie, Damir Filipovi{\'c}, and Walter Schachermayer.
\newblock Affine processes and applications in finance.
\newblock {\em The Annals of Applied Probability}, 13(3):984--1053, 2003.

\bibitem{eberlein2010analysis}
Ernst Eberlein, Kathrin Glau, and Antonis Papapantoleon.
\newblock Analysis of {F}ourier transform valuation formulas and applications.
\newblock {\em Applied Mathematical Finance}, 17(3):211--240, 2010.

\bibitem{elliott1970uniform}
David Elliott.
\newblock Uniform asymptotic expansions of the classical orthogonal polynomials
  and some associated functions.
\newblock In {\em Technical Report No. 21}. Mathematics Department, The
  University of Tasmania Hobart, Tasmania, 1970.

\bibitem{elliott1974asymptotic}
David Elliott and P.D. Tuan.
\newblock Asymptotic estimates of {F}ourier coefficients.
\newblock {\em SIAM Journal on Mathematical Analysis}, 5(1):1--10, 1974.

\bibitem{fang2008novel}
Fang Fang and Cornelis.~W. Oosterlee.
\newblock A novel pricing method for {E}uropean options based on
  {F}ourier-cosine series expansions.
\newblock {\em SIAM Journal on Scientific Computing}, 31(2):826--848, 2008.

\bibitem{gautschi2004orthogonal}
Walter Gautschi.
\newblock {\em Orthogonal {P}olynomials: {C}omputation and {A}pproximation}.
\newblock OUP Oxford, 2004.

\bibitem{gerstner1998numerical}
Thomas Gerstner and Michael Griebel.
\newblock Numerical integration using sparse grids.
\newblock {\em Numerical Algorithms}, 18(3):209--232, 1998.

\bibitem{Gerstner2003DimensionAdaptiveTQ}
Thomas Gerstner and Michael Griebel.
\newblock Dimension–adaptive tensor–product quadrature.
\newblock {\em Computing}, 71:65--87, 2003.

\bibitem{glasserman2004monte}
Paul Glasserman.
\newblock {\em Monte Carlo {M}ethods in {F}inancial {E}ngineering}, volume~53.
\newblock Springer, 2004.

\bibitem{guillaume2008making}
Tristan Guillaume.
\newblock Making the best of best-of.
\newblock {\em Review of Derivatives Research}, 11(1):1--39, 2008.

\bibitem{healy2021pricing}
Jherek Healy.
\newblock The pricing of vanilla options with cash dividends as a classic
  vanilla basket option problem.
\newblock {\em arXiv preprint arXiv:2106.12971}, 2021.

\bibitem{hormander2015analysis}
Lars H{\"o}rmander.
\newblock {\em The analysis of linear partial differential operators I:
  Distribution theory and Fourier analysis}.
\newblock Springer, 2015.

\bibitem{hubalek2005variance}
Friedrich Hubalek and Jan Kallsen.
\newblock Variance-optimal hedging and {M}arkowitz-efficient portfolios for
  multivariate processes with stationary independent increments with and
  without constraints.
\newblock Technical report, Working paper, TU M{\"u}nchen, 2005.

\bibitem{hurd2010fourier}
Thomas~R Hurd and Zhuowei Zhou.
\newblock A {F}ourier transform method for spread option pricing.
\newblock {\em SIAM Journal on Financial Mathematics}, 1(1):142--157, 2010.

\bibitem{junike2022precise}
Gero Junike and Konstantin Pankrashkin.
\newblock Precise option pricing by the {COS} method—how to choose the
  truncation range.
\newblock {\em Applied Mathematics and Computation}, 421:126935, 2022.

\bibitem{junike2023multidimensional}
Gero Junike and Hauke Stier.
\newblock The multidimensional cos method for option pricing.
\newblock {\em arXiv preprint arXiv:2307.12843}, 2023.

\bibitem{kahl2010fourier}
Christian Kahl and Roger Lord.
\newblock {F}ourier inversion methods in finance.
\newblock {\em Handbook of Computational Finance}, 2010.

\bibitem{kirkby2015efficient}
J.~Lars. Kirkby.
\newblock Efficient option pricing by frame duality with the fast {F}ourier
  transform.
\newblock {\em SIAM Journal on Financial Mathematics}, 6(1):713--747, 2015.

\bibitem{kirkby2020general}
J~Lars Kirkby, Dang~H Nguyen, and Duy Nguyen.
\newblock A general continuous time {M}arkov chain approximation for
  multi-asset option pricing with systems of correlated diffusions.
\newblock {\em Applied Mathematics and Computation}, 386:125472, 2020.

\bibitem{kwok2012efficient}
Yue~Kuen Kwok, Kwai~Sun Leung, and Hoi~Ying Wong.
\newblock Efficient options pricing using the fast {F}ourier transform.
\newblock In {\em Handbook of Computational Finance}, pages 579--604. Springer,
  2012.

\bibitem{Lee04}
Roger.~W. Lee.
\newblock {Option pricing by transform methods: extensions, unification, and
  error control}.
\newblock {\em Journal of Computational Finance}, 7(3):50--86, 2004.

\bibitem{leentvaar2008multi}
C.C.W. Leentvaar and Cornelis.~W. Oosterlee.
\newblock Multi-asset option pricing using a parallel {F}ourier-based
  technique.
\newblock {\em Journal of Computational Finance}, 12(1):1, 2008.

\bibitem{lewis2001simple}
Alan~L. Lewis.
\newblock A simple option formula for general jump-diffusion and other
  exponential l{\'e}vy processes.
\newblock {\em Available at SSRN 282110}, 2001.

\bibitem{lord2008fast}
Roger Lord, Fang Fang, Frank Bervoets, and Cornelis.~W. Oosterlee.
\newblock A fast and accurate {FFT}-based method for pricing early-exercise
  options under l{\'e}vy processes.
\newblock {\em SIAM Journal on Scientific Computing}, 30(4):1678--1705, 2008.

\bibitem{lord2007optimal}
Roger Lord and Christian Kahl.
\newblock Optimal {F}ourier inversion in semi-analytical option pricing.
\newblock 2007.

\bibitem{luciano2006multivariate}
Elisa Luciano and Wim Schoutens.
\newblock A multivariate jump-driven financial asset model.
\newblock {\em Quantitative finance}, 6(5):385--402, 2006.

\bibitem{lukacs1970characteristic}
E~Lukacs.
\newblock Characteristic functions, {C}harles {G}riffin \& co.
\newblock {\em Ltd.(London, 1960)}, 18(62):134, 1970.

\bibitem{Madan1990TheVG}
Dilip~B. Madan and Eugene Seneta.
\newblock The variance {G}amma ({V}.{G}.) model for share market returns.
\newblock {\em The Journal of Business}, 63:511--524, 1990.

\bibitem{margrabe1978value}
William Margrabe.
\newblock The value of an option to exchange one asset for another.
\newblock {\em The journal of finance}, 33(1):177--186, 1978.

\bibitem{mijatovic2013continuously}
Aleksandar Mijatovi{\'c} and Martijn Pistorius.
\newblock Continuously monitored barrier options under {M}arkov processes.
\newblock {\em Mathematical Finance: An International Journal of Mathematics,
  Statistics and Financial Economics}, 23(1):1--38, 2013.

\bibitem{oosterlee2019mathematical}
Cornelis~W Oosterlee and Lech~A Grzelak.
\newblock {\em Mathematical modeling and computation in finance: with exercises
  and Python and MATLAB computer codes}.
\newblock World Scientific, 2019.

\bibitem{piazzola.tamellini:SGK}
C.~Piazzola and L.~Tamellini.
\newblock {The Sparse Grids Matlab kit - a Matlab implementation of sparse
  grids for high-dimensional function approximation and uncertainty
  quantification}.
\newblock {\em ArXiv}, (2203.09314), 2022.

\bibitem{raible2000levy}
Sebastian Raible.
\newblock {\em L{\'e}vy processes in finance: Theory, numerics, and empirical
  facts}.
\newblock PhD thesis, Universit{\"a}t Freiburg i. Br, 2000.

\bibitem{ruijter2012two}
Marjon~J Ruijter and Cornelis.~W. Oosterlee.
\newblock Two-dimensional {F}ourier cosine series expansion method for pricing
  financial options.
\newblock {\em SIAM Journal on Scientific Computing}, 34(5):B642--B671, 2012.

\bibitem{schoutens2003levy}
Wim Schoutens.
\newblock {\em {L}{\'e}vy processes in finance: pricing financial derivatives}.
\newblock Wiley Online Library, 2003.

\bibitem{smolyak1963quadrature}
Sergei~Abramovich Smolyak.
\newblock Quadrature and interpolation formulas for tensor products of certain
  classes of functions.
\newblock In {\em Doklady Akademii Nauk}, volume 148, pages 1042--1045. Russian
  Academy of Sciences, 1963.

\bibitem{takahasi1971estimation}
Hidetosi Takahasi and Masatake Mori.
\newblock Estimation of errors in the numerical quadrature of analytic
  functions.
\newblock {\em Applicable Analysis}, 1(3):201--229, 1971.

\bibitem{titchmarsh1948introduction}
Edward~Charles Titchmarsh et~al.
\newblock {\em Introduction to the {T}heory of {F}ourier {I}ntegrals},
  volume~2.
\newblock Clarendon Press Oxford, 1948.

\bibitem{trefethen2008gauss}
Lloyd~N Trefethen.
\newblock Is {G}auss quadrature better than {C}lenshaw--{C}urtis?
\newblock {\em SIAM Review}, 50(1):67--87, 2008.

\bibitem{Wiel2015ValuationOI}
Dennis van~de Wiel.
\newblock Valuation of insurance products using a normal inverse {G}aussian
  distribution.
\newblock Bachelor's thesis, 2015.

\bibitem{von2015benchop}
Lina von Sydow, Lars Josef~H{\"o}{\"o}k, Elisabeth Larsson, Erik Lindstr{\"o}m,
  Slobodan Milovanovi{\'c}, Jonas Persson, Victor Shcherbakov, Yuri
  Shpolyanskiy, Samuel Sir{\'e}n, Jari Toivanen, et~al.
\newblock Benchop--the benchmarking project in option pricing.
\newblock {\em International Journal of Computer Mathematics},
  92(12):2361--2379, 2015.

\bibitem{WASILKOWSKI19951}
Grzegorz~W. Wasilkowski and Henryk Wozniakowski.
\newblock Explicit cost bounds of algorithms for multivariate tensor product
  problems.
\newblock {\em Journal of Complexity}, 11(1):1--56, 1995.

\bibitem{xi2019simultaneous}
Yuejuan Xi, Kailin Ding, and Ning Ning.
\newblock Simultaneous two-dimensional continuous-time {M}arkov chain
  approximation of two-dimensional fully coupled {M}arkov diffusion processes.
\newblock {\em Available at SSRN 3461115}, 2019.

\bibitem{xiang2012asymptotics}
Shuhuang Xiang.
\newblock Asymptotics on {L}aguerre or {H}ermite polynomial expansions and
  their applications in {G}auss quadrature.
\newblock {\em Journal of Mathematical Analysis and Applications},
  393(2):434--444, 2012.

\bibitem{zhang2013efficient}
Bowen Zhang and Cornelis~W. Oosterlee.
\newblock Efficient pricing of {E}uropean-style {A}sian options under
  exponential {L}{\'e}vy processes based on {F}ourier cosine expansions.
\newblock {\em SIAM Journal on Financial Mathematics}, 4(1):399--426, 2013.

\bibitem{zhu2009applications}
Jianwei Zhu.
\newblock {\em Applications of {F}ourier {T}ransform to {S}mile {M}odeling:
  Theory and {I}mplementation}.
\newblock Springer Science \& Business Media, 2009.

\end{thebibliography}
